\documentclass[11pt]{article}

\usepackage{booktabs} 
\usepackage[ruled]{algorithm2e} 

\SetAlFnt{\small}
\SetAlCapFnt{\small}
\SetAlCapNameFnt{\small}
\SetAlCapHSkip{0pt}
\IncMargin{-\parindent}

 \usepackage[margin=1in]{geometry}

\usepackage{bbding}
\usepackage[nointegrals]{wasysym}
\usepackage[utf8]{inputenc}
\usepackage{xspace}
\usepackage{mathtools}
\usepackage{url}
\usepackage{multirow,multicol}
\usepackage{subcaption}
\usepackage{tikz}
\usepackage{hyperref}
\usetikzlibrary{arrows, fit, positioning}
\usetikzlibrary{backgrounds,automata,calc}
\usetikzlibrary{decorations.pathreplacing}
\usepackage{csquotes}
\usepackage{amsmath,amsfonts,amsthm}
\allowdisplaybreaks
\usepackage{bm,bbm}
\usepackage{pgfplots}
\pgfplotsset{compat=1.17}

\usepackage{enumerate}
\usepackage{paralist}
\usepackage[shortlabels]{enumitem}
\usepackage{appendix}
\usepackage[capitalize]{cleveref}
\usepackage{thmtools}
\usepackage{mathtools}
\usepackage{natbib}
\usepackage{changepage}
 \usepackage{authblk}
\usepackage{commath}
\usepackage{makecell}

\newcommand{\NSW}{\normalfont{\texttt{NSW}}\xspace}
\newcommand{\logNSW}{\normalfont{\texttt{logNSW}}\xspace}

\newcommand{\vare}{\varepsilon}

\newcommand{\R}{\mathbb{R}}

\DeclareMathOperator*{\E}{\mathbb{E}}

\renewcommand{\cal}[1]{\mathcal{#1}}

\newcommand{\x}{\mathbf{x}}

\newcommand{\z}{\mathbf{z}}
\newcommand{\h}{\mathbf{h}}

\newcommand{\ST}{\texttt{ST}}

\usepackage{mathtools}

\setlist{topsep=0.5ex,itemsep=0.1ex}

\newtheorem{theorem}{Theorem}[section]

\newtheorem{lemma}[theorem]{Lemma}
\newtheorem{ass}[theorem]{Assumption}

\newtheorem{claim}[theorem]{Claim}

\newtheorem{corr}[theorem]{Corollary}
\newtheorem{fact}[theorem]{Fact}
\newtheorem{obs}[theorem]{Observation}
\newtheorem*{theorem*}{Theorem}

\theoremstyle{definition}
\newtheorem{remark}[theorem]{Remark}
\newtheorem{definition}[theorem]{Definition}

\makeatletter
\newcommand{\descitem}[2]{\item[#1]\def\@currentlabel{#1}\label{#2}}
\makeatother

\title{A Better-than-$e^{1/e}$ Approximation Algorithm for Nash Social Welfare under Additive Valuations}

\author{Vignesh Viswanathan}
\affil{University of Massachusetts, Amherst \\ \texttt{vviswanathan@umass.edu}}
\date{}

\begin{document}

\maketitle

\begin{abstract}
We present an $(e^{1/e} - c)$-approximation algorithm for maximizing Nash social welfare under additive valuations, for some constant $c > 0$. This result improves upon the previous best-known approximation factor of $e^{1/e}$ \cite{Barman2018FindingFA}.
\end{abstract}

\thispagestyle{empty}
\newpage

\tableofcontents
\thispagestyle{empty}
\newpage

\section{Introduction}
We study the problem of maximizing Nash social welfare when dividing a set of indivisible goods $G$ among a set of $n$ agents $N$. Each agent $i \in N$ is associated with an additive valuation function $v_i: 2^G \rightarrow \R_{\ge 0}$. 
The goal of the problem is to find an allocation of the goods $X = (X_1, \dots, X_n)$ that maximizes the Nash social welfare --- which is defined as the geometric mean of agent utilities $\left (\prod_{i \in N} v_i(X_i) \right )^{\frac{1}{n}}$.

Designing algorithms to compute fair and efficient allocations of indivisible goods is a fundamental problem in theoretical computer science. There are many different objectives one can choose to optimize for when computing allocations. Utilitarian social welfare maximizing allocations\footnote{The utilitarian social welfare of an allocation $X$ is defined as the sum of agent utilities $\sum_{i \in N}v_i(X_i)$.}, for example, are guaranteed to be efficient but may be unfair to some agents. Egalitarian welfare maximizing allocations\footnote{The egalitarian welfare of an allocation $X$ is defined as the value $\min_{i \in N}v_i(X_i)$.}, on the other hand, are guaranteed to be fair but may not be very efficient. Nash social welfare (or simply, Nash welfare) is an objective that perfectly balances fairness and efficiency. To quote the Kalai prize winning work of \cite{Caragiannis2019MNW}, Nash welfare is 
\begin{center}
\textit{``\dots arguably the ultimate solution --- for the division of indivisible goods under additive valuations.''}
\end{center}

Naturally, there has been significant interest in finding efficient approximation algorithms for maximizing Nash welfare. Despite this, the problem remains largely open, with a large gap between the best approximation algorithm and the best inapproximability result.
The previous best approximation algorithm for Nash social welfare achieves an approximation ratio of $e^{1/e}$ \cite{Barman2018FindingFA} using a market equilibrium based approach. Prior to their work, the main technique used was rounding a spending restricted market equilibrium \cite{cole2015nash,Cole2017nashwelfare} to achieve a constant factor approximation. In contrast, the best known hardness results prove that the max Nash welfare is inapproximable by a factor of 1.069 \cite{garg2018budgetadditive} unless P = NP and 1.076 unless the unique games conjecture is false \cite{Viswanathan2026UGCAdditive}. 

In this paper, we prove the following result, making the first improvement in approximation factor in over eight years.

\begin{restatable}{theorem}{thmnash}\label{thm:nash}
There exists an efficient randomized $(e^{1/e} - c)$-approximation algorithm for maximizing Nash social welfare when agents have additive valuations, for some constant $c > 0$.
\end{restatable}

Our algorithm rounds the configuration LP for this problem. It builds on recent work \cite{Feng2025weightednashadditive} which uses the configuration LP to design an $e^{1/e}$-approximation algorithm for maximizing {\em weighted} Nash welfare. 

Aside from improving the approximation factor, our result has two interesting implications for the problem of maximizing Nash welfare. First, it shows that the configuration LP has a strictly better integrality gap than the spending restricted market equilibrium, which was shown to have an integrality gap of $e^{1/e}$ \cite{Cole2017nashwelfare}. Second, it proves a separation between weighted Nash welfare and unweighted Nash welfare with respect to the configuration LP. Specifically, in recent work, \cite{Bei2026MNW} show that the configuration LP has an integrality gap of $e^{1/e}$ for weighted Nash welfare. Our result implies that the configuration LP has integrality gap strictly better than $e^{1/e}$ for unweighted Nash welfare. 

\subsection{Overview of Proof}
As mentioned above, our algorithm rounds the solution of the configuration LP. The configuration LP for the max Nash welfare problem is defined as follows:
\begin{alignat}{3}
  &\max  \quad && \frac1n \sum_{i \in N} \sum_{S \subseteq G} y_{i, S} \ln{v_i(S)} \notag \\
  &\text{s.t.} \quad
               && \sum_{S \subseteq G} y_{i, S} = 1,
                   &&\quad \forall\, i \in N, \notag \\
  &            && \sum_{i \in N}\; \sum_{\substack{S \subseteq G: g \in S}} y_{i, S} = 1,
                   &&\quad \forall\, g \in G, \notag \\
  &            && y_{i,S} \geq 0,
                   &&\quad \forall\, i \in N,\ S \subseteq G. \notag
\end{alignat}
The variable $y_{i, S}$ can be interpreted as an indicator variable for whether agent $i$ is allocated the bundle $S$ in the optimal allocation.  The configuration LP solution $y^*$ implies a natural fractional allocation $\x$\footnote{A fractional allocation $\x$ is an $n \times m$ matrix where each row represents an agent and each column represents a good. The value $x_{ig}$ denotes the amount of good $g$ allocated to agent $i$. Naturally, $x_{ig} \in [0, 1]$ and $\sum_{i \in N} x_{ig} \le 1$ for each $g\in G$.}, where $x_{ig} = \sum_{S: g\in S} y^*_{i, S}$ for each agent $i$ and good $g$. In our algorithm, we either round the fractional allocation $\x$ or a modified version of $\x$ using the Shmoys-Tardos rounding method \cite{ShmoysTardos1993GAP}. Let $X$ denote the final integral allocation that the rounding method outputs. 

Our goal is to show that 
\begin{align}
\frac1n \sum_{i \in N} \E[\ln v_i(X_i)] \ge \frac1n \sum_{i \in N, S \subseteq G} y^*_{i, S}\ln{v_i(S)} - \frac1e + c^*, \label{eq:goal}
\end{align}
for some constant $c^*$. It is easy to see that this implies $X$ (in expectation) is strictly better than an $e^{1/e}$-approximation of the max Nash welfare. 

\cite{Feng2025weightednashadditive} show that if we round the fractional allocation $\x$ implied by the configuration LP solution $y^*$ using the Shmoys-Tardos rounding method, we output an integral allocation $X$ that satisfies, for each agent $i \in N$,
\begin{align*}
\E[\ln v_i(X_i)] \ge \sum_{S \subseteq G} y_{i, S}\ln{v_i(S)} - \frac1e.
\end{align*}

A close examination of their proof shows that equality holds in the above inequality only when the fractional bundle $\x_i$ has a very specific structure. That is, for equality to hold, the agent $i$ must have a set of {\em large} goods $L_i$ that they value very highly, and all the remaining goods must have value significantly less than the value of the large goods. We label all the non-large goods as {\em small}. Additionally, the total fractional amount of large goods in the bundle $\x_i$ must be equal to $1-1/e$. This is illustrated in Figure \ref{fig:well-behaved}.  

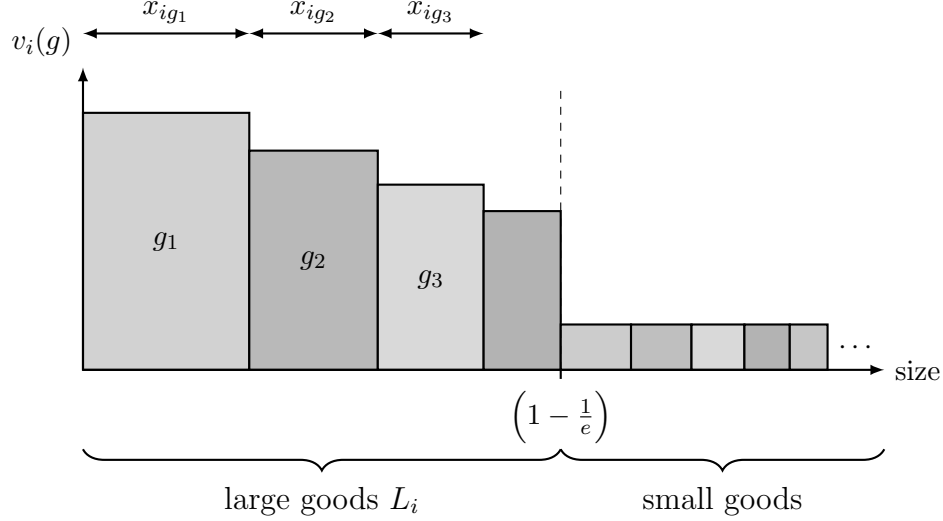
\begin{figure}[ht]
\centering
\begin{tikzpicture}

\def\Le{6.32}      
\def\xend{10.6}    

\draw[thick, fill=gray!35] (0,0)    rectangle (2.2,3.4);   
\draw[thick, fill=gray!55] (2.2,0)  rectangle (3.9,2.9);   
\draw[thick, fill=gray!30] (3.9,0)  rectangle (5.3,2.45);  
\draw[thick, fill=gray!60] (5.3,0)  rectangle (\Le,2.1);   

\draw[thick, fill=gray!40] (\Le,0)  rectangle (7.25,0.6);
\draw[thick, fill=gray!50] (7.25,0) rectangle (8.05,0.6);
\draw[thick, fill=gray!30] (8.05,0) rectangle (8.75,0.6);
\draw[thick, fill=gray!60] (8.75,0) rectangle (9.35,0.6);
\draw[thick, fill=gray!45] (9.35,0) rectangle (9.85,0.6);
\node at (10.25,0.3) {$\cdots$};

\node at (1.1, 1.7) {$g_1$};
\node at (3.05,1.45){$g_2$};
\node at (4.6, 1.2) {$g_3$};

\draw[->,>=latex,thick] (0,0) -- (0,4.0) node[above left]{$v_i(g)$};
\draw[->,>=latex,thick] (0,0) -- (\xend,0) node[right]{size};

\draw[dashed] (\Le,0) -- (\Le,3.7);

\draw[thick] (\Le,0) -- (\Le,-0.12);
\node[anchor=north] at (\Le,-0.14) {$\left(1-\frac1e\right)$};

\draw[<->,>=latex,thick] (0,4.45)   -- (2.2,4.45) node[midway,above]{$x_{ig_1}$};
\draw[<->,>=latex,thick] (2.2,4.45) -- (3.9,4.45) node[midway,above]{$x_{ig_2}$};
\draw[<->,>=latex,thick] (3.9,4.45) -- (5.3,4.45) node[midway,above]{$x_{ig_3}$};

\draw[decorate,decoration={brace,mirror,amplitude=8pt},thick]
      (0,-1.05) -- (\Le,-1.05) node[midway,below=10pt]{\large large goods $L_i$};
\draw[decorate,decoration={brace,mirror,amplitude=8pt},thick]
      (\Le,-1.05) -- (\xend,-1.05) node[midway,below=10pt]{\large small goods};

\end{tikzpicture}
\caption{Structure of the fractional bundle $\x_i$ of a well-behaved agent $i$. Each good $g$ is drawn as a bar whose width is the fractional amount $x_{ig}$ and whose height is the value $v_i(g)$ (assuming $v_i(g_1) \ge v_i(g_2) \ge \dots \ge v_i(g_m)$). The first few goods are large goods, valued significantly higher than the small goods, and the total fractional amount of large goods is exactly $1 - 1/e$.}
\label{fig:well-behaved}
\end{figure}

If a fractional bundle $\x_i$ satisfies this structure, we call the agent $i$ {\em well-behaved}. We show that if an agent $i$ is far from well-behaved, then applying Shmoys-Tardos rounding on the fractional allocation $\x$ results in an integral allocation $X$ that satisfies:
\begin{align*}
\E[\ln v_i(X_i)] \ge \sum_{S \subseteq G} y_{i, S}\ln{v_i(S)} - \frac1e + c^*_1, 
\end{align*}
for some constant $c^*_1 > 0$.
Therefore, if a constant fraction of agents are far from being well-behaved, we achieve our goal of \eqref{eq:goal} by rounding the fractional allocation $\x$ implied by the configuration LP solution $y^*$.

It therefore suffices to consider the case where almost all agents are well-behaved. In this case, we closely analyze the large goods of each well-behaved agent $i$. If it is the case that a good $g$ is large according to agent $i$ and small according to agent $j$, then moving a fraction of $g$ from $j$'s fractional bundle $\x_j$ to $i$'s fractional bundle $\x_i$ improves the Nash welfare of the final allocation. Intuitively, this is because a small good has very small value and therefore removing it results in a small drop in $\E[\ln v_j(X_j)]$ but adding a fraction of a large good to $\x_i$ significantly increases $\E[\ln v_i(X_i)]$. If we are able to do this with a constant fraction of the goods, then we achieve \eqref{eq:goal}.

Therefore, we may assume that most goods are either large for most agents or small for most agents. We call this the {\em large good consistency property}. For the case where most agents are well-behaved and the large good consistency property is satisfied, we present our main algorithm. This is arguably the most technical part of the paper. 

Our main algorithm is a two phase algorithm, where in the first phase we allocate large goods, and in the second phase, we allocate the small goods. To gain intuition about how such a two phase procedure can be implemented to improve the Nash welfare guarantees, we first describe a similar procedure for the case where {\em all} agents are well-behaved and {\em all} goods are either large for all the agents or small for all the agents. 

Since large goods have significantly higher value than small goods, we first allocate the large goods using dependent rounding \cite{gandhi2006dependent} such that the probability that each agent $i$ receives large good $g$ is exactly $x_{ig}$. This is Phase 1. If an agent $i$ receives a large good $g$ in Phase 1, then the small goods do not add much value to the agent's bundle. Therefore, we modify the fractional allocation $\x_i$ by removing all fractions of small goods in the bundle $\x_i$ and keeping only the large good $g$. Formally, we update $\x_i$ by setting $x_{ig} = 1$ and $x_{ig'} = 0$ for all $g' \ne g$. We re-allocate these small goods (which were removed from $\x_i$) to agents who do not receive large goods. That is, if we remove $\delta_g$ amount of small good $g$,  we update the fractional allocation $\x$ using $x_{jg} \gets x_{jg}(1 + \delta_g)$ for all agents $j$ who do not receive a large good in Phase 1\footnote{We note that the total amount of small goods that get removed does not get re-allocated by this update rule. We present this update rule anyway since it is simple and very similar to the update rule we end up using in our main algorithm.}. We then round the modified allocation $\x$ using the Shmoys-Tardos rounding procedure. This rounding does not change the allocation of the large goods from what was decided in Phase 1.  Using the negative correlation properties of dependent rounding, we can show that if an agent $i$ does not receive a large good, its fractional bundle significantly improves during the small good re-allocation procedure. Using this idea, we can show that (in expectation) \eqref{eq:goal} is satisfied. 

The above approach works in an idealized setting where all agents are well-behaved and we can perfectly classify each good as large or small. However, this may not be the case in the instances we work with. Therefore, we cannot perfectly round large goods since some of them may be allocated as small goods to a few agents, or may be allocated to non-well-behaved agents. To deal with this, we design a modification to the dependent rounding procedure specifically for our problem. We call it {\em dependent partial rounding}. The key idea of dependent partial rounding is to apply dependent rounding to the best of our abilities till we cannot anymore. It may not be possible to perfectly round even a single large good, but dependent partial rounding will modify the fractional allocation of the large goods such that most agents get a fractional amount of large goods either close to $1$ or close to $0$. What we end up with is a new fractional allocation $\z$ which differs from $\x$ only in the allocation of the large goods. We then employ a small good re-distribution scheme where for each agent $i$, the amount of small goods removed from $\z_i$ is proportional to the amount of large goods it receives from the dependent partial rounding procedure. We give these removed fractions of small goods to agents who have a small amount of large goods after the dependent partial rounding procedure. Using an analogous negative correlation property, we show that applying Shmoys-Tardos rounding to the final fractional allocation $\z$ results in an integral allocation that satisfies \eqref{eq:goal}.

\noindent \textbf{Lower Bounding $\E[\ln v_i(X_i)]$}. At several points in the above description, we have made claims like `removing a small fraction of small good from $\x_i$ reduces $\E[\ln v_i(X_i)]$ by a small amount' or `adding fractions of goods to $\x_i$ increases $\E[\ln v_i(X_i)]$'. It is not immediately clear how to quantify this change, or how to analyze $\E[\ln v_i(X_i)]$ in general. \cite{Feng2025weightednashadditive} present an elegant solution for this issue. They construct an instance $\cal I^{i, \x}(\Delta)$ with $\Delta$ agents and $\Delta x_{ig}$ copies of each good $g$, where $\Delta$ is suitably chosen such that $\Delta x_{ig}$ is an integer for all $g \in G$. They show that, assuming all $\Delta$ agents have the valuation function $v_i$, $\E[\ln v_i(X_i)]$ corresponds to the log of the Nash welfare of a specific envy-free up to one good allocation $W$ in the instance $\cal I^{i, \x}(\Delta)$. We build on this result by showing that $W$ can be assumed to be the round robin allocation of the instance.

Using this trick, any change to the fractional bundle $\x_i$ can be interpreted as a modification to the instance $\cal I^{i, \x}(\Delta)$. Therefore, if we want to quantify the effect of removing goods from or adding goods to $\x_i$, it suffices to bound the change to the Nash welfare of the round robin allocation when goods are added or removed from an instance. Round robin allocations, especially when agents have identical valuations, have a structure that is intuitive and easy to analyze. This is the key technique we use to lower bound $\E[\ln v_i(X_i)]$.

\subsection{Other Related Work}
Our work is heavily inspired by the two papers \cite{Kalaitzis2016ImprovedBudgetedAllocation} and \cite{Kalaitzis2015ConfigurationLPBudgetedAllocation}, which together present an improved rounding method for the configuration LP of the budgeted allocation problem. Specifically, our idea to use dependent rounding is inspired by the use of dependent rounding in \cite{Kalaitzis2015ConfigurationLPBudgetedAllocation}. We note that budgeted allocation is, in some sense, easier than Nash welfare, particularly when it comes to applying dependent rounding. In \cite{Kalaitzis2015ConfigurationLPBudgetedAllocation}, the algorithm is able to perfectly round all the large goods essentially by removing non-well-behaved agents and inconsistent large goods from the problem instance. This cannot be done with Nash welfare since it could lead to an unbounded decrease in the Nash welfare of the final allocation --- setting one agent's utility to $0$ sets the entire Nash welfare to $0$. In a similar vein, \cite{Kalaitzis2015ConfigurationLPBudgetedAllocation} are able to analyze the re-allocation of small goods after dependent rounding using the assignment LP which, in the case of budgeted allocation, has an integrality gap of $4/3$. This again cannot be done with Nash welfare since the assignment LP has an unbounded integrality gap. Both issues require novel ideas to solve, which we provide in the design of the dependent partial rounding method, and a connection between the output of the Shmoys-Tardos rounding algorithm and round robin allocations. 

Negatively correlated rounding has also been used in other fair allocation-adjacent problems like minimizing weighted completion time in job scheduling \cite{Harris2025Scheduling,Bansal2021Scheduling} and submodular welfare maximization with a demand oracle \cite{Feige2006GAPImprovement}. Our paper adds to this rich body of literature showcasing the power of dependent rounding in algorithm design.

Finally, our work adds to the ever growing body of literature on approximation algorithms for Nash welfare. There have been several breakthrough results on this problem in the last few years. Notably, we now have a constant factor approximation algorithm for unweighted \cite{li2021mnw,garg2023mnw} and weighted Nash welfare \cite{Feng2025WeightedNashSubmodular,Bei2026MNW} under submodular valuations. There is also a constant factor approximation algorithm for unweighted Nash welfare under subadditive valuations (assuming access to a demand oracle) \cite{dobzinski2024subadditive}.

\subsection{Organization}
Section \ref{sec:prelims} establishes the preliminaries. Sections \ref{sec:round-robin-guarantees} and \ref{sec:add-remove} analyze the Nash welfare guarantees of the round robin allocation when agents have identical valuations. Section \ref{sec:suboptimal} presents an important property of Nash optimal allocations in instances with identical valuations. Section \ref{sec:st-rounding} shows that we achieve our goal when a constant fraction of agents are not well-behaved. Section \ref{sec:large-good-consistency} formally defines the large good consistency property and shows that we achieve our goal if it is not satisfied by an instance. Section \ref{sec:main-algo} presents our main algorithm for the problem and analyzes its correctness. Section \ref{sec:main-theorem} uses all the results from the previous sections to prove Theorem \ref{thm:nash}.

\section{Preliminaries}\label{sec:prelims}
We use $[k]$ to denote the set $\{1, 2, \dots, k\}$. A fair allocation instance $\cal I$ is defined by a set of $n$ agents $N = [n]$ and a set of $m$ goods $G = \{g_1, \dots, g_m\}$. Each agent $i \in N$ has a {\em valuation function} $v_i: 2^G \rightarrow \R_{\ge 0}$; $v_i(S)$ denotes how much value agent $i$ has for the bundle $S$. Throughout this paper, we assume agent valuations are additive; that is, $v_i(S) = \sum_{g \in S} v_i(\{g\})$. For an agent $i$ and good $g$, we often use $v_i(g)$ to denote $v_i(\{g\})$.

An {\em allocation} $X = (X_1, \dots, X_n)$ is an $n$-subpartition of the set of goods. Each agent $i$ is allocated the bundle $X_i$ and receives {\em utility} $v_i(X_i)$. 

While our goal is to compute integral allocations, our algorithms heavily use fractional allocations. A fractional allocation $\x$ is an $n \times m$ matrix where each row represents an agent and each column represents a good. Additionally, $x_{ig} \in [0, 1]$ for each $i \in N$, $g \in G$, and $\sum_{i \in N} x_{ig} \le 1$ for each $g \in G$. The value of $x_{ig}$ denotes the amount of good $g$ allocated to agent $i$. We denote the allocated bundle to agent $i$ using $\mathbf{x}_{i} \in [0, 1]^G$. We also define the {\em size} of a fractional bundle $\mathbf{x}_i$ as the value $\sum_{g \in G} x_{ig}$. 

The Nash social welfare (or simply, the Nash welfare) of an allocation $X$ (denoted $\NSW(X)$) is defined as the geometric mean of agent utilities; formally, $\NSW(X) = (\prod_{i \in N} v_i(X_i))^{\frac1n}$. An allocation which maximizes the Nash social welfare is referred to as a max Nash welfare allocation. We denote the optimal Nash welfare of an instance $\cal I$ using $\NSW(\cal I)$.
An allocation $X$ is said to be an $\alpha$-approximation $(\alpha \ge 1)$ of the max Nash welfare for some instance $\cal I$ if $\NSW(X) \ge \frac1{\alpha}\NSW(\cal I)$.

We will need the following simple lemma about Nash welfare.
\begin{lemma}\label{lem:nash-welfare-swap}
Given a fair allocation instance where all agents have the same valuation function $v$, let $X$ be an allocation and let $i, j \in N$ be two agents. Let $S_i \subseteq X_i$ and $S_j \subseteq X_j$ be two sets of goods such that $v(S_j) \ge v(S_i)$ and $v(X_i \setminus S_i) \ge v(X_j \setminus S_j)$. Let $Y$ be an allocation constructed by starting at $X$ and swapping the set of goods $S_i$ and $S_j$; that is, $Y_i = (X_i \setminus S_i) \cup S_j$ and $Y_j = (X_j \setminus S_j) \cup S_i$. Then, $v(Y_i)v(Y_j) \le v(X_i)v(X_j)$. 

Moreover, equality holds if and only if $v(S_j) = v(S_i)$ or $v(X_i \setminus S_i) = v(X_j \setminus S_j)$. 
\end{lemma}
\begin{proof}
We simplify the notation a little further to make the proof readable. Define $a = v(X_i \setminus S_i)$, $b = v(X_j \setminus S_j)$, $c = v(S_j)$ and $d = v(S_i)$. The lemma requires that $a \ge b$ and $c \ge d$. 
\begin{align*}
v(Y_i)v(Y_j) &= (v(X_i \setminus S_i) + v(S_j))(v(X_j \setminus S_j) + v(S_i)) \\
&= (a + c)(b + d) \\
&= (a + d)(b + c) + (a-b)(d-c) \\
&\le (a+d)(b+c) \\
&= v(X_i)v(X_j).
\end{align*} 
In the inequality, we use the fact that $(a - b)(d -c) \le 0$ since $a \ge b$ and $c \ge d$. Note that equality holds if and only if $(a-b)(d-c) = 0$; that is either $a = b$ or $c = d$ or both. 
\end{proof}

We will often use the log Nash welfare since it makes our proofs slightly easier to read. Formally, the log Nash welfare of an allocation $X$ (denoted $\logNSW(X)$) is defined as the value $\frac1n \sum_{i \in N} \ln(v_i(X_i))$. Similar to the Nash social welfare, we use $\logNSW(\cal I)$ to denote the optimal log Nash welfare of the instance $\cal I$. We require the following simple fact.

\begin{fact}\label{fact:log-nash-welfare}
Let $X$ be any (random) allocation and $t \in \R$ be any real value. If $\E[{\logNSW(X)}] \ge t$, then $\E[\NSW(X)] \ge e^{t}$.
\end{fact}
\begin{proof}
The proof follows using Jensen's inequality.
\begin{equation*}
\E[\NSW(X)] = \E[e^{\logNSW(X)}] \ge e^{\E[\logNSW(X)]} \ge e^{t}. \qedhere
\end{equation*}
\end{proof}

We often need the following simple facts.
\begin{fact}\label{fact:x-onebyx}
The function $f: \R_{> 0} \rightarrow \R_{> 0}$ defined by $f(x) = x^{1/x}$ is maximized at $x = e$. 
\end{fact}

\begin{fact}\label{fact:identity-inequality}
For any $x > -1$, $\ln(1 + x) \ge \frac{x}{1 + x}$. 
\end{fact}
\subsection{The configuration LP}
The configuration LP for the max Nash welfare problem is defined as follows:
\begin{alignat}{3}
  &\max  \quad && \frac1n \sum_{i \in N} \sum_{S \subseteq G} y_{i, S} \ln{v_i(S)} \notag \\
  &\text{s.t.} \quad
               && \sum_{S \subseteq G} y_{i, S} = 1,
                   &&\quad \forall\, i \in N, \label{eq:agent-cover} \\
  &            && \sum_{i \in N}\; \sum_{\substack{S \subseteq G: g \in S}} y_{i, S} = 1,
                   &&\quad \forall\, g \in G, \label{eq:item-capacity} \\
  &            && y_{i,S} \geq 0,
                   &&\quad \forall\, i \in N,\ S \subseteq G. \label{eq:nonnegativity}
\end{alignat}

We have an indicator variable $y_{i, S}$ that denotes whether agent $i$ is allocated bundle $S$. \eqref{eq:agent-cover} ensures that each agent is allocated one bundle, and \eqref{eq:item-capacity} ensures each good is allocated exactly once. It is easy to see that this LP is a fractional relaxation of the max log Nash welfare problem, and therefore the max Nash welfare problem.


\cite[Theorem 2.1]{Feng2025weightednashadditive} show that this LP can be efficiently solved within an additive error of $\ln{(1 + \vare)}$ for any constant $\vare > 0$. 

\begin{theorem}[\cite{Feng2025weightednashadditive}]\label{thm:configuration-lp}
For any $\vare > 0$, there is an algorithm that outputs a valid solution $\{y_{i, S} \in \mathbb{Q}_{\ge 0}\}_{i \in N, S \subseteq G}$ whose value is at least the optimum value of the configuration LP minus $\ln{(1 + \vare)}$, represented using a list of the non-zero entries. The running time of the algorithm is polynomial in the input size and $\frac1{\vare}$.
\end{theorem}

\begin{remark}
The algorithm of \cite{Feng2025weightednashadditive} solves the configuration LP with a relaxation of \eqref{eq:item-capacity} where they assume $\sum_{i \in N}\sum_{\substack{S \subseteq G: g \in S}} y_{i, S} \le 1$ for each 
$g\in G$ as opposed to the equality we place in \eqref{eq:item-capacity}. Any such configuration LP solution can be trivially `completed' to enforce \eqref{eq:item-capacity} without reducing the objective value. Therefore the condition \eqref{eq:item-capacity} can be assumed without loss of generality.
\end{remark}

\begin{remark}\label{rem:well-defined-objective}
We assume without loss of generality that the optimal Nash welfare of the input instance is positive. The problem of deciding whether the optimal Nash welfare is positive can be easily checked with a max weight matching algorithm. A positive optimal Nash welfare crucially implies the existence of a configuration LP solution $y$ with finite objective value. This means that for any positive variable $y_{i, S}$, we have $v_i(S) > 0$ and $\ln v_i(S)$ is well-defined. 
\end{remark}

\subsection{Round Robin Algorithm}
Given an instance of the fair division problem, the round robin algorithm takes as input a permutation $\pi: N \rightarrow N$ (referred to as the picking sequence) and outputs an allocation. It is defined as follows: we start with an empty allocation, where no agent receives any good. Then the agents take turns according to the permutation $\pi$ and add their highest valued unallocated good to their bundle. Agent $\pi(1)$ goes first followed by $\pi(2)$ and so on till $\pi(n)$, after which we start with agent $\pi(1)$ again. The process stops when there are no goods left to allocate. An allocation $X$ computed using the round robin algorithm with the picking sequence $\pi$ is referred to as the round robin allocation with picking sequence $\pi$. When $\pi$ is the identity permutation, we refer to $X$ simply as the round robin allocation. We use the following simple facts about round robin allocations.

\begin{fact}\label{fact:round-robin-valuations}
Let $X$ be the round robin allocation with picking sequence $\pi$ in an instance where all agents have the same valuation function $v$. Then, $v(X_{\pi(1)}) \ge v(X_{\pi(2)})\ge \dots \ge v(X_{\pi(n)})$.
\end{fact}

Our second fact concerns the envy-freeness of round robin allocations. An allocation $X$ is said to be envy-free up to one good (EF1) if for any two agents $i, j \in N$, there exists a good $g \in X_j$ such that $v_i(X_i) \ge v_i(X_j \setminus \{g\})$.

\begin{fact}\label{fact:round-robin-ef1}
Let $X$ be the round robin allocation with picking sequence $\pi$. When all agents have additive valuations, $X$ is envy-free up to one good (EF1). Specifically, this implies that for any two agents $i, j \in N$, $v_i(X_j) \le v_i(X_i) + \max_{g \in X_j} v_i(g)$.
\end{fact}

\subsection{The constant $\gamma$}

Throughout this paper, we will use $\gamma$ to denote the value $10^{-10}$. It is useful to think of $\gamma$ as a very small rational constant close to $0$.

\section{Nash Welfare Guarantees of the Round Robin Allocation}\label{sec:round-robin-guarantees}
In this section, we prove Nash welfare guarantees for the round robin allocation when agents have identical valuations. Specifically, let $\cal I$ be any instance of the fair division problem with $\Delta$ agents, all with the same valuation function $v$ over a set of goods $G = \{g_1, \dots, g_m\}$; assume that $v(g_1) \ge v(g_2) \ge \dots \ge v(g_m)$. Let $W$ be the round robin allocation in this instance. To make notation slightly easier, we define $W$ as the allocation where each agent $i$ receives the bundle $W_i = \{g_i, g_{\Delta + i}, \dots\}$ --- this is the round robin allocation where ties are broken in favor of goods with a lower index.  

We have the following result due to \cite{Barman2018FindingFA}.

\begin{lemma}\label{lem:barman-eonebye}
$W$ is an $e^{1/e}$-approximation of the max Nash welfare in the instance $\cal I$. 
\end{lemma}
\begin{proof}
\cite{Barman2018FindingFA} shows that when agents have identical valuations, any envy-free up to one good allocation is an $e^{1/e}$-approximation of the max Nash welfare. By Fact \ref{fact:round-robin-ef1}, $W$ is envy-free up to one good.
\end{proof}

In this section, we examine when the worst case behavior occurs. We show that unless $W$ has a specific structure, it is better than an $e^{1/e}$-approximation of the max Nash welfare. This structure will be used later on to define what it means for an agent to be well-behaved.

\begin{definition}[Large and Small Goods]\label{def:large-goods}
In the instance $\cal I$, we define the set of large goods $L$ (with $|L| = \ell$) as the set of at most $\Delta -1$ goods that satisfies 
\begin{enumerate}[(i)]
\item for each $g \in L$, $v(g) > \frac{v(G \setminus L)}{\Delta - \ell}$, and 
\item for each $g \in G \setminus L$, $v(g) \le \frac{v(G \setminus L)}{\Delta - \ell}$. 
\end{enumerate}
We refer to all the goods in $G \setminus L$ as small goods.
\end{definition}

We have the following simple observation, whose proof is relegated to Appendix \ref{apdx:round-robin-guarantees}.

\begin{restatable}{obs}{obslargegoodsunique}\label{obs:large-goods-unique}
For any instance $\cal I$ with identical valuations, there exists a unique set of large goods $L$ satisfying the properties of Definition \ref{def:large-goods}.
\end{restatable}

\begin{definition}[Mean Value]\label{def:mean}
In the instance $\cal I$, we define the mean value of the instance (denoted $\mu$) to be the value $\frac{v(G \setminus L)}{\Delta}$, where $L$ is the set of large goods of the instance and $\Delta$ is the number of agents.  
\end{definition}

Finally, we make the following assumption about the instance $\cal I$ that applies to all the results in this section. This assumption mainly exists to ensure that every time we apply a logarithm, it is well-defined.  

\begin{ass}\label{ass:non-zero}
The instance $\cal I$ consists of at least $\Delta$ goods that are positively valued by all the agents. This implies that $\NSW(\cal I)$ is positive and $\NSW(W)$ is positive. It also implies that the mean value $\mu$ is positive.
\end{ass}

We use the definitions of large goods and mean value to upper bound the max Nash welfare of the instance $\cal I$. 

\begin{lemma}\label{lem:nash-upper-bound}
Let $\cal I$ be an instance where all agents have the same valuation function $v$. Then, 
$$\normalfont{\NSW}(\cal I) \le \left (\frac{\Delta \mu}{\Delta - \ell} \right )^{1 - \frac{\ell}{\Delta}} \prod_{g \in L} v(g)^{\frac{1}{\Delta}}.$$

Alternatively, 
\begin{align*}
\normalfont{\logNSW}(\cal I) \le \frac{1}{\Delta} \left [\sum_{g \in L} \ln{v(g)} + (\Delta - \ell)\ln\left ( \frac{\Delta \mu}{\Delta - \ell} \right ) \right ] \le  \frac{1}{\Delta} \left [\sum_{g \in L} \ln{v(g)} + (\Delta - \ell)\ln\mu  \right ] + \frac1e.
\end{align*}
\end{lemma}
\begin{proof}
Let $X^*$ be the Nash optimal allocation in this instance.
We first show that any agent in $X^*$ who receives a large good does not receive any other positively valued good. Assume for contradiction that an agent $i$ who receives the large good $g$ receives a set of other goods $S_i$ such that $v(S_i) > 0$. Consider the agent $j$ who is the worst off among the agents who do not receive a large good; that is, $j$ minimizes $v(X^*_j)$ among the agents who do not receive a large good. Since $\ell \le \Delta - 1$, such a $j$ must exist. This agent $j$ receives utility at most $\frac{v(G \setminus L)}{\Delta - \ell} < v(g)$. Therefore, moving the bundle $S_i$ from agent $i$ to agent $j$ strictly increases Nash welfare (using Lemma \ref{lem:nash-welfare-swap}). This contradicts the optimality of $X^*$.

Let $N_L$ be the set of agents who receive a large good in $X^*$. Let $N_S$ be the set of all the other agents. Since each agent in $N_L$ receives exactly one good of positive value,
\begin{align}
\prod_{i \in N_L} v(X^*_i) = \prod_{g \in L} v(g). \label{eq:nash-upper-bound-1}
\end{align}   

We can upper bound the product of all other agents' utilities using the AM-GM inequality:
\begin{align}
\prod_{i \in N_S} v(X^*_i) \le \left ( \frac{v(G \setminus L)}{\Delta - \ell}\right )^{\Delta - \ell}. \label{eq:nash-upper-bound-2}
\end{align}

Combining \eqref{eq:nash-upper-bound-1} and \eqref{eq:nash-upper-bound-2} proves the first statement of the lemma. The second statement follows from applying a logarithm on both sides of the first statement, and using Fact \ref{fact:x-onebyx}.
\end{proof}

To analyze the Nash welfare guarantees of the round robin allocation $W$, we separate the allocation using large goods. 
\begin{definition}[Small Bundles of the Round Robin Allocation]\label{def:small-bundles-round-robin}
In the instance $\cal I$, for each agent $i$, we use $S_i$ to denote the set of {\em small} goods that agent $i$ receives in the round robin allocation $W$.
\end{definition}

\begin{obs}\label{obs:small-bundle-round-robin}
If there are $\ell$ large goods in the instance, $v(S_{\ell + 1}) \ge v(S_{\ell + 2}) \ge \dots \ge v(S_{\Delta}) \ge v(S_{1}) \ge v(S_2) \ge \dots \ge v(S_{\ell})$.
\end{obs}
\begin{proof}
The set of small goods $G \setminus L$ is the set $\{g_{\ell + 1}, \dots, g_{m}\}$. Therefore, the allocation of small goods $(S_1, \dots, S_{\Delta})$ is equivalent to the round robin allocation of the small goods $G \setminus L$ with picking sequence $(\ell +1, \ell + 2, \dots, \Delta, 1, 2, \dots, \ell)$. Therefore, the observation follows from Fact \ref{fact:round-robin-valuations}.
\end{proof}

To prove improved Nash welfare guarantees, we first describe the structure that the instance $\cal I$ must have. This structure is defined by five properties. Formally, we say that the instance $\cal I$ (and the corresponding round robin allocation $W$) is {\em well-behaved} if it satisfies the following properties:

\begin{description}
\descitem{Property A}{prop:a} $\frac{\ell}{\Delta} \in \left [ 1 - \frac{1}{e} - \gamma, 1 - \frac{1}{e} + \gamma \right ]$.
\descitem{Property B}{prop:b} For all but at most $\gamma \Delta$ agents $i \in [\Delta]$, $v(S_i) \in \left [ (1 - \gamma)\mu, (1 + \gamma)\mu \right ]$.
\descitem{Property C}{prop:c} For all but at most $\gamma^2 \Delta$ agents $i \in [\ell]$ who receive a large good, $v(S_i) < \gamma^2 v(g_i)$.
\descitem{Property D}{prop:d} For all but at most $\gamma^2 \Delta$ agents $i \in [\ell]$ who receive a large good, $v(g_i) > \frac{\mu}{\gamma}$.
\descitem{Property E}{prop:e} For all agents $i \in [\Delta]$, $v(S_i) \ge 10^{-4} \mu$.
\end{description}

To gain intuition, one may read the properties initially thinking of $\gamma$ as $0$. The above properties describe our worst case instance, where $W$ is {\em not} better than an $e^{1/e}$-approximation of the Nash welfare. If any of these properties are violated, we show that $W$ is slightly better than an $e^{1/e}$-approximation of the Nash welfare. Formally, we prove the following two results.

\begin{restatable}{theorem}{thmnonwellbehavedinstanceimprovement}\label{thm:non-well-behaved-instance-improvement}
Let $\cal I$ be an instance with $\Delta$ identical agents that satisfies Assumption \ref{ass:non-zero}. Define $W$ as the round robin allocation of this instance. If $\cal I$ is not a well-behaved instance, then for some constant $c > 0$,
\begin{align*}
\logNSW(W) \ge  \logNSW(\cal I) - \frac{1}{e} + c.
\end{align*}
\end{restatable}

Using Fact \ref{fact:log-nash-welfare}, this proves that $W$ is strictly better than an $e^{1/e}$-approximation of the Nash welfare. En route, we also prove the following useful result that applies to all instances.

\begin{restatable}{theorem}{thmroundrobinnashlowerbound}\label{thm:round-robin-nash-lower-bound}
Let $\cal I$ be an instance with $\Delta$ identical agents that satisfies Assumption \ref{ass:non-zero}. Define $W$ as the round robin allocation. The following holds,
\begin{align*}
\logNSW(W) \ge  \frac{1}{\Delta}\left [ \sum_{g \in L} \ln{v(g)} + (\Delta - \ell) \ln\left ( \frac{\Delta\mu}{\Delta - \ell} \right ) \right ] - \frac{1}{e}.
\end{align*}
\end{restatable}

For some intuition of why the above result is useful, note that the expression on the right hand side within the square brackets is identical to our upper bound on the log Nash welfare of the instance (Lemma \ref{lem:nash-upper-bound}). The result shows that $W$ is always not just an $e^{1/e}$-approximation of the Nash welfare, it is an $e^{1/e}$-approximation of our upper bound on the Nash welfare. 

\subsection{Proof of Theorems \ref{thm:non-well-behaved-instance-improvement} and \ref{thm:round-robin-nash-lower-bound}}
Both of these results follow (for the most part) from a closer examination of \cite[Lemma 1]{Barman2018FindingFA} using the lens of large and small goods. The one exception is \ref{prop:e}, which requires an entirely different approach. We first partition the agents using these large and small goods. 

Let $N_L$ be the set of agents who receive a large good in $W$; $N_L = \{1, 2, \dots, \ell\}$. Let $N_1 \subseteq N_L$ be the set of agents $i$ such that $v(S_i) < \gamma^2 v(g_i)$; this is the set of agents where the large good is a significant fraction of the value of the bundle. Similarly, define $N_2 \subseteq N_L$ as the set of agents $i$ such that $v(S_i) \ge \gamma^2 v(g_i)$; this is the set of agents which violate \ref{prop:c}. 

Let $N_S$ be the set of agents who do not receive a large good in the allocation $W$. Define $N_3 \subseteq N_S$ as the set of agents $i$ who receive a bundle with value $v(W_i) > (1 + \gamma^2)\frac{\Delta \mu}{\Delta - \ell}$. Define $N_4 \subseteq N_S$ as the set of agents $i$ who receive a bundle with value in the interval $v(W_i) \in \left [\frac{\Delta \mu}{\Delta - \ell}, (1 + \gamma^2)\frac{\Delta \mu}{\Delta - \ell} \right ]$. Let $N_5$ denote all the remaining agents in $N_S$. The sets $N_1$, $N_2$, $N_3$ and $N_4$ are defined so that the utility of any agent in these groups has a natural lower bound. 

Note that if $N_5$ is empty, then all the $(\Delta - \ell)$ agents who do not receive a large good in $W$ receive a utility of at least $\frac{\Delta \mu}{\Delta - \ell}$. Therefore,
\begin{align*}
\logNSW(W) \ge \frac{1}{\Delta}\left [\sum_{g \in L}\ln{v(g)} + (\Delta - \ell)\ln{\frac{\Delta\mu}{\Delta - \ell}} \right ].
\end{align*}

The right hand side is at least $\logNSW(\cal I)$ (from Lemma \ref{lem:nash-upper-bound}). Therefore, it proves Theorem \ref{thm:non-well-behaved-instance-improvement} with $c = 1/e$ and it trivially proves Theorem \ref{thm:round-robin-nash-lower-bound}. Therefore, from here on, we assume $N_5$ is non-empty.

Using the partition of the agents and the definition of each part, we can lower bound the log Nash welfare of the allocation $W$. In this lower bound, for each agent $i \in N_5$, we use $\alpha_i \in [0, 1]$ to denote a constant such that $v(W_i) = \alpha_i \frac{\Delta \mu}{\Delta - \ell} + (1 - \alpha_i)v(W_{\Delta})$. For each agent $i \in N_5$, this value $\alpha_i$ is well-defined by the definition of agents in $N_5$ having utility at most $\frac{\Delta\mu}{\Delta - \ell}$ and at least $v(W_{\Delta})$. We define $\alpha = \sum_{i \in N_5} \alpha_i$. 

We lower bound the log Nash welfare of $W$ as follows:
\begin{align*}
\logNSW(W) &\ge \frac{1}{\Delta} \left [\sum_{i \in N_1}\ln{v(g_i)} + \sum_{i \in N_2}\ln(v(g_i)(1 + \gamma^2)) + \sum_{i \in N_3}\ln\left ((1 + \gamma^2)\frac{\Delta \mu}{\Delta - \ell} \right ) 
\right .\\
&\qquad \qquad\left .  + \sum_{i \in N_4} \ln\frac{\Delta \mu}{\Delta -\ell} + \sum_{i \in N_5} \ln v(W_i) \right ].
\end{align*}

This simplifies to the following:
\begin{align*}
\logNSW(W) \ge \frac{1}{\Delta} \left [\sum_{g \in L} \ln v(g) + (|N_2| + |N_3|)\ln(1 + \gamma^2) + (|N_3| + |N_4|)\ln\frac{\Delta \mu}{\Delta - \ell} + \sum_{i \in N_5}\ln v(W_{i})\right ].
\end{align*}

Using the concavity of the log function, we have $\ln v(W_i) \ge \alpha_i \ln\frac{\Delta\mu}{\Delta - \ell} + (1 - \alpha_i)\ln v(W_{\Delta})$ for each $i \in N_5$. Using this, we can simplify the above inequality further to
\begin{align*}
\logNSW(W) \ge \frac{1}{\Delta} \left [\sum_{g \in L} \ln v(g) + (|N_2| + |N_3|)\ln(1 + \gamma^2) + (|N_3| + |N_4| + \alpha)\ln\frac{\Delta \mu}{\Delta - \ell} + (|N_5| - \alpha)\ln v(W_{\Delta})\right ].
\end{align*}

To relate this expression to the statements of Theorems \ref{thm:non-well-behaved-instance-improvement} and \ref{thm:round-robin-nash-lower-bound}, we use $|N_3| + |N_4| + |N_5| = |N_S| = \Delta - \ell$ to re-arrange the above inequality.

\begin{align}
\logNSW(W) &\ge \frac{1}{\Delta} \left [\sum_{g \in L} \ln v(g) + (\Delta - \ell)\ln\frac{\Delta \mu}{\Delta - \ell} \right ]+ \frac{(|N_2| + |N_3|)}{\Delta}\ln(1 + \gamma^2) \notag\\
& \qquad \qquad + \frac{(|N_5| - \alpha)}{\Delta}\ln v(W_{\Delta}) - \frac{(|N_5| - \alpha)}{\Delta}\ln\frac{\Delta \mu}{\Delta - \ell} \notag\\
&= \frac{1}{\Delta} \left [\sum_{g \in L} \ln v(g) + (\Delta - \ell)\ln\frac{\Delta \mu}{\Delta - \ell} \right ] + \frac{(|N_2| + |N_3|)}{\Delta}\ln(1 + \gamma^2) + \frac{(|N_5| - \alpha)}{\Delta}\ln \left (\frac{v(W_{\Delta})(\Delta-\ell)}{\Delta \mu} \right) \label{eq:w-lognash-lowerbound}
\end{align}

The key to proving our results is carefully lower bounding the term $\frac{(|N_5| - \alpha)}{\Delta}\ln \left (\frac{v(W_\Delta)(\Delta-\ell)}{\Delta \mu} \right)$. We do this in the following Lemma.

\begin{lemma}\label{lem:w-delta-lower-bound}
There is a positive constant $c_1 > 0$ such that
\begin{align*}
\frac{(|N_5| - \alpha)}{\Delta}\ln \left (\frac{v(W_\Delta)(\Delta-\ell)}{\Delta \mu} \right) \ge -\frac{1}{e} + c_1,
\end{align*}
if any of the following conditions hold:
\begin{enumerate}[(i)]
\item $\frac{|N_5| - \alpha}{\Delta} \notin \left [\frac{1}e - 10^{-4}\gamma, \frac{1}e + 10^{-4}\gamma \right]$,
\item $\alpha \ge 10^{-4} \gamma^4 \Delta$,
\item $|N_4| \ge 10^{-4} \gamma^2 \Delta$,
\item $v(W_{\Delta}) \ge (1 + 10^{-4}\gamma^2)\mu$.
\end{enumerate}

Additionally, irrespective of whether these conditions are satisfied, the above inequality holds with $c_1 = 0$.
\end{lemma}
\begin{proof}
To prove this result, we lower bound $v(W_{\Delta})$ using the following relation:

\begin{align*}
\mu \Delta = v(G \setminus L) = \sum_{i \in [\Delta]} v(S_i)
\end{align*}

We upper bound $\sum_{i \in [\Delta]} v(S_i)$ using the grouping of the agents we constructed earlier. For all the agents $i \in N_L$, $v(S_i) \le v(W_{\Delta})$ from Observation \ref{obs:small-bundle-round-robin}. For all the agents $i \in N_3 \cup N_4$, $v(S_i) \le \frac{\Delta \mu}{\Delta - \ell} + v(W_{\Delta})$ using Fact \ref{fact:round-robin-ef1}. For the agents $i \in N_4$, we additionally have $v(S_i) \le (1 + \gamma^2)\frac{\Delta \mu}{\Delta - \ell}$. For the remaining agents in $N_5$, we have $\sum_{i \in N_5} v(S_i) = \sum_{i \in N_5} v(W_i) = \alpha\frac{\Delta \mu}{\Delta - \ell} + (|N_5| - \alpha)v(W_{\Delta})$. Combining all of these inequalities, we get:

\begin{align}
\mu \Delta &= \sum_{i \in [\Delta]} v(S_i) \notag \\
&= \sum_{i \in N_L} v(S_i) + \sum_{i \in N_3} v(S_i) + \sum_{i \in N_4}v(S_i) + \sum_{i \in N_5} v(S_i) \notag\\
&\le |N_L| v(W_{\Delta}) + |N_3| \left ( \frac{\Delta \mu}{\Delta - \ell} + v(W_{\Delta})\right ) + |N_4| \min\left \{\left ( \frac{\Delta \mu}{\Delta - \ell} + v(W_{\Delta})\right ),  (1 + \gamma^2)\frac{\Delta \mu}{\Delta - \ell} \right \} \notag\\ 
&\qquad \qquad + \alpha\frac{\Delta \mu}{\Delta - \ell} + (|N_5| - \alpha)v(W_{\Delta}) \label{eq:w-delta-lower-bound-1}.
\end{align}

Depending on how we upper bound the $\min\{\}$ function, we get two different lower bounds on $v(W_{\Delta})$. We first use $\min\left \{\left ( \frac{\Delta \mu}{\Delta - \ell} + v(W_{\Delta})\right ),  (1 + \gamma^2)\frac{\Delta \mu}{\Delta - \ell} \right \} \le \frac{\Delta \mu}{\Delta - \ell} + v(W_{\Delta})$. This simplifies \eqref{eq:w-delta-lower-bound-1} to:
\begin{align*}
\mu \Delta \le |N_L| v(W_{\Delta}) + (|N_3| + |N_4|)\left ( \frac{\Delta \mu}{\Delta - \ell} + v(W_{\Delta})\right ) + \alpha\frac{\Delta \mu}{\Delta - \ell} + (|N_5| - \alpha)v(W_{\Delta})
\end{align*}

Re-arranging terms and using the facts that $\Delta = |N_L| + |N_3| + |N_4| + |N_5|$ and $\Delta - \ell = |N_S| = |N_3| + |N_4| + |N_5|$, we get
\begin{align*}
\frac{v(W_{\Delta})(\Delta - \ell)}{\Delta \mu} \ge \frac{|N_5| - \alpha}{\Delta - \alpha}.
\end{align*}

This implies
\begin{align}
\frac{(|N_5| - \alpha)}{\Delta}\ln \left (\frac{v(W_\Delta)(\Delta-\ell)}{\Delta \mu} \right) &\ge \frac{(|N_5| - \alpha)}{\Delta}\ln \left ( \frac{|N_5| - \alpha}{\Delta - \alpha} \right ) \notag \\
& = \frac{(|N_5| - \alpha)}{\Delta}\ln \left ( \frac{|N_5| - \alpha}{\Delta} \right ) + \frac{(|N_5| - \alpha)}{\Delta} \ln \left (\frac{\Delta}{\Delta - \alpha} \right ). \label{eq:w-delta-lower-bound-2}
\end{align}

The first term is minimized when $|N_5| - \alpha = \frac{\Delta}{e}$ and achieves a minimum value of $-\frac{1}{e}$ (Fact \ref{fact:x-onebyx}). The second term is non-negative. This shows that
\begin{align*}
\frac{(|N_5| - \alpha)}{\Delta}\ln \left (\frac{v(W_{\Delta})(\Delta-\ell)}{\Delta \mu} \right) &\ge -\frac{1}{e}.
\end{align*}

This proves the second part of the lemma. Additionally, when $\frac{|N_5| - \alpha}{\Delta}$ is not in the interval $\left [\frac{1}e - 10^{-4}\gamma, \frac{1}e + 10^{-4}\gamma \right]$, there is some constant $c'_1$ such that 
\begin{align*}
\frac{(|N_5| - \alpha)}{\Delta}\ln \left (\frac{v(W_{\Delta})(\Delta-\ell)}{\Delta \mu} \right) &\ge -\frac1e + c'_1 + \frac{(|N_5| - \alpha)}{\Delta} \ln \left (\frac{\Delta}{\Delta - \alpha} \right ) \\
&\ge -\frac1e + c'_1.
\end{align*}

This shows that if condition (i) is satisfied, then we achieve our desired result. So from here on we assume condition (i) is not satisfied. Consider the case where condition (ii) is satisfied. In this case, \eqref{eq:w-delta-lower-bound-2} reduces to
\begin{align*}
\frac{(|N_5| - \alpha)}{\Delta}\ln \left (\frac{v(W_{\Delta})(\Delta-\ell)}{\Delta \mu} \right) &\ge -\frac{1}{e} + \frac{(|N_5| - \alpha)}{\Delta} \ln \left (\frac{\Delta}{\Delta - \alpha} \right ) \\
&\ge -\frac{1}{e} + \frac13 \ln \left (\frac{1}{1 - 10^{-4}\gamma^4} \right ).
\end{align*}
The final inequality uses our assumptions that condition (i) is not satisfied but condition (ii) is satisfied.
Therefore, if condition (ii) is satisfied, then for $c'_2 = \frac13 \ln \left (\frac{1}{1 - 10^{-4}\gamma^4} \right )$, we have $\frac{(|N_5| - \alpha)}{\Delta}\ln \left (\frac{v(W_{\Delta})(\Delta-\ell)}{\Delta \mu} \right) \ge -\frac1e + c'_2$. This is our desired result. So from here on we assume condition (ii) is not satisfied.

To prove our result for condition (iii), we return to \eqref{eq:w-delta-lower-bound-1} and upper bound $\min\left \{\left ( \frac{\Delta \mu}{\Delta - \ell} + v(W_{\Delta})\right ),  (1 + \gamma^2)\frac{\Delta \mu}{\Delta - \ell} \right \}$ using $(1 + \gamma^2)\frac{\Delta \mu}{\Delta - \ell}$. Using a similar analysis, we get
\begin{align*}
\frac{v(W_{\Delta})(\Delta - \ell)}{\mu \Delta} \ge \frac{|N_5| - \alpha - \gamma^2 |N_4|}{\Delta - \alpha - |N_4|} &\ge \frac{|N_5| - \alpha - \gamma^2 |N_4|}{\Delta- |N_4|}.
\end{align*}

The right hand side of the above inequality can be simplified as follows using basic arithmetic operations and using $|N_5| - \alpha \ge \frac{\Delta}{3}$ (from condition (i) not being satisfied):
\begin{align*}
\frac{v(W_{\Delta})(\Delta - \ell)}{\mu \Delta} &\ge \frac{|N_5| - \alpha - \gamma^2 |N_4|}{\Delta- |N_4|} \\
&\ge \frac{|N_5| - \alpha}{\Delta} \left [ \frac{1 - \frac{\gamma^2 |N_4|}{|N_5| - \alpha}}{1 - \frac{|N_4|}{\Delta}} \right ] \\
&\ge \frac{|N_5| - \alpha}{\Delta} \left [ \frac{1 - \frac{3\gamma^2 |N_4|}{\Delta}}{1 - \frac{|N_4|}{\Delta}} \right ] \\
&\ge \left (\frac{|N_5| - \alpha}{\Delta} \right ) e^{\frac{|N_4|}{\Delta}(1 - 6\gamma^2)}.
\end{align*}

This implies
\begin{align*}
\frac{(|N_5| - \alpha)}{\Delta}\ln \left (\frac{v(W_{\Delta})(\Delta-\ell)}{\Delta \mu} \right) &\ge \frac{(|N_5| - \alpha)}{\Delta} \ln \left ( \frac{|N_5| - \alpha}{\Delta} \right ) + \left (\frac{|N_5| - \alpha}{\Delta} \right )\left (\frac{|N_4|}{\Delta}(1 - 6\gamma^2) \right ) \\
&\ge -\frac1e + \frac13 \left (\frac{|N_4|}{\Delta}(1 - 6\gamma^2) \right ).
\end{align*}

Therefore, if $|N_4| \ge 10^{-4}\gamma^2 \Delta$ (condition (iii) is satisfied), then there is some positive constant $c'_3$ such that $\frac{(|N_5| - \alpha)}{\Delta}\ln \left (\frac{v(W_{\Delta})(\Delta-\ell)}{\Delta \mu} \right) \ge - \frac1e + c'_3$.

Finally, to prove our result for condition (iv), we directly plug in the lower bound on $v(W_{\Delta})$. This gives us:
\begin{align*}
\frac{(|N_5| - \alpha)}{\Delta}\ln \left (\frac{v(W_{\Delta})(\Delta-\ell)}{\Delta \mu} \right) &\ge \frac{(|N_5| - \alpha)}{\Delta}\ln \left (\frac{(1 + 10^{-4}\gamma^2)(\Delta-\ell)}{\Delta} \right) \\
&\ge \frac{(|N_5| - \alpha)}{\Delta}\ln \left (\frac{(\Delta-\ell)}{\Delta} \right ) + \frac{(|N_5| - \alpha)}{\Delta}\ln(1 + 10^{-4}\gamma^2).
\end{align*}

To simplify this further, we use the fact that $|N_5| - \alpha$ is lowerbounded by $\frac{\Delta}3$ (from condition (i) not being satisfied) and upper bounded by $\Delta - \ell$. This gives us
\begin{align*}
\frac{(|N_5| - \alpha)}{\Delta}\ln \left (\frac{v(W_{\Delta})(\Delta-\ell)}{\Delta \mu} \right) &\ge \frac{(\Delta - \ell)}{\Delta}\ln \left (\frac{(\Delta-\ell)}{\Delta} \right )+ \frac{1}{3}\ln(1 + 10^{-4}\gamma^2) \\
&\ge -\frac1e + \frac{1}{3}\ln(1 + 10^{-4}\gamma^2).
\end{align*}

This shows that if condition (iv) is satisfied, we prove our desired result with constant $c'_4 = \frac{1}{3}\ln(1 + 10^{-4}\gamma^2)$. In conclusion, the lemma is proved with constant $c_1 = \min\{c'_1, c'_2, c'_3, c'_4\}$.
\end{proof}

We can use this lowerbound to prove approximation guarantees for the Nash welfare of $W$. 

\begin{lemma}\label{lem:nash-welfare-improvement-conditions}
There is some positive constant $c_2 > 0$ such that 
\begin{align*}
\logNSW(W) \ge \frac{1}{\Delta} \left [\sum_{g \in L} \ln v(g) + (\Delta - \ell)\ln\frac{\Delta \mu}{\Delta - \ell} \right ] -\frac{1}{e} + c_2, 
\end{align*}
if any of the following conditions hold:
\begin{enumerate}[(a)]
\item $|N_2| \ge 10^{-4}\gamma^2 \Delta$,
\item $|N_3| \ge 10^{-4}\gamma^2 \Delta$,
\item $\frac{|N_5| - \alpha}{\Delta} \notin \left [\frac{1}e - 10^{-4}\gamma, \frac{1}e + 10^{-4}\gamma \right]$,
\item $\alpha \ge 10^{-4} \gamma^4 \Delta$,
\item $|N_4| \ge 10^{-4} \gamma^2 \Delta$,
\item $v(W_{\Delta}) \ge (1 + 10^{-4}\gamma^2)\mu$.
\end{enumerate}
Additionally, irrespective of whether these conditions are satisfied, the above inequality holds with $c_2 = 0$.
\end{lemma}
\begin{proof}
We return to \eqref{eq:w-lognash-lowerbound}. Lemma \ref{lem:w-delta-lower-bound} shows that $\frac{(|N_5| - \alpha)}{\Delta}\ln \left (\frac{v(W_{\Delta})(\Delta-\ell)}{\Delta \mu} \right) \ge -\frac{1}{e}$. Using this lower bound, \eqref{eq:w-lognash-lowerbound} simplifies to
\begin{align*}
\logNSW(W) &\ge \frac{1}{\Delta} \left [\sum_{g \in L} \ln v(g) + (\Delta - \ell)\ln\frac{\Delta \mu}{\Delta - \ell} \right ] + \frac{(|N_2| + |N_3|)}{\Delta}\ln(1 + \gamma^2) -\frac1e.
\end{align*}
The above inequality immediately proves the second part of the lemma. Additionally, if either (a) or (b) holds, the above inequality further simplifies to
\begin{align*}
\logNSW(W) &\ge \frac{1}{\Delta} \left [\sum_{g \in L} \ln v(g) + (\Delta - \ell)\ln\frac{\Delta \mu}{\Delta - \ell} \right ] + 10^{-4}\gamma^2\ln(1 + \gamma^2) -\frac1e.
\end{align*}

This proves our desired result for the conditions (a) and (b). If any of the remaining four conditions hold, Lemma \ref{lem:w-delta-lower-bound} shows that $\frac{(|N_5| - \alpha)}{\Delta}\ln \left (\frac{v(W_{\Delta})(\Delta-\ell)}{\Delta \mu} \right) \ge -\frac{1}{e} + c_1$. Plugging this into \eqref{eq:w-lognash-lowerbound} gives us
\begin{align*}
\logNSW(W) &\ge \frac{1}{\Delta} \left [\sum_{g \in L} \ln v(g) + (\Delta - \ell)\ln\frac{\Delta \mu}{\Delta - \ell} \right ] - \frac1e + c_1.
\end{align*}

This proves our desired result for the conditions (c), (d), (e) and (f).
\end{proof}

The second part of the above lemma implies Theorem \ref{thm:round-robin-nash-lower-bound}. To prove Theorem \ref{thm:non-well-behaved-instance-improvement}, we show an improved approximation factor when any of the properties are violated.

\begin{lemma}\label{lem:property-a}
If \ref{prop:a} is violated, the following holds for some positive constant $c_2 > 0$: 
\begin{align*}
\logNSW(W) \ge \logNSW(\cal I) -\frac{1}{e} + c_2. 
\end{align*}
\end{lemma}
\begin{proof}
Note that $\ell = |N_L| = \Delta - |N_3| - |N_4| - |N_5|$. If \ref{prop:a} is violated, one of the following conditions must hold:
\begin{enumerate}[(a)]
\item $|N_3| \ge 0.01\gamma \Delta$,
\item $|N_4| \ge 0.01\gamma \Delta$,
\item $\frac{|N_5| - \alpha}{\Delta} \notin \left [\frac{1}e - 0.01\gamma, \frac{1}e + 0.01\gamma \right]$,
\item $\alpha \ge 0.01\gamma\Delta$.
\end{enumerate}

If any of these conditions are true, Lemma \ref{lem:nash-welfare-improvement-conditions} shows that
\begin{align*}
\logNSW(W) \ge \frac{1}{\Delta} \left [\sum_{g \in L} \ln v(g) + (\Delta - \ell)\ln\frac{\Delta \mu}{\Delta - \ell} \right ] -\frac{1}{e} + c_2.
\end{align*}

The result follows by combining the above inequality with Lemma \ref{lem:nash-upper-bound}.
\end{proof}

\begin{lemma}\label{lem:property-c}
If \ref{prop:c} is violated, the following holds for some positive constant $c_2 > 0$: 
\begin{align*}
\logNSW(W) \ge \logNSW(\cal I) -\frac{1}{e} + c_2.
\end{align*}
\end{lemma}
\begin{proof}
\ref{prop:c} is violated if and only if $|N_2| \ge \gamma^2 \Delta$. In this case, Lemma \ref{lem:nash-welfare-improvement-conditions} shows that 
\begin{align*}
\logNSW(W) \ge \frac{1}{\Delta} \left [\sum_{g \in L} \ln v(g) + (\Delta - \ell)\ln\frac{\Delta \mu}{\Delta - \ell} \right ] -\frac{1}{e} + c_2.
\end{align*}

The result follows by combining the above inequality with Lemma \ref{lem:nash-upper-bound}.
\end{proof}

We turn to \ref{prop:b} next.

\begin{lemma}\label{lem:property-b}
If \ref{prop:b} is violated, the following holds for some positive constant $c_2 > 0$: 
\begin{align*}
\logNSW(W) \ge \logNSW(\cal I) -\frac{1}{e} + c_2.
\end{align*}
\end{lemma}
\begin{proof}
In this proof, we assume for contradiction that none of the conditions in Lemma \ref{lem:nash-welfare-improvement-conditions} are satisfied. Note that this also implies \ref{prop:a} and \ref{prop:c} are satisfied by the instance. 
We will show that if \ref{prop:b} is violated, then one of the conditions of Lemma \ref{lem:nash-welfare-improvement-conditions} is satisfied. This then implies the lemma via Lemma \ref{lem:nash-upper-bound}.

Assume \ref{prop:b} is violated such that there are at least $\frac{\gamma \Delta}{2}$ agents $i$ with $v(S_i) \le (1 - \gamma)\mu$. Let $N'$ be the set of $i$ such that $i \le \ell$ and $v(S_i) \le (1 - \gamma)\mu$. Note that $|N'| \ge \frac{\gamma \Delta}{2}$ using Observation \ref{obs:small-bundle-round-robin} and the fact that $\ell \ge \frac{\Delta}{2}$ (assuming Property A is satisfied). Recall that $\Delta \mu = \sum_{i \in [\Delta]} v(S_i)$. Let $\mu^*$ denote the mean of $v(S_i)$ among all the agents in $[\Delta] \setminus N'$. 

\begin{align*}
\mu^* = \frac1{\Delta - |N'|} \sum_{i \in [\Delta] \setminus N'} v(S_i) &\ge \frac{1}{\Delta - |N'|} \left (\Delta \mu - |N'|(1- \gamma) \mu\right ) \\
&\ge \frac{1}{\Delta - \frac{\gamma \Delta}{2}} \left (\Delta \mu -  \frac{\gamma \Delta}{2} (1 - \gamma) \mu \right )\\
&\ge \mu \left ( 1 + \frac{\gamma^2}{2} \right ).
\end{align*}

From Observation \ref{obs:small-bundle-round-robin}, the bundles $\{S_{\ell+1}, \dots, S_\Delta\}$ are the $\Delta-\ell$ highest valued small bundles. Therefore, 
\begin{align*}
\sum_{i \in N_S} v(W_i) \ge \mu^* |N_S| \ge \mu |N_S| \left (1 + \frac{\gamma^2}{2} \right ).
\end{align*}

Moreover, using Fact \ref{fact:round-robin-ef1}, for any agent $i \in N_S$,
\begin{align*}
v(W_i) \le \max_{g \in W_i}v(g) + v(W_{\Delta}) \le \frac{\Delta \mu}{\Delta - \ell} + v(W_{\Delta}) \le 3\mu + (1 + 10^{-4}\gamma^2)\mu \le 5\mu.
\end{align*}
In this inequality, we assume \ref{prop:a} is satisfied and $v(W_{\Delta}) \le (1+10^{-4}\gamma^2) \mu$ (Lemma \ref{lem:nash-welfare-improvement-conditions}[condition (f)]).

Using a simple counting argument, at least $\frac{\gamma^2 |N_S|}{20} \ge \frac{\gamma^2 \Delta}{60}$ agents $i \in N_S$ have $v(W_i) \ge \left (1 + \frac{\gamma^2}{4} \right ) \mu$. Let $\widetilde{N}$ denote the set of these agents.

Since $\widetilde N$ only contains agents who do not receive a large good, one of $N_3 \cap \widetilde{N}, N_4 \cap \widetilde{N}$ and $N_5 \cap \widetilde{N}$ must have size at least $\frac{\gamma^2 \Delta}{200}$. If $N_3$ or $N_4$ have size at least $\frac{\gamma^2 \Delta}{200}$, then one of conditions (b) and (e) in Lemma \ref{lem:nash-welfare-improvement-conditions} must be satisfied --- a contradiction. Therefore, it must be that $|N_5 \cap \widetilde{N}| \ge \frac{\gamma^2 \Delta}{200}$. In this case, we lower bound $\alpha$.

To lower bound $\alpha$, we lower bound the $\alpha_i$ values. For any agent $i \in \widetilde{N} \cap N_5$, we have $\alpha_i \ge \frac{\gamma^2}{24}$. This follows from the following sequence of inequalities.
\begin{align*}
v(W_i) \ge \left (1 + \frac{\gamma^2}{4} \right )\mu = \frac{\gamma^2}{24} (3\mu) + \left (1 + \frac{\gamma^2}{8} \right )\mu \ge \frac{\gamma^2}{24} \frac{\Delta \mu}{\Delta - \ell} + v(W_{\Delta}).
\end{align*}
In the final inequality we use both Property A and Lemma \ref{lem:nash-welfare-improvement-conditions}[condition (f)].

Summing over all $i \in \widetilde{N} \cap N_5$, we get $\alpha \ge \sum_{i \in \tilde N \cap N_5} \alpha_i \ge 10^{-4} \gamma^4 \Delta$. This implies that condition (d) of Lemma \ref{lem:nash-welfare-improvement-conditions} is satisfied --- another contradiction. This shows that when $\frac{\gamma \Delta}{2}$ agents $i$ have $v(S_i) \le (1 - \gamma)\mu$, the lemma holds. 

The proof for the other case, where there are $\frac{\gamma \Delta}{2}$ agents $i 
\in N_S$ with $v(W_i) \ge (1 + \gamma)\mu$ follows using a similar argument. In summary, if \ref{prop:b} is not satisfied, one of the conditions of Lemma \ref{lem:nash-welfare-improvement-conditions} must be satisfied, which implies 
\begin{align*}
\logNSW(W) \ge \logNSW(\cal I) -\frac{1}{e} + c_2.
\end{align*}
This proves the lemma.
\end{proof}

Next, we show that $\cal I$ must satisfy \ref{prop:e}. To prove that this property must be satisfied, we establish an improved upper bound on the optimal Nash welfare for the special case where $v(S_{\ell}) < 10^{-4}\mu$. 
\begin{lemma}\label{lem:property-e}
Let $\cal I$ be an instance with $\Delta$ identical agents with valuation function $v$ that satisfies Properties A, B and C. If the instance has $\ell$ large goods and $v(S_{\ell}) < 10^{-4}\mu$, then the following holds for some constant $c_3 > 0$.
\begin{align*}
\logNSW(W) \ge \logNSW(\cal I) - \frac1e + c_3.
\end{align*}
\end{lemma}
\begin{proof}
We define $G' \subseteq G$ as the set of goods with value at least $0.99\mu$. All the large goods trivially fall into $G'$ since each large good has value of at least $\frac{\Delta \mu}{\Delta - \ell}$ and $\ell \ge \left (1 - 1/e - \gamma \right )\Delta$ (\ref{prop:a}). 

Consider the small bundle of an agent $i$ who receives a large good; that is, $i \le \ell$. By our construction, $S_i = \{g_{\Delta + i}, g_{2\Delta + i}, \dots\}$. The bundle $S_i \setminus \{g_{\Delta + i}\} = \{g_{2\Delta + i}, g_{3\Delta + i}, \dots\}$ is at most as valuable as $S_{\ell} = \{g_{\Delta + \ell}, g_{2\Delta + \ell}, \dots\}$. This has two implications. First, the only good in $S_i$ which can be in $G'$ is $g_{\Delta + i}$. Second, we can lower bound $v(g_{\Delta + i})$ using $v(S_{\ell})$ as follows:
\begin{align*}
v(S_i) = v(g_{\Delta + i}) + v(S_i \setminus \{g_{\Delta + i}\}) \le v(g_{\Delta + i}) + v(S_\ell) \le v(g_{\Delta + i}) + 10^{-4}\mu. 
\end{align*}
Re-arranging the terms of this inequality, we get for all agents $i \in [\ell]$
\begin{align}
v(g_{\Delta + i}) \ge v(S_i) - 10^{-4}\mu. \label{eq:large-small-bundle}
\end{align}

We can make a similar argument about agents who do not receive a large good. If an agent $i$ does not receive a large good, $S_i = \{g_i, g_{\Delta + i}, \dots\}$.
The bundle $S_i \setminus \{g_{i}\} = \{g_{\Delta + i}, g_{2\Delta + i}, \dots\}$ is at most as valuable as $S_{\ell} = \{g_{\Delta + \ell}, g_{2\Delta + \ell}, \dots\}$. 
Using a similar argument, we get that the only good from $G'$ in $S_i$ is the good $g_i$, and we get for all agents $i > \ell$,
\begin{align}
v(g_i) \ge v(S_i) - 10^{-4}\mu. \label{eq:small-small-bundle}
\end{align}

By \ref{prop:b}, all but $\gamma \Delta$ agents satisfy $v(S_i) \ge (1 - \gamma)\mu$. Let $N'$ be the set of these agents $i \in [\Delta]$ such that $v(S_i) \ge (1 - \gamma)\mu$. For all the agents $i \in N'$, from \eqref{eq:large-small-bundle} and \eqref{eq:small-small-bundle}, there is exactly one good in each small bundle $S_i$ with value at least $0.99\mu$. Using this, we can upper bound the total value of $G \setminus G'$ as follows:
\begin{align}
v(G \setminus G') &\le v(\{g_{\Delta + \ell}, g_{\Delta + \ell + 1}, \dots, g_m\})+\sum_{i \in [\Delta] \setminus N'} v(S_i) \notag\\
&\le \Delta v(S_{\ell}) + \sum_{i \in [\Delta] \setminus N'} v(S_i) \notag \\
&\le \Delta v(S_{\ell}) + \gamma \Delta (1 - \gamma) \mu  \notag\\
&\le 2 \cdot 10^{-4} \Delta \mu. \label{eq:g-minus-gprime}
\end{align} 
Additionally, $|G'| \in [\Delta + \ell - \gamma \Delta, \Delta + \ell]$ since each small bundle $S_i$ has at most one good in $G'$ and at least $\Delta - \gamma \Delta$ small bundles have one good in $G'$.

Let $X^*$ denote the optimal Nash allocation of $\cal I$.
To upper bound the Nash welfare of $X^*$, we consider the optimal Nash allocation where the goods in $G'$ are allocated integrally and all the other goods are allocated fractionally. Formally, this is a fractional allocation with the constraint that the goods in $G'$ are allocated integrally. It is easy to see that the optimal Nash welfare of such an allocation upper bounds the Nash welfare of any fully integral allocation in this instance. We denote this allocation using $(Z^*, p)$ where $Z^*_i$ denotes the goods allocated from $G'$ to agent $i$, and $p_i$ denotes the total value allocated from the goods in $G \setminus G'$ to agent $i$. Note that $\sum_{i \in [\Delta]} p_i = v(G \setminus G') \le 2\cdot 10^{-4} \Delta \mu$ (from \eqref{eq:g-minus-gprime}). In this allocation, the utility of agent $i$ is given by $v(Z^*_i) + p_i$. For all agents $i$, we denote $g^*_i$ as the highest valued good in the bundle $Z^*_i$. 

We may assume, without loss of generality, $Z^*_i$ is non-empty for each agent $i \in [\Delta]$. This holds due to the following swapping argument. Assume that $Z^*_i = \emptyset$. There must be some other agent $j$ such that $|Z^*_j| \ge 2$. Let $g$ be some good in $Z^*_j \setminus \{g^*_j\}$. If $p_i < v(g^*_j)$, then moving $g$ from $Z^*_j$ to $Z^*_i$ strictly increases the Nash welfare of the allocation (using an argument similar to Lemma \ref{lem:nash-welfare-swap}). If $p_i \ge v(g^*_j)$, we add $g$ to $Z^*_i$ and deduct $p_i$ by the value $v(g)$. We add this deducted value to agent $j$ by increasing $p_j$ by $v(g)$. This swap does not change the utility of any agent but reduces the number of agents with empty bundles by $1$. Applying this swap repeatedly, we end up with an optimal Nash allocation $(Z^*, p)$ where each $Z^*_i$ is non-empty. 

Additionally, we can assume that if there is some large good $g$ in the bundle $Z^*_i$, then $Z^*_i = \{g\}$ and $p_i = 0$. This follows from the definition of the large goods. There are at least $\Delta - \ell$ agents who do not receive any large goods, which implies there is at least one agent $j$ who receives a utility of at most $\frac{v(G \setminus L)}{\Delta - \ell} < v(g)$ for any large good $g$. Therefore, if $p_i \ne 0$ or $Z^*_i \setminus \{g\} \ne \emptyset$, moving all of this value from $i$ to $j$ strictly improves the Nash welfare (from Lemma \ref{lem:nash-welfare-swap}).

Let $N_L^*$ denote the set of agents who receive a large good in the allocation $Z^*$ and let $N_S^*$ denote all the other agents. Since each agent in $N_L^*$ receives only one good:
\begin{align}
\prod_{i \in N_L^*} (v(Z^*_i) + p_i) = \prod_{g \in L} v(g) = \prod_{i = 1}^{\ell} v(g_i). \label{eq:z-star-large-product}
\end{align}

For all the other agents:
\begin{align*}
\prod_{i \in N_S^*} (v(Z^*_i) + p_i) &\le \prod_{i \in N_S^*} |Z^*_i|v(g^*_i)\left (1 + \frac{p_i}{v(g^*_i)} \right ) \\
&= \left ( \prod_{i \in N_S^*} |Z^*_i| \right ) \left ( \prod_{i \in N_S^*} v(g^*_i) \right ) \prod_{i \in N_S^*} \left (1 + \frac{p_i}{v(g^*_i)} \right )
\end{align*}

We upper bound each of these three terms separately. For the third term, $v(g^*_i) \ge 0.99\mu$ by definition of $G'$ and $\sum_{i \in [\Delta]} p_i \le 2 \cdot 10^{-4} \Delta \mu$. Therefore, we can use the AM-GM inequality to get the following upper bound:

\begin{align*}
\prod_{i \in N_S^*} \left (1 + \frac{p_i}{v(g^*_i)} \right ) &\le \left (1 + \frac{\sum_{i\in [\Delta]} p_i}{|N_S^*| \cdot 0.99\mu} \right )^{|N_S^*|} \\
&\le \left (1 + \frac{2\cdot 10^{-4} \Delta \mu}{|N_S^*| \cdot 0.99 \mu} \right )^{|N_S^*|} \\
&\le \left (1 + 10^{-3} \right )^{|N_S^*|}.
\end{align*}
In the last inequality, we use \ref{prop:a} which implies that $|N_S^*| = \Delta - \ell \ge \Delta \left  (\frac{1}e - \gamma \right )\ge \frac{\Delta}3$. The second term is maximized when the $|N_S^*|$ goods of the form $g^*_i$ are the highest valued small goods in the instance. That is $\prod_{i \in N_S^*} v(g^*_i) \le \prod_{i = \ell + 1}^{\Delta} v(g_i)$. 

For the first term, we need to upper bound $\prod_{i \in N_S^*} |Z^*_i|$ subject to each $|Z^*_i|$ being an integer and $\sum_{i \in N_S^*} |Z^*_i| \le \Delta$, since there are at most $\Delta$ small goods in $G'$. Since $\frac{|N_S^*|}{\Delta} \in [\frac{1}{e} - \gamma, \frac{1}{e} + \gamma]$ (\ref{prop:a}), $\prod_{i \in [\Delta]} |Z^*_i|$ is maximized when some agents receive three goods and all the remaining agents receive two goods. Therefore, 

\begin{align*}
\prod_{i \in N_S^*} |Z^*_i| \le 3^{\beta} 2^{|N_S^*| - \beta} = \left ( \frac{3}{2} \right )^{\beta} 2^{|N_S^*|},
\end{align*}
where $\beta$ is such that $3\beta +  2(|N_S^*| - \beta) \le \Delta$. Re-arranging this inequality and applying $|N_S^*| \ge \left (\Delta(\frac{1}{e} - \gamma) \right )$, we get $\beta \le \frac{e - 2 + 2e\gamma}{e}\Delta$. When $\gamma = 10^{-10}$, we can calculate $\beta \le 0.265 \Delta$ and $|N_S^*| \le 0.368\Delta$. Therefore,
\begin{align*}
\prod_{i \in N_S^*} |Z^*_i| \le \left ( \frac{3}{2} \right )^{0.265 \Delta} 2^{0.368\Delta} \le 1.437^{\Delta}.
\end{align*}

Using these three inequalities, we get the following upper bound on $\prod_{i \in N_S^*} (v(Z^*_i) + p_i)$.
\begin{align}
\prod_{i \in N_S^*} (v(Z^*_i) + p_i) \le (1.437)^{\Delta} \left (\prod_{i = \ell + 1}^{\Delta} v(g_i)  \right ) \left (1 + 10^{-3} \right )^\Delta \le (1.439)^{\Delta} \left (\prod_{i = \ell + 1}^{\Delta} v(g_i)  \right ). \label{eq:z-star-small-product}
\end{align}

Putting \eqref{eq:z-star-large-product} and \eqref{eq:z-star-small-product} together, we get:

\begin{align*}
\prod_{i \in [\Delta]} v(X^*_i) \le \prod_{i \in [\Delta]} (v(Z^*_i) + p_i) \le (1.439)^{\Delta} \prod_{i \in [\Delta]} v(g_i) \le \left (e^{1/e} - 0.001 \right )^{\Delta} \prod_{i \in [\Delta]} v(g_i).
\end{align*}

The lemma follows from noting that $v(W_i) \ge v(g_i)$ for each agent $i \in [\Delta]$. Therefore, 
\begin{align*}
\NSW(W) \ge \frac{1}{e^{1/e} - 0.001}\NSW(\cal I).  
\end{align*}

Applying a logarithm on both sides completes the proof.
\end{proof}

Combining all of these results, we can prove the main result of this section.

\thmnonwellbehavedinstanceimprovement*
\begin{proof}
If any of properties A, B and C are not satisfied, Lemmas \ref{lem:property-a}, \ref{lem:property-b} and \ref{lem:property-c} prove the theorem. If \ref{prop:e} is not satisfied, Lemma \ref{lem:property-e} proves the theorem since $S_\ell$ is the least valuable small bundle (Observation \ref{obs:small-bundle-round-robin}). 

Finally, note that \ref{prop:c} and \ref{prop:e} imply \ref{prop:d}, because all agents $i$ who receive a large good and satisfy $\gamma^2 v(g_i) > v(S_i)$ also satisfy $\gamma v(g_i) > \mu$ since \ref{prop:e} implies that $v(S_i) \ge 10^{-4} \mu$. Therefore \ref{prop:c} and \ref{prop:e} together imply that for all but at most $\gamma^2\Delta$ agents $i$ who receive a large good, $v(g_i) > \frac{\mu}{\gamma}$, which implies that \ref{prop:d} is satisfied. 
\end{proof}

\section{Adding and Removing Goods to and from the Instance}\label{sec:add-remove}
In this section, we continue our analysis of round robin allocations in instances with identical valuations. We bound the Nash welfare of the round robin allocation when goods are added and/or removed from the instance $\cal I$. This instance modification can take various forms. We present four specific results which will come in handy later on. All of them follow from roughly similar arguments, so we present just one and relegate the other three to Appendix \ref{apdx:add-remove}.

In all the results in this section, we continue to make Assumption \ref{ass:non-zero} about the original instance $\cal I$. We do not make any such assumptions about the modified instance after adding or removing goods.

\begin{theorem}\label{thm:small-good-removal-large-good-addition}
Let $\cal I$ be a well-behaved instance that satisfies Assumption \ref{ass:non-zero} with $\Delta$ identical agents, a set of large goods $L$ and mean value $\mu$. Let $W$ denote the round robin allocation of this instance. Let $\cal I^{mod}$ be the instance $\cal I$ after removing a set of small goods $R$ and adding a set of large goods $A$; let $W^{mod}$ be the round robin allocation in this modified instance. If each good in $A$ has value at least $\frac{\mu}{\gamma}$, $|A| \le 0.01 \Delta$ and $v(R) \le 0.1\Delta\mu$, then
\begin{align*}
\logNSW(W^{mod}) \ge \logNSW(\cal I) -\frac1e+ \frac{1}{\Delta}\left (10|A| - 2|R| \right ).
\end{align*}
\end{theorem}
\begin{proof}
We note that $\cal I^{mod}$ might have a different set of large goods than $\cal I$. However, to avoid any confusion, we only refer to the large goods of $\cal I$ when we refer to the set of large goods $L$. This proof does not explicitly use the large goods of the instance $\cal I^{mod}$. In the same vein, when we use the value $\mu$, we mean the value $\frac{v(G \setminus L)}{\Delta}$, where $L$ is the set of large goods in the instance $\cal I$. In the allocation $W^{mod}$, we use $S^{mod}_i$ to denote for each agent $i$ the set of small goods in $W^{mod}_i$ --- again, note that these small goods are defined with respect to the original instance $\cal I$. 

We lower bound the Nash welfare of $W^{mod}$ the same way we did in Theorem \ref{thm:round-robin-nash-lower-bound}. Let $N_L$ be the set of agents who receive a large good or a good from $A$ in $W^{mod}$. Since each good in $A$ has value strictly greater than $\frac{\Delta\mu}{\Delta - \ell}$, we must have $N_L = \{1, \dots, \ell + |A|\}$. We let $N_S$ denote the remaining agents. This set of remaining agents must be non-empty since $\ell + |A| \le (1 - 1/e + \gamma)\Delta + 0.01\Delta \le 0.7\Delta$ (\ref{prop:a}). Within the set $N_S$, we use $N_1$ to denote the set of agents $i$ such that $v(W^{mod}_i) \ge \frac{\Delta \mu}{\Delta - \ell}$. 
We use $N_2$ to denote the set of all other agents in $N_S$. Note that all agents in $N_2$ have utility between $v(W^{mod}_{\Delta})$ and $\frac{\Delta \mu}{\Delta - \ell}$. 

For each $i \in N_2$ we define the value $\alpha_i \in [0, 1]$ as the value which satisfies $v(W^{mod}_i) = \alpha_i \frac{\Delta \mu}{\Delta - \ell} + (1 - \alpha_i)v(W^{mod}_{\Delta})$. We define $\alpha = \sum_{i \in N_2} \alpha_i$. Similar to Theorem \ref{thm:round-robin-nash-lower-bound}, we can lower bound the log Nash welfare of the allocation $W^{mod}$:

\begin{align*}
\logNSW(W^{mod}) &\ge \frac{1}{\Delta} \left [\sum_{g \in L \cup A}\ln{v(g)} + \sum_{i \in N_1}\ln\frac{\Delta \mu}{\Delta -\ell} + \sum_{i \in N_2} \ln v(W^{mod}_i) \right ].
\end{align*}

This can be simplified in the same way using the concavity of the log function and the definition of each agent part:

\begin{align*}
\logNSW(W^{mod}) &\ge \frac{1}{\Delta} \left [\sum_{g \in L}\ln{v(g)} + \sum_{g \in A}\ln{v(g)} + (|N_1| + \alpha)\ln\frac{\Delta \mu}{\Delta -\ell} + (|N_2| - \alpha) \ln v(W^{mod}_{\Delta}) \right ] \\
&\ge \frac{1}{\Delta} \left [\sum_{g \in L}\ln{v(g)} + |A|\ln{\frac{\mu}{\gamma}} + (|N_1| + \alpha)\ln\frac{\Delta \mu}{\Delta -\ell} + (|N_2| - \alpha) \ln v(W^{mod}_{\Delta}) \right ].
\end{align*}

By definition, $|N_1| + |N_2| + |A|= \Delta - \ell$. We can therefore re-arrange the above expression to bring it closer to the upper bound of $\logNSW(\cal I)$ in Lemma \ref{lem:nash-upper-bound}. 

\begin{align*}
\logNSW(W^{mod}) &\ge \frac1{\Delta} \left [\sum_{g \in L}\ln{v(g)} + (\Delta - \ell)\ln\frac{\Delta \mu}{\Delta -\ell} \right ] + \frac{|A|}{\Delta}\ln\frac{\Delta - \ell}{\Delta \gamma} + \frac{(|N_2| - \alpha)}{\Delta}\ln\frac{(\Delta - \ell)v(W^{mod}_{\Delta})}{\Delta\mu}.
\end{align*}

Using \ref{prop:a}, $\frac{\Delta - \ell}{\Delta \gamma} \ge 10^{9}$. This simplifies the above expression to:

\begin{align}
\logNSW(W^{mod}) &\ge \frac1{\Delta} \left [\sum_{g \in L}\ln{v(g)} + (\Delta - \ell)\ln\frac{\Delta \mu}{\Delta -\ell} \right ] + \frac{10|A|}{\Delta} + \frac{(|N_2| - \alpha)}{\Delta}\ln\frac{(\Delta - \ell)v(W^{mod}_{\Delta})}{\Delta\mu}. \label{eq:small-good-removal-large-good-addition-1}
\end{align}

The key to proving this theorem is analyzing the term $\frac{(|N_2| - \alpha)}{\Delta}\ln\frac{(\Delta - \ell)v(W^{mod}_{\Delta})}{\Delta\mu}$ in the above expression. We do this the same way as Theorem \ref{thm:round-robin-nash-lower-bound}. We require the following observations to make our analysis simpler.

\begin{obs}\label{obs:round-robin-worse-off}
For each $i \in [\Delta - \ell - |A|]$, $v(S^{mod}_{\ell + |A| + i}) \le v(S_{\ell + i})$. 
\end{obs}
\begin{proof}
The bundle $S^{mod}_{\ell + |A| + i}$ can be constructed using the round robin algorithm over the goods $G \setminus (L \cup R)$ with picking sequence $(\ell + |A| + 1, \dots, \Delta, 1, \dots, \ell + |A|)$. On the other hand, the bundle $S_{\ell + i}$ can be constructed using the round robin algorithm over the goods $G \setminus L$ with picking sequence $(\ell +  1, \dots, \Delta, 1, \dots, \ell)$. Since $G \setminus (L \cup R)$ is a subset of $G \setminus L$, the observation follows from the definition of the round robin algorithm.
\end{proof}

\begin{obs}\label{obs:x-bound}
$\frac{|N_2| -\alpha} {\Delta}\in \left [\frac16, \frac12 \right ]$.
\end{obs}
\begin{proof}
To prove this claim, we use the fact that the instance $\cal I$ is well-behaved. By construction $|N_1| + |N_2| = \Delta - \ell - |A|$. Via \ref{prop:a}, this implies, 
\begin{align}
\frac{|N_1| + |N_2|}{\Delta} \in \left [\frac1e - \gamma - \frac{|A|}{\Delta}, \frac1e + \gamma - \frac{|A|}{\Delta} \right ]. \label{eq:x-bound-1}
\end{align}

Via Property B, there are at most $\gamma \Delta$ agents who do not receive a large good and receive a utility greater than $\mu (1 + \gamma)$ in $W$. Using Observation \ref{obs:round-robin-worse-off}, this implies that at most $\gamma \Delta$ agents who do not receive a large good receive a utility of at least $\frac{\Delta \mu}{\Delta - \ell}$ in $W^{mod}$. This implies $|N_1| \le \gamma \Delta$. Plugging this into \eqref{eq:x-bound-1}, we get
\begin{align}
\frac{|N_2|}{\Delta} \in \left [\frac1e - 2\gamma - \frac{|A|}{\Delta}, \frac1e + \gamma - \frac{|A|}{\Delta} \right ]. \label{eq:x-bound-2}
\end{align}

This immediately gives us our upper bound of $\frac{|N_2| - \alpha}{\Delta} \le \frac12$.

Next, we upper bound $\alpha$ using the $\alpha_i$ values. Via Property B and Observation \ref{obs:round-robin-worse-off}, at most $\gamma \Delta$ agents in $N_2$ have utility greater than $(1 + \gamma)\mu$. These agents $i$ have $\alpha_i \le 1$. For all the other $|N_2| - \gamma\Delta$ agents, we have 
\begin{align*}
\alpha_i \le \frac{(1 + \gamma)\mu}{\frac{\Delta \mu}{\Delta - \ell}} \le \frac{(1 + \gamma)(\Delta - \ell)}{\Delta} \le \frac{(1 + \gamma)\Delta \left (\frac1e + \gamma \right )}{\Delta} \le 1/2.
\end{align*}

The third inequality uses \ref{prop:a} to bound the number of large goods. Using these upper bounds, we get $\alpha = \sum_{i \in N_2} \alpha_i \le \frac{|N_2|}{2} + \gamma \Delta$. Therefore, 
\begin{align*}
\frac{|N_2| - \alpha}{\Delta} \ge \frac{|N_2|}{2\Delta} - \gamma \ge \frac{1}{2e} - 2\gamma - 0.005 \ge 1/6.
\end{align*} 
In the final inequality, we use $|A| \le 0.01\Delta$. This completes the proof. 
\end{proof} 

\begin{obs}\label{obs:w-mod-delta-lower-bound}
$\frac{(\Delta - \ell)v(W^{mod}_{\Delta})}{\Delta \mu} \ge \left (\frac{|N_2| - \alpha}{\Delta} \right )e^{-\frac{2|R|}{|N_2|- \alpha}}$.
\end{obs}
\begin{proof}
We use the same idea as Lemma \ref{lem:w-delta-lower-bound} where we use $\sum_{i \in [\Delta]} v(S^{mod}_i) = v(G \setminus (L \cup R))$. 
\begin{align*}
\Delta \mu - v(R) &= v(G \setminus (L \cup R)) = \sum_{i \in [\Delta]} v(S^{mod}_i) \\
&= \sum_{i \in N_L} v(S^{mod}_i) + \sum_{i \in N_1} v(S^{mod}_i) + \sum_{i \in N_2} v(S^{mod}_i) 
\end{align*}
Each agent $i \in N_L$ satisfies $v(S^{mod}_i) \le v(W^{mod}_{\Delta})$ by the definition of the round robin allocation. Each agent $i \in N_1$ satisfies $v(W^{mod}_i) \le \max_{g \in W^{mod}_i} v(g) + v(W^{mod}_{\Delta})\le \frac{\Delta \mu}{\Delta - \ell} + v(W^{mod}_{\Delta})$ using Fact \ref{fact:round-robin-ef1}. Finally, we can bound $\sum_{i \in N_2} v(S^{mod}_i)$ using the $\alpha$ values. This simplifies the above equation to
\begin{align*}
\Delta \mu - v(R) &\le |N_L| v(W^{mod}_{\Delta}) + |N_1| \left (\frac{\Delta \mu}{\Delta - \ell} + v(W^{mod}_{\Delta}) \right ) + \alpha \frac{\Delta \mu}{\Delta - \ell} + (|N_2| - \alpha)v(W^{mod}_{\Delta}). \\
&\le \Delta v(W^{mod}_{\Delta}) + (|N_1| + \alpha)\frac{\Delta \mu}{\Delta - \ell}.
\end{align*}
Rearranging terms we get, 
\begin{align*}
\frac{(\Delta - \ell)v(W^{mod}_{\Delta})}{\Delta \mu} &\ge \frac{|N_2| - \alpha}{\Delta} \left ( 1  - \frac{(\Delta - \ell)v(R)}{\Delta\mu(|N_2| - \alpha)} \right ).
\end{align*}

By definition $v(R) \le 0.1\Delta\mu$, and from Observation \ref{obs:x-bound}, $|N_2| - \alpha \ge \frac{\Delta}{6}$. Moreover, $\ell \ge \frac{\Delta}{2}$ from \ref{prop:a}. Therefore, $\frac{(\Delta - \ell)v(R)}{\Delta\mu(|N_2| - \alpha)} \le 0.3$ and we can use Fact \ref{fact:identity-inequality} to obtain:
\begin{align*}
\frac{(\Delta - \ell)v(W^{mod}_{\Delta})}{\Delta \mu} \ge \left (\frac{|N_2| - \alpha}{\Delta} \right ) e^{\frac{-2(\Delta - \ell)v(R)}{\Delta\mu(|N_2| - \alpha)}}.
\end{align*}

We finally upper bound $v(R)$ using $|R|\frac{\Delta \mu}{\Delta - \ell}$ which holds because $R$ only has small goods. This gives us the following, which completes the proof.
\begin{equation*}
\frac{(\Delta - \ell)v(W^{mod}_{\Delta})}{\Delta \mu} \ge \left (\frac{|N_2| - \alpha}{\Delta} \right )e^{-\frac{2|R|}{|N_2|- \alpha}}. \qedhere
\end{equation*}
\end{proof}

We return to \eqref{eq:small-good-removal-large-good-addition-1} and simplify it using our observations. Specifically, we plug in Observation \ref{obs:w-mod-delta-lower-bound} and Lemma \ref{lem:nash-upper-bound} to get

\begin{align*}
\logNSW(W^{mod}) &\ge \logNSW(\cal I) + \frac{10|A|}{\Delta} + \frac{(|N_2| - \alpha)}{\Delta}\ln\frac{|N_2| - \alpha}{\Delta} - \frac{2|R|}{\Delta}. 
\end{align*}

This can be further simplified by noting that $\frac{(|N_2| - \alpha)}{\Delta}\ln\frac{|N_2| - \alpha}{\Delta} \ge -\frac1e$. This gives us our final result:
\begin{equation*}
\logNSW(W^{mod}) \ge \logNSW(\cal I) -\frac1e + \frac{10|A|}{\Delta}  - \frac{2|R|}{\Delta}. \qedhere
\end{equation*}
\end{proof}

We also have the following simple corollary of the above result.

\begin{corr}\label{cor:small-good-removal-large-good-addition}
Let $\cal I$ be a well-behaved instance with $\Delta$ identical agents and round robin allocation $W$. Let $\cal I^{mod}$ be the instance $\cal I$ after removing a set of small goods $R$ and adding a set of goods $A$; let $W^{mod}$ be the round robin allocation in this modified instance. Define the set $A' \subseteq A$ as the subset of goods of $A$ with value at least $\frac{\mu}{\gamma}$. If $|A'| \le 0.01 \Delta$ and $v(R) \le 0.1\Delta\mu$, then
\begin{align*}
\logNSW(W^{mod}) \ge \logNSW(\cal I) -\frac1e+ \frac{1}{\Delta}\left (10|A'| - 2|R| \right ).
\end{align*}\end{corr}
\begin{proof}
The corollary follows from Theorem \ref{thm:small-good-removal-large-good-addition} and the fact that adding the goods $A \setminus A'$ can only improve the Nash welfare of the allocation.
\end{proof}

We also prove the following three results. The proofs for these results have been relegated to Appendix \ref{apdx:add-remove}.
\begin{restatable}{theorem}{thmsmallgoodadditionlargegoodremoval}\label{thm:small-good-addition-large-good-removal}
Let $\cal I$ be a well-behaved instance with $\Delta$ identical agents, a set of large goods $L$ and mean value $\mu$. Let $\cal I^{mod}$ be the instance that results from making the following modifications to $\cal I$:
\begin{enumerate}
\item Replace the set of large goods $L$ with a new set of large goods $L'$ such that $v(g) > \frac{\Delta \mu}{\Delta - \ell}$ for each new large good $g \in L'$ and $|L'| \le \Delta$,
\item Add a set $A$ of small goods to $\cal I$ such that $v(g) \le \frac{\Delta \mu}{\Delta - \ell}$ for each $g \in A$, and $v(A) \le \Delta \mu$.
\end{enumerate}

Define $W^{mod}$ to be the round robin allocation of the instance $\cal I^{mod}$. The following holds:
\begin{align*}
\logNSW(W^{mod}) \ge \frac{1}{\Delta} \left [\sum_{g \in L'} \ln v(g) + (\Delta - |L'|)\ln \mu \right ] + \frac{(\Delta - |L'|)v(A)}{5\Delta^2 \mu} - 12\gamma.
\end{align*}
\end{restatable}

\begin{restatable}{theorem}{thmgeneraladdremoval}\label{thm:general-add-removal}
Let $\cal I$ be a well-behaved instance that satisfies Assumption \ref{ass:non-zero} with $\Delta$ identical agents, a set of large goods $L$ and mean value $\mu$. Let $\cal I^{mod}$ be the instance that results from making the following modifications to $\cal I$:
\begin{enumerate}
\item Replace the set of large goods $L$ with a new set of large goods $L'$ such that $v(g) > \frac{\Delta \mu}{\Delta - \ell}$ for each new large good $g \in L'$, and $|L'| \in [100\gamma\Delta, \Delta]$.
\item Remove a set $R$ of small goods from $\cal I$ such that $v(R) \le \frac{|L'| \mu}{100}$.
\item Add a set $A$ of small goods to $\cal I$ such that $v(g) \le \frac{\Delta \mu}{\Delta - \ell}$ for each $g \in A$, and $v(A) \le \Delta \mu$.
\end{enumerate}
Define $W^{mod}$ to be the round robin allocation in the instance $\cal I^{mod}$. The following holds:
\begin{align*}
\logNSW(W^{mod}) \ge \frac{1}{\Delta} \left [\sum_{g \in L'} \ln v(g) + (\Delta - |L'|)\ln \mu \right ] + \frac{(\Delta - |L'|)v(A)}{10\Delta^2 \mu} - \frac{3(\Delta - |L'|)v(R)}{\Delta^2 \mu} - 15\gamma.
\end{align*}
\end{restatable}

\begin{restatable}{theorem}{thmnonwellbehavedremoval}\label{thm:non-well-behaved-removal}
Let $\cal I$ be an instance with $\Delta$ identical agents that satisfies Assumption \ref{ass:non-zero}. Partition the goods into classes $G_1, G_2, \dots, G_{k'}$ such that all goods in each class have the same value. Let $G^{\Delta}$ be a set of $\Delta$ goods in the instance such that each good in $G^{\Delta}$ is weakly more valuable than each of the goods outside $G^{\Delta}$. Let $\cal I^{mod}$ be the instance $\cal I$ after removing a set $R$ of goods, such that $|R \cap G_k| \le \frac12 |G_{k} \setminus G^{\Delta}|$ for each $k \in [k']$. Let $W^{mod}$ be the round robin allocation in the instance $\cal I^{mod}$. The following holds:
\begin{align*}
\logNSW(W^{mod}) \ge \logNSW(\cal I) - \frac1e - 2.
\end{align*}
\end{restatable}

\section{Sub-optimal Allocations}\label{sec:suboptimal}
We have so far proved approximation guarantees for the round robin allocation. In this section, we prove an important property that can certify, in some sense, that an allocation is far from optimal. 

\begin{lemma}\label{lem:suboptimal}
Let $\cal I$ be a well-behaved instance that satisfies Assumption \ref{ass:non-zero} with $\Delta$ identical agents, a set of large goods $L$ and mean value $\mu$. Assume the small goods are partitioned into classes $G_1, G_2, \dots, G_k$ such that all goods in the same class have the same value and $|G_j| \le \Delta$ for each $j \in [k]$. Let $Y$ be an allocation in this instance with positive Nash welfare where all agents receive at most one good from each class. If $\sum_{j = 1}^k \mathbbm{1}\{|G_j| > 0.5\Delta\}v(G_j) \ge \gamma \Delta \mu$, then
\begin{align*}
\logNSW(Y) < \frac1{\Delta} \left [\sum_{g \in L}\ln v(g) + (\Delta - \ell)\ln \left (\frac{\Delta \mu}{\Delta - \ell} \right ) \right ]  - \gamma^2.
\end{align*}
\end{lemma}
\begin{proof}
To upper bound the Nash welfare of $Y$, we modify it slightly to make it easier to analyze. We first assume without loss of generality that $Y$ allocates all the goods. If not, we can allocate the unallocated goods arbitrarily in a way that satisfies the class constraints. This only improves the log Nash welfare of the allocation $Y$.

In the allocation $Y$, let $g^*_i$ denote the highest valued large good that agent $i$ is allocated (if $i$ is allocated a large good). 

If there is some agent who has two (or more) large goods in $Y$, we move one of the large goods to another agent in a way that increases the Nash welfare of $Y$. Specifically, consider an agent $i$ with at least two large goods. Let $\tilde{g}_i$ be the least valued large good in $Y_i$. Since $\tilde{g}_i$ is not the highest valued good of $i$, we have
\begin{align*}
\frac{v(Y_i \setminus \{\tilde{g}_i\})}{v(Y_i)} \ge \frac12.
\end{align*}

By the definition of large goods, this good $\tilde{g}_i$ has value strictly greater than $\frac{\Delta \mu}{\Delta - \ell}$. Since the total value of the small goods is $\Delta \mu$ and there are at least $\Delta - \ell$ agents who do not receive a large good, there is at least one agent who does not receive a large good and receives utility at most $\frac{\Delta \mu}{\Delta - \ell}$. Let $j$ denote one such agent. Adding $\tilde{g}_i$ to $Y_j$ at least doubles the utility of $j$. Therefore, 
\begin{align*}
\frac{v(Y_j \cup \{\tilde{g}_i\})}{v(Y_j)} > 2.
\end{align*}

Combining these two inequalities, we get that moving $\tilde{g}_i$ from agent $i$ to $j$ strictly improves the Nash welfare of $Y$, and strictly increases the number of agents who receive a large good. We can repeat this transfer whenever there is an agent who has two large goods. Let $Y'$ denote the resulting allocation where $\ell$ agents receive a large good. Note that $\NSW(Y') \ge \NSW(Y)$. Moreover, $Y'$ still allocates at most one good from each class to every agent since we only modified the allocations of the large goods.

Let $L^*$ be the set of large goods in the instance with value at least $\frac{\mu}{\gamma}$, and $N_{L^*}$ be the set of agents who receive goods in $L^*$. 
Since the instance is well-behaved, there are at least $(1 - 1/e - 2\gamma)\Delta$ large goods in $L^*$ (\ref{prop:a} and \ref{prop:d}). 
Let $G_1, \dots, G_t$ denote the set of classes of goods such that $|G_j| > 0.5\Delta$ for each $j \in [t]$. Via a simple counting argument, at least $0.1\Delta$ goods from $G_j$ for each $j \in [t]$ are allocated to the agents in $N_{L^*}$ in the allocation $Y'$. Therefore, the total value of small goods allocated to agents in $N_{L^*}$ is at least
\begin{align*}
\sum_{j = 1}^t 0.1\Delta \frac{v(G_j)}{|G_j|} \ge \sum_{j = 1}^t 0.1v(G_j) \ge 0.1\gamma v(G \setminus L). 
\end{align*}

For each agent $i$, let $s_i$ denote the total value that agent $i$ receives from small goods. From our discussion above, we have
\begin{align}
\sum_{i \in N_{L^*}} s_i \ge 0.1\gamma v(G \setminus L) = 0.1 \gamma \Delta \mu. \label{eq:suboptimal-0}
\end{align}

Let $N_L$ denote the set of agents who receive a large good in $Y'$. For any agent $i \in N_L$, let $g^*_i$ denote the large good allocated to agent $i$. Partitioning the agents in $N_{L^*}$, $N_{L} \setminus N_{L^*}$ and $[\Delta] \setminus N_L$, we can upper bound the log Nash welfare of $Y'$:
\begin{align}
\logNSW(Y') &= \frac1{\Delta} \left [ \sum_{i \in N_{L^*}}\ln(v(g^*_i) + s_i) + \sum_{i \in N_{L} \setminus N_{L^*}}\ln(v(g^*_i) + s_i) + \sum_{i \in [\Delta] \setminus N_L} \ln(s_i) \right ] \notag \\
&\le \frac1{\Delta} \left [ \sum_{i \in N_{L^*}}\ln \left (v(g^*_i) \left (1 + \frac{\gamma s_i}{\mu} \right ) \right ) + \sum_{i \in N_{L} \setminus N_{L^*}}\ln \left (v(g^*_i)\left (1 + \frac{(\Delta - \ell)s_i}{\Delta\mu} \right ) \right ) + \sum_{i \in [\Delta] \setminus N_L} \ln(s_i) \right ] \notag\\
&\le \frac1{\Delta} \left [ \sum_{g \in L}\ln v(g) + \sum_{i \in [\Delta] \setminus N_L} \ln(s_i) \right ] \notag \\
&\qquad \qquad + \frac{1}{\Delta} \sum_{i \in N_{L^*}}\ln \left (1 + \frac{\gamma s_i}{\mu} \right ) + \frac{1}{\Delta}\sum_{i \in N_{L} \setminus N_{L^*}}\ln \left (1 + \frac{(\Delta - \ell)s_i}{\Delta\mu} \right )\label{eq:suboptimal} 
\end{align}

In the first inequality, we use the lower bound on $v(g^*_i)$ given by the group each agent in $N_L$ belongs to. If $i \in N_{L^*}$, $v(g^*_i) \ge \frac{\mu}{\gamma}$, and if $i \in N_L$, $v(g^*_i) \ge \frac{\Delta\mu}{\Delta - \ell}$.  

Next, we upper bound $\sum_{i \in [\Delta] \setminus N_L} s_i$. We do this using the AM-GM inequality. 
\begin{align*}
\sum_{i \in [\Delta]\setminus N_L} \ln(s_i) &\le (\Delta - \ell)\ln \left ( \frac{\sum_{i \in [\Delta]\setminus N_L} s_i}{\Delta - \ell} \right ) \\
&\le (\Delta - \ell)\ln \left ( \frac{\Delta \mu - \sum_{i \in N_L} s_i}{\Delta - \ell} \right ) \\
&\le (\Delta - \ell)\ln \left (\frac{\Delta \mu}{\Delta - \ell}\right ) + (\Delta - \ell)\ln \left (1 - \frac{\sum_{i \in N_L} s_i}{\Delta \mu} \right ).
\end{align*}

Plugging this upper bound into \eqref{eq:suboptimal}, we get
\begin{align*}
\logNSW(Y') &\le \frac1{\Delta} \left [ \sum_{g \in L}\ln v(g) + (\Delta - \ell) \ln\left (\frac{\Delta\mu}{\Delta - \ell} \right) \right ]  \\
&\qquad + \frac{1}{\Delta} \sum_{i \in N_{L^*}}\ln \left (1 + \frac{\gamma s_i}{\mu} \right ) + \frac{1}{\Delta}\sum_{i \in N_{L} \setminus N_{L^*}}\ln \left (1 + \frac{(\Delta - \ell)s_i}{\Delta\mu} \right ) + \frac{(\Delta - \ell)}{\Delta}\ln \left (1 - \frac{\sum_{i \in N_L} s_i}{\Delta \mu} \right ).
\end{align*}

We now apply AM$\ge$GM once more to simplify the right hand side.
\begin{align*}
\logNSW(Y') &\le \frac1{\Delta} \left [ \sum_{g \in L}\ln v(g) + (\Delta - \ell) \ln\left (\frac{\Delta\mu}{\Delta - \ell} \right) \right ] \\
& \qquad \qquad + \ln \left (1 - \sum_{i \in N_{L^*}} s_i \left (\frac{\Delta - \ell}{\Delta^2 \mu} - \frac{\gamma}{\Delta\mu} \right )  \right ) \\
&\le \frac1{\Delta} \left [ \sum_{g \in L}\ln v(g) + (\Delta - \ell) \ln\left (\frac{\Delta\mu}{\Delta - \ell} \right) \right ] + \ln \left (1 - \sum_{i \in N_{L^*}} \frac{s_i}{3\Delta \mu}  \right ) \\
&\le \frac1{\Delta} \left [ \sum_{g \in L}\ln v(g) + (\Delta - \ell) \ln\left (\frac{\Delta\mu}{\Delta - \ell} \right) \right ] + \ln \left (1 - \frac{0.1\gamma\Delta\mu}{3\Delta \mu}  \right ) \\
&\le \frac1{\Delta} \left [ \sum_{g \in L}\ln v(g) + (\Delta - \ell) \ln\left (\frac{\Delta\mu}{\Delta - \ell} \right) \right ] -\gamma^2.
\end{align*}

The second inequality uses \ref{prop:a}. The third inequality uses \eqref{eq:suboptimal-0}. The final inequality uses $1 + x \le e^x$. Since $\NSW(Y') \ge \NSW(Y)$, this completes the proof.
\end{proof}

\section{Shmoys-Tardos Rounding}\label{sec:st-rounding}
In this section, we move to instances with non-identical agents. Specifically, we consider instances with $n$ agents, each with an additive valuation function $v_i$ over a set of $m$ goods $G = \{g_1, \dots, g_m\}$.
We describe the Shmoys-Tardos rounding method (referred to as the ST rounding algorithm), and present key tools we use to analyze this method. We then show that unless agents are well-behaved (for some appropriate definition of well-behavedness), the ST rounding algorithm outputs a better than $e^{1/e}$-approximation of the max Nash welfare. 

The ST algorithm takes as input a fractional allocation $\x$ and outputs an integral allocation $X$. This is done via the following bucketing procedure.
Fix an agent $i \in N$ and assume without loss of generality, $v_i(g_1) \ge v_i(g_2) \ge \dots \ge v_i(g_m)$. Define $k_i = \left \lceil \sum_{g \in G} x_{ig} \right \rceil$. We create $k_i$ buckets for agent $i$, denoted $b^1, \dots, b^{k_i}$ where each $b^t$ is defined by a vector $[0, 1]^G$. We divide the fractional bundle ${\x}_i$ into these $k_i$ buckets such that each $b^t$ gets a fractional bundle of size at most $1$. In this construction, $b^1$ gets the highest valued fraction of goods (according to agent $i$) of size at most one. Then, of the remaining fractional bundle $\x_i$ not in $b^1$, $b^2$ gets the highest valued fraction of goods of size at most $1$, and so on. 

\begin{figure}[ht]
\centering
\begin{tikzpicture}

\def\h{0.8}

\draw[thick, fill=gray!35] (0,0)     rectangle (2.4,\h);   
\draw[thick, fill=gray!55] (2.4,0)   rectangle (3.2,\h);   
\draw[thick, fill=gray!25] (3.2,0)   rectangle (5.2,\h);   
\draw[thick, fill=gray!60] (5.2,0)   rectangle (6.8,\h);   
\draw[thick, fill=gray!40] (6.8,0)   rectangle (8.0,\h);   
\draw[thick, fill=gray!50] (8.0,0)   rectangle (9.2,\h);   
\draw[thick, fill=gray!30] (9.2,0)   rectangle (9.8,\h);   
\draw[thick, fill=gray!65] (9.8,0)   rectangle (10.4,\h);  

\node at (1.2, \h/2) {$g_1$};
\node at (2.8, \h/2) {$g_2$};

\draw[thick] (0,0) -- (12.0,0);

\draw[line width=1.6pt] (0,-0.2)    -- (0,\h+0.15);
\draw[line width=1.6pt] (4.0,-0.2)  -- (4.0,\h+0.15);
\draw[line width=1.6pt] (8.0,-0.2)  -- (8.0,\h+0.15);
\draw[line width=1.6pt] (12.0,-0.2) -- (12.0,\h+0.15);

\node[anchor=north] at (4.0,-0.22)  {$1$};
\node[anchor=north] at (8.0,-0.22)  {$2$};
\node[anchor=north] at (12.0,-0.22) {$3$};

\draw[<->,>=latex,thick] (0,\h+0.55)   -- (2.4,\h+0.55)
      node[midway,above]{$x_{ig_1}$};
\draw[<->,>=latex,thick] (2.4,\h+0.55) -- (3.2,\h+0.55)
      node[midway,above]{$x_{ig_2}$};
\draw[<->,>=latex,thick] (3.2,\h+0.55) -- (5.2,\h+0.55)
      node[midway,above]{$x_{ig_3}$};

\node at (6.0,\h+0.7) {$\cdots$};

\draw[decorate,decoration={brace,mirror,amplitude=8pt},thick]
      (0,-0.85)   -- (4.0,-0.85)
      node[midway,below=10pt]{\large $b^1$};
\draw[decorate,decoration={brace,mirror,amplitude=8pt},thick]
      (4.0,-0.85) -- (8.0,-0.85)
      node[midway,below=10pt]{\large $b^2$};
\draw[decorate,decoration={brace,mirror,amplitude=8pt},thick]
      (8.0,-0.85) -- (12.0,-0.85)
      node[midway,below=10pt]{\large $b^3$};

\end{tikzpicture}
\caption{Bucketing construction for agent $i$. The fractional bundle $x_i$ is divided into $k_i = 3$ buckets, each of size at most $1$.}
\label{fig:bucketing}
\end{figure}
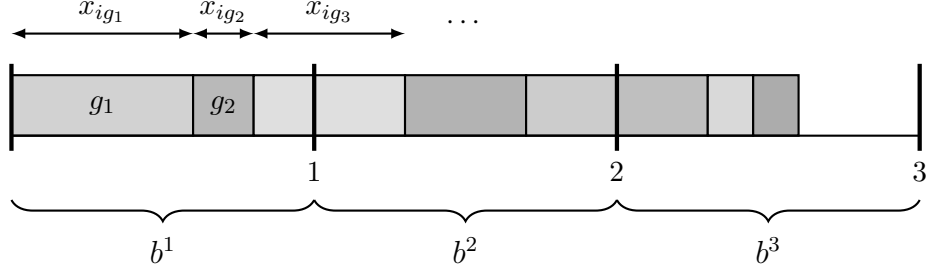

This is illustrated in Figure \ref{fig:bucketing}. Formally, the bucketing is uniquely defined by the following properties \cite{Feng2025weightednashadditive}:
 
\begin{description}
\item[(P1)] For every $t \in [k_i - 1]$, we have $|b^t|_1 = 1$; for $t = k_i$, $|b^t|_1 = \sum_{g \in G} x_{ig} - k_i + 1$.
\item[(P2)] For each $g \in G$, $\sum_{t \in [k_i]} b^t_g = x_{ig}$.
\item[(P3)] For every two goods $g_j, g_{j'}$ such that $j < j'$, and $t, t' \in [k_i]$ such that $t < t'$, it cannot be that $b^{t'}_{g_j} > 0$ and $b^{t}_{g_{j'}} > 0$.
\end{description}

We construct these buckets for each agent $i \in N$. Let $B_i = \{b^1, \dots, b^{k_i}\}$ be the set of buckets constructed for agent $i$, and let $B = \bigcup_{i \in N} B_i$ be the set of buckets constructed for all agents $i \in N$. We use these buckets to construct an edge weighted bipartite graph $\cal G(\x)$ over the set of nodes $B \cup G$. For each bucket $b \in B$ and good $g \in G$, we add an edge of weight $b_g$ between $b$ and $g$. From this graph $\cal G(\x)$, we compute an integral matching $M_{\cal G(\x)}$ between $B$ and $G$ in a way that satisfies the following properties:
\begin{description}
\item[Marginals] For every bucket $b$ and good $g$, $\Pr[g \text{ is matched to } b] = b_g$. 
\item[Cardinality] For every bucket $b$ such that $|b|_1 = 1$, $b$ is matched in the integral matching with probability $1$.  
\item[Rationality] If the edge weights are rational, the probability of each integral matching is a rational value.
\end{description}

An integral matching that satisfies the above properties can be computed using standard techniques, such as \cite{birkhoff1946bvndecomposition} and \cite{gandhi2006dependent}. The integral matching $M_{\cal G(\x)}$ implies a natural allocation $X$. The allocation for agent $i$ (the bundle $X_i$) is the set of goods matched to the buckets of agent $i$; that is, the set of goods matched to the buckets in $B_i$ in the integral matching $M_{\cal G(\x)}$.

Our goal in this section is to analyze the Nash welfare of the allocation $X$. To make this easier, we define $\ST(\x, i)$ to be the value $\E[\ln(v_i(X_i))]$; that is, the expected log of the utility of agent $i$ obtained by rounding fractional allocation $\x$ using the ST rounding algorithm. We also define $\ST(\x) = \frac1n \sum_{i \in N}\E[\ln( v_i(X_i))]$ to be the expected log Nash welfare of the final allocation $X$.

\subsection{Worst Case Bundles}
In this subsection, we present a technique that connects $\ST(\x, i)$ to round robin allocations when agents have identical valuations. This will allow us to lower bound $\ST(\x, i)$ using the techniques developed in Section \ref{sec:round-robin-guarantees}.

\begin{definition}[Feasible Integers]
We say a positive integer $\Delta$ is feasible with respect to a rational configuration LP solution $y$ if $\Delta y_{i, S}$ is an integer for each $i \in N$ and $S \subseteq G$. Somewhat similarly, we say that a positive integer $\Delta$ is feasible with respect to a rational fractional allocation $\x$ if 
\begin{inparaenum}[(a)]
\item $\Delta b_g$ is an integer for each bucket $b \in B$ and $g \in G$, and
\item $\Delta \Pr[M_{\cal G(\x)} = M']$ is an integer for each matching $M'$ in the support of $M_{\cal G(\x)}$.
\end{inparaenum}
\end{definition}

A feasible $\Delta$ is guaranteed to exist for a configuration LP solution $y$ if $y$ is rational. A feasible $\Delta$ is guaranteed to exist for the fractional allocation $\x$ if $\x$ is rational and we pick a rounding algorithm that satisfies the Rationality property. Throughout this paper, we only work with rational configuration LP solutions, rational fractional allocations and rational rounding methods. 
Therefore, a feasible $\Delta$ will always exist. 

Fix an agent $i \in N$. Given a fractional allocation $\x$, we construct an instance with $\Delta$ agents (for any feasible $\Delta$ w.r.t. $\x$). For each bucket $b \in B_i$, we construct a set of goods $G_b$ that contains $\Delta b_g$ copies of each good $g$. The set of goods of the instance we construct is $G' := \bigcup_{b \in B_i} G_b$.  Essentially, for each $g \in G$, the set of goods $G'$ has $\Delta x_{ig}$ copies of $g$. All agents in this instance have the valuation function $v_i$. We refer to this instance as the {\em artificial instance} for agent $i$ with $\Delta$ agents, denoted $\cal I^{i, \x}(\Delta)$. We use $[\Delta]$ to denote the set of agents of this instance and $G' = \{g'_1, g'_2, \dots, g'_{m'}\}$ to denote the set of goods with the assumption that $v_i(g'_1) \ge \dots \ge v_i(g'_{m'})$. Since all goods from $G_{b^1}$ are weakly more valuable than the goods in $G_{b^2}$ and so on, we assume $G_{b^1} = \{g'_1, \dots, g'_{\Delta}\}$ and $G_{b^2} = \{g'_{\Delta + 1}, \dots, g'_{2\Delta}\}$, and so on.

Our main result in this subsection is the following. 

\begin{lemma}\label{lem:w-worst}
Let $\x$ be a fractional allocation. Let $\cal I^{i, \x}(\Delta)$ be the artificial instance of agent $i$ for some $\Delta$ that is feasible with respect to $\x$. Let $W$ be the round robin allocation of this instance $\cal I^{i, \x}(\Delta)$. The following holds: 
$$\ST(\x, i) \ge \logNSW(W).$$
\end{lemma}
\begin{proof}
Let $M_{\cal G(\x)}$ be the integral matching constructed by the ST rounding algorithm applied to allocation $\x$. For each matching $M'$ in the support of $M_{\cal G(\x)}$, we define a bundle $S(M') \subseteq G'$, where $G'$ is the set of goods of $\cal I^{i, \x}(\Delta)$. The bundle $S(M')$ consists of one copy of the good $g$ matched to bucket $b$ (in $M'$) from the set $G_{b}$, for each bucket $b \in B_i$. 

We define an allocation $Z$ in the instance $\cal I^{i, \x}(\Delta)$ as one where $\Delta \Pr[M_{\cal G(\x)} = M']$ agents receive the bundle $S(M')$, for each $M'$ in the support of $M_{\cal G(\x)}$. This is a valid allocation because the number of copies of good $g$ from set $G_b$ it allocates is exactly $\Delta \Pr[g \text{ is matched to } b] = \Delta b_g$ for each good $g \in G$ and bucket $b \in B_i$. Using the Cardinality property of the allocation $X$, we know that each of the first $k_i - 1$ buckets is matched to a good. Therefore, each bundle in the allocation $Z$ must have exactly one good from $G_{b^1}, \dots G_{b^{k_i - 1}}$, and at most one good from the set $G_{b^{k_i}}$.

The proof has two steps. We first show that $\logNSW(Z) \ge \logNSW(W)$. Then, we show that $\logNSW(Z) = \E[\ln(v_i(X_i))]$. 

Assume that $Z$ and $W$ are not identical allocations. We start at the allocation $Z$ and change it using a series of steps till we reach the allocation $W$. We show that each of these steps weakly reduces the Nash welfare of the allocation. This shows that the Nash welfare of $W$ is upper bounded by the Nash welfare of $Z$. To make the rest of this proof more readable, we say good $g$ is from bucket $b$ instead of saying good $g$ is from the set $G_b$.

Let $g'_k$ be the good with the smallest $k$ such that it is allocated to different agents in $Z$ and $W$.
Let $k = (t-1)\Delta + j$ for positive integer $t$ and positive integer $j$ with $j \le \Delta$; this means $g'_k$ is from the $t$-th bucket $b^t$ and is allocated to agent $j$ in $W$. Assume $g'_k$ is allocated to agent $j'$ in $Z$. 

If $j' < j$, then the good that $Z_{j'}$ receives from bucket $b^t$ is $g'_k$. Since $Z_{j'}$ receives at most one good from bucket $b^t$, this implies that $g'_{(t-1)\Delta + j'}$ is allocated to different agents in both $W$ and $Z$, contradicting our assumption about $k$. Therefore $j' > j$. 

Let the set of goods in $Z_j$ from bucket $b^t$ be $S_j$. Note that $|S_j| \le 1$; $S_j$ could be empty if $t = k_i$. Since all goods $\{g'_1, \dots, g'_{k-1}\}$ are allocated to the same agents in $Z$ and $W$, it must be that $g'_k$ is more valuable than the singleton good in $S_j$ (if one exists); specifically, we have $v_i(g'_k) \ge v_i(S_j)$. 

We define $Z^{<t}_j$ as the set of goods in $Z_j$ from the buckets with index $< t$. We similarly define $Z^{> t}_j, W^{< t}_j$ and $W^{>t}_j$. Note that $Z^{<t}_j$ and $W^{<t}_j$ are identical by assumption and $v_i(W^{<t}_j) \ge v_i(W^{<t}_{j'})$ since $j' > j$. We have two cases in our proof. 

\noindent \textbf{Case 1: } $v_i(Z_{j}^{> t}) \ge v_i(Z_{j'}^{> t})$. We convert $Z$ into a new allocation $\widetilde{Z}$ by swapping $g'_k$ and $S_j$; that is, we allocate $g'_k$ to agent $j$ and $S_j$ to agent $j'$ leaving everything else unchanged. This swap reduces the Nash welfare of allocation $Z$ due to Lemma \ref{lem:nash-welfare-swap}. Additionally, $\widetilde{Z}$ and $W$ do not differ in the allocations of the first $k$ goods $g'_1, \dots, g'_k$.

\noindent \textbf{Case 2: } $v_i(Z_{j}^{> t}) < v_i(Z_{j'}^{> t})$. We convert $Z$ into a new allocation $\widetilde{Z}$ by swapping $Z_{j'}^{> t} \cup \{g'_k\}$ and $Z_{j}^{> t} \cup S_j$; that is, we allocate $Z_{j'}^{> t} \cup \{g'_k\}$ to agent $j$ and $Z_{j}^{> t} \cup S_j$ to agent $j'$ leaving the rest of the allocation unchanged. This swap reduces the Nash welfare of allocation $Z$ due to Lemma \ref{lem:nash-welfare-swap}. Once again, $\widetilde{Z}$ and $W$ do not differ in the allocations of the first $k$ goods $g'_1, \dots, g'_k$.

What we have ensured above is that given an allocation $Z$ that agrees with $W$ on the allocation of the first $k-1$ goods, we can convert it into an allocation $\widetilde{Z}$ that agrees with $W$ on the allocation of the first $k$ goods and has weakly lower Nash welfare. Moreover, $\widetilde{Z}$ allocates to each agent, exactly one good from the first $k_i - 1$ buckets and at most one good from the last bucket. Repeating this step at most $m'$ times, we reach the allocation $W$. This proves that $W$ has a lower Nash welfare than $Z$.

Finally, we recall the definition of $Z$ as an allocation where $\Delta \Pr[M_{\cal G(\x)} = M']$ agents receive the bundle $S(M')$ for each $M'$ in the support of $M_{\cal G(\x)}$. Assuming $X$ is the final allocation output by ST rounding, we have
\begin{align*}
\logNSW(Z) = \frac1{\Delta}\sum_{j \in [\Delta]}\ln(v_i(Z_j)) &= \frac{1}{\Delta} \sum_{M' \in supp(M_{\cal G(\x)})} \Delta \Pr[M_{\cal G(\x)} = M']\ln(v_i(S(M'))) \\
&= \E[\ln(v_i(X_i))] = \ST(\x, i). 
\end{align*}
The penultimate equality holds because $S(M')$ is exactly $X_i$ when $M' = M_{\cal G(\x)}$ (more pedantically, $S(M')$ contains one copy of each good in $X_i$ when $M' = M_{\cal G(\x)}$). 
\end{proof}

\subsection{Configuration LP solutions and their Corresponding Fractional Allocations}
In the previous section, we established a connection between $\ST(\x, i)$ and the round robin allocation in an instance with identical valuations. In this section, we will connect the configuration LP solution $y$ to the same instance with identical valuations.
Let $y$ be a rational configuration LP solution with a well-defined objective value. Using this solution, we define a fractional allocation $\x$ as follows:
\begin{align}
 x_{ig} = \sum_{S: g \in S} y_{i, S}.\label{eq:fractional-allocation-definition}
\end{align}
We refer to the fractional allocation $\x$ constructed using configuration LP solution $y$ as the {\em corresponding fractional allocation} of $y$. Every time we mention a configuration LP solution $y$, we mean a rational configuration LP solution with a well-defined objective value. Such a solution is guaranteed to exist due to Remark \ref{rem:well-defined-objective}. It can also be found in polynomial time using Theorem \ref{thm:configuration-lp}.

We prove three results. The first one is a simple observation about the fractional allocation $\x$, the second one shows that artificial instances always satisfy Assumption \ref{ass:non-zero} and the third one is a result from \cite{Feng2025weightednashadditive}, which we include for completeness.

We start by showing that all goods are completely allocated in the fractional allocation $\x$. 
\begin{obs}\label{obs:x-complete}
Let $y$ be a configuration LP solution and $\x$ be its corresponding fractional allocation. For all goods $g \in G$, $\sum_{i \in N} x_{ig} = 1$. 
\end{obs}
\begin{proof}
This follows from the condition \eqref{eq:item-capacity} that configuration LP solutions satisfy. This condition states that $\sum_{i \in N}\sum_{S\subseteq G: g \in S} y_{i, S} = 1$ for all $g \in G$. The observation follows from noting, by the definition of $\x$, that $\sum_{i \in N} x_{ig} = \sum_{i \in N}\sum_{S\subseteq G: g \in S} y_{i, S} = 1$.
\end{proof}

Next, we show that that artificial instances $\cal I^{i, \x}(\Delta)$ satisfy Assumption \ref{ass:non-zero} when $\x$ is the corresponding fractional allocation of some configuration LP solution. This allows us to apply the results from Section \ref{sec:round-robin-guarantees} to this instance. 

\begin{lemma}
Let $y$ be a rational configuration LP solution and $\x$ be the corresponding fractional allocation. Let $\cal I^{i, \x}(\Delta)$ be the artificial instance of agent $i$ for some $\Delta$ that is feasible with respect to $\x$. The instance $\cal I^{i, \x}(\Delta)$ satisfies Assumption \ref{ass:non-zero}.
\end{lemma}
\begin{proof}
We assume that $y$ has a well-defined objective value. This means for every $S \subseteq G$ such that $y_{i, S} > 0$, the value $v_i(S)$ is positive. This implies there is at least one good in $S$ with positive value for agent $i$. Let $G^+$ be the set of goods that agent $i$ has positive value for. By our discussion, each set $S$ with $y_{i, S} > 0$ has at least one good from $G^+$. Therefore,
\begin{align*}
\sum_{g \in G^+} x_{ig} = \sum_{g \in G^+} \sum_{S \subseteq G: g \in S} y_{i, S} \ge \sum_{S \subseteq G} y_{i, S} = 1.
\end{align*}

The final equality holds since $y$ satisfies \eqref{eq:agent-cover}.
Therefore, the number of goods with positive value in $\cal I^{i, \x}(\Delta)$ is at least $\Delta \sum_{g \in G^+} x_{ig} \ge \Delta$. This means that Assumption \ref{ass:non-zero} is satisfied.
\end{proof}

Finally, we present a result from \cite{Feng2025weightednashadditive} that lower bounds the log Nash welfare of $\cal I^{i, \x}(\Delta)$.
\begin{lemma}[\cite{Feng2025weightednashadditive}]\label{lem:y-substitute}
Let $y$ be a rational configuration LP solution and $\x$ be the corresponding fractional allocation. Let $\cal I^{i, \x}(\Delta)$ be the artificial instance of agent $i$ for some $\Delta$ that is feasible with respect to both $\x$ and $y$. Then the following holds:
\begin{align*}
\logNSW(\cal I^{i, \x}(\Delta)) \ge \sum_{S \subseteq G} y_{i, S} \ln{v_i(S)}
\end{align*}
\end{lemma}
\begin{proof}
In the instance $\cal I^{i, \x}(\Delta)$, we define the allocation $Y$ as one where $\Delta y_{i, S}$ agents receive the bundle $S$ for each $S\subseteq G$. This is a valid allocation since the number of copies of good $g$ it allocates is exactly $\sum_{S \subseteq G: g \in S} \Delta y_{i, S} = \Delta x_{i, g}$ for each good $g \in G$. This implies
\begin{align*}
\logNSW(\cal I^{i, \x}(\Delta)) \ge \logNSW(Y).
\end{align*}
The proof is completed by using the definition of $Y$ to measure its log Nash welfare:
\begin{equation*}
\logNSW(Y) = \frac1{\Delta} \sum_{j \in [\Delta]} \ln{v_i(Y_j)} = \frac1{\Delta} \sum_{S \subseteq G} \Delta y_{i, S} \ln{v_i(S)} = \sum_{S \subseteq G} y_{i, S} \ln{v_i(S)}. \qedhere
\end{equation*}
\end{proof}

\begin{remark}\label{rem:fl}
As \cite{Feng2025weightednashadditive} show, rounding the corresponding fractional allocation $\x$ of configuration LP solution $y$ results in an $e^{1/e}$-approximation of the max Nash welfare. This is because for each agent $i$, 
\begin{align*}
\ST(\x, i) \ge \logNSW(W) \ge \logNSW(\cal I^{i, \x}(\Delta)) - \frac1e \ge \sum_{S \subseteq G} y_{i, S} \ln v_i(S) - \frac1e.
\end{align*}
The first inequality follows from Lemma \ref{lem:w-worst}. The second inequality follows from Lemma \ref{lem:barman-eonebye}. The final inequality follows from Lemma \ref{lem:y-substitute}.
\end{remark}

\subsection{Well-behaved Agents}
In the previous subsections, we established a connection between the round robin allocation of an instance with identical valuations and the approximation guarantees of the ST rounding algorithm. To exploit our results and improve over $e^{1/e}$, we define a notion of well-behavedness for agents as we did for instances with identical valuations. In this section, we define what it means for an agent $i$ to be well-behaved in a fractional allocation $\x$. 
We emphasize here that while these definitions hold for any fractional allocation $\x$, they will only be applied to corresponding fractional allocations of some configuration LP solution.
We start by defining a set of large goods and a mean value for fractional bundles. 

\begin{definition}[Large and Small Goods in Fractional Bundles]\label{def:large-goods-fractional}
Given a fractional allocation $\x$, we define the set of large goods $L_i$ for agent $i$ as the set of goods which satisfies
\begin{enumerate}[(i)]
\item for each $g \in L_i$, $x_{ig} > 0$ and $v_i(g) > \frac{v_i(\x_i) - \sum_{g \in L_i} x_{ig} v_i(g)}{1 - \sum_{g \in L_i} x_{ig}}$, 
\item for each $g \in G \setminus L_i$, $x_{ig} = 0$ or $v_i(g) \le \frac{v_i(\x_i) - \sum_{g \in L_i} x_{ig} v_i(g)}{1 - \sum_{g \in L_i} x_{ig}}$, and
\item $\sum_{g \in L_i} x_{ig} < 1$.
\end{enumerate}
We refer to all the goods in $G \setminus L_i$ as small goods for the agent $i$.
\end{definition}

\begin{definition}[Mean Value in Fractional Bundles]\label{def:mean-fractional}
Given a fractional allocation $\x$, the mean value of the agent $i$, denoted $\mu^i$, is defined as the value $\sum_{g \in G \setminus L_i} x_{ig} v_i(g)$.
\end{definition}

\begin{definition}[Large Good Amount in Fractional Bundles]
Given a fractional allocation $\x$, the amount of large goods allocated to agent $i$, denoted $\ell_i$, is defined as the value $\sum_{g \in L_i} x_{ig}$.
\end{definition}

We can prove a uniqueness lemma similar to Observation \ref{obs:large-goods-unique}.

\begin{lemma}\label{lem:large-goods-unique-fractional}
For any fractional allocation $\x$ and agent $i$, there is a unique set of large goods $L_i$ satisfying the properties of Definition \ref{def:large-goods-fractional}. 
Moreover, this set of large goods can be found in polynomial time.
\end{lemma}
\begin{proof}
Fix an agent $i$. We first construct a set $L_i$ satisfying the properties of Definition \ref{def:large-goods-fractional}. This set can be constructed iteratively. We start by initializing $L_i$ as the empty set, and while there is some good $g' \in G \setminus L_i$ such that $x_{ig'} > 0$ and $v_i(g') > \frac{v_i(\x_i) - \sum_{g \in L_i} x_{ig} v_i(g)}{1 - \sum_{g \in L_i} x_{ig}}$, we add this $g'$ to $L_i$. It is easy to see that this process always terminates and it does so in time polynomial in the number of goods $m$. Additionally, by definition, it satisfies property (ii) of Definition \ref{def:large-goods-fractional} when it terminates.

We show that this process satisfies property (iii) next. If for some $g'$, we have $x_{ig'} \ge 1 - \sum_{g \in L_i} x_{ig}$, then
\begin{align*}
\frac{v_i(\x_i) - \sum_{g \in L_i} x_{ig} v_i(g)}{1 - \sum_{g \in L_i} x_{ig}} \ge \frac{x_{ig'} v_i(g')}{1 - \sum_{g \in L_i} x_{ig}} \ge \frac{x_{ig'} v_i(g')}{x_{ig'}} = v_i(g').
\end{align*}
This implies $g'$ will never get added to $L_i$ which implies that our addition process will never violate property (iii) of Definition \ref{def:large-goods-fractional}.

Moreover, the addition of $g'$ to $L_i$ reduces the value of $\frac{v_i(\x_i) - \sum_{g \in L_i} x_{ig} v_i(g)}{1 - \sum_{g \in L_i} x_{ig}}$ due to the following inequalities:

\begin{align*}
\frac{v_i(\x_i) - \sum_{g \in L_i} x_{ig} v_i(g) - x_{ig'}v_i(g')}{1 - \sum_{g \in L_i} x_{ig} - x_{ig'}} 
&< \frac{v_i(\x_i) - \sum_{g \in L_i} x_{ig} v_i(g) - x_{ig'}\frac{v_i(\x_i) - \sum_{g \in L_i} x_{ig} v_i(g)}{1 - \sum_{g \in L_i} x_{ig}}}{1 - \sum_{g \in L_i} x_{ig} - x_{ig'}} 
\\
&= \frac{v_i(\x_i) - \sum_{g \in L_i} x_{ig} v_i(g)}{1 - \sum_{g \in L_i} x_{ig}}.
\end{align*}
Therefore adding $g'$ to $L_i$ does not violate the largeness of the goods already in $L_i$. In conclusion, we have presented a polynomial time algorithm which outputs a set of large goods $L_i$ satisfying Definition \ref{def:large-goods-fractional}. 

We show uniqueness next. Assume there are two different sets of large goods $L_i$ and $L'_i$ satisfying the properties of Definition \ref{def:large-goods-fractional}. Note that if any good $g$ is in $L_i$ (resp. $L'_i$), all goods $g'$ with $x_{ig'} > 0$ and value at least $v_i(g)$ must also be in $L_i$ (resp. $L'_i$). Therefore, one of the two sets $L_i$ and $L'_i$ must be contained in the other. Assume $L_i \subset L'_i$. By applying the properties of large goods with the set $L_i$, we get for any $g \in L'_i \setminus L_i$, 
\begin{align*}
v_i(g) \le \frac{v_i(\x_i) - \sum_{g \in L_i} x_{ig} v_i(g)}{1 - \sum_{g \in L_i} x_{ig}} = \frac{v_i(\x_i) - \sum_{g \in L'_i} x_{ig} v_i(g) + \sum_{g \in L'_i \setminus L_i} x_{ig} v_i(g)}{1 - \sum_{g \in L_i} x_{ig}}.
\end{align*}

Adding the above inequality up for all $g \in L'_i \setminus L_i$, we get
\begin{align*}
\sum_{g \in L'_i \setminus L_i} x_{ig} v_i(g) \le \left (\frac{\sum_{g \in L'_i \setminus L_i} x_{ig}}{1 - \sum_{g \in L_i}x_{ig}} \right ) \left (v_i(\x_i) - \sum_{g \in L'_i} x_{ig} v_i(g) + \sum_{g \in L'_i \setminus L_i} x_{ig} v_i(g) \right ).
\end{align*}

Re-arranging terms, we get
\begin{align*}
\frac{\sum_{g \in L'_i \setminus L_i} x_{ig} v_i(g)}{ \sum_{g \in L'_i \setminus L_i} x_{ig}} \le \frac{ \left ( v_i(\x_i) - \sum_{g \in L'_i} x_{ig} v_i(g) \right )}{1 - \sum_{g \in L'_i}x_{ig}}.
\end{align*}

This contradicts the set $L'_i$ being a set of large goods since each good $g \in L'_i\setminus L_i$ satisfies $x_{ig} > 0$ and
\begin{align*}
v_i(g) > \frac{\left ( v_i(\x_i) - \sum_{g \in L'_i} x_{ig} v_i(g) \right )}{1 - \sum_{g \in L'_i}x_{ig}},
\end{align*}
by the definition of being a large good. Therefore, the set of large goods must be unique and the proof is complete.
\end{proof}

With these definitions set up, we show an important connection between the large goods $L_i$, mean value $\mu^i$ and amount of large goods $\ell_i$ of the fractional allocation $\x$ and the set of large goods in the artificial instance $\cal I^{i, \x}(\Delta)$. 

\begin{lemma}\label{lem:large-good-connection}
Let $\x$ be a fractional allocation, and let $\cal I^{i, \x}(\Delta)$ denote the artificial instance for agent $i$ defined for any feasible $\Delta$ with respect to $\x$. The set of large goods of $\cal I^{i, \x}(\Delta)$ is the set of all copies of the large goods $L_i$ of the agent $i$ in the fractional allocation $\x$. 

Additionally, the mean value of the instance $\cal I^{i, \x}(\Delta)$ is equal to $\mu^i$ and the number of large goods of $\cal I^{i, \x}(\Delta)$ is equal to $\Delta\ell_i$.
\end{lemma}
\begin{proof}
Let $L_i$ be the set of large goods for agent $i$ in the allocation $\x$. The instance $\cal I^{i, \x}(\Delta)$ is constructed by creating $\Delta x_{ig}$ copies of each good $g$. We use $G^i$ to denote the set of goods in instance $\cal I^{i, \x}(\Delta)$.

Let $L$ (with $|L| = \ell$) be the set of all copies of each good in $L_i$. From the definition of $\cal I^{i, \x}(\Delta)$, $v_i(G^i) = \sum_{g \in G}\Delta x_{ig} v_i(g)$, $v_i(L) = \Delta \sum_{g \in L_i} x_{ig} v_i(g)$ and $\ell = \Delta \ell_i$.

We show that this set $L$ satisfies Definition \ref{def:large-goods}. For each good $g \in L_i$, we have by the properties of $L_i$, 

\begin{align*}
v_i(g) > \frac{v_i(\x_i) - \sum_{g \in L_i} x_{ig} v_i(g)}{1 - \sum_{g \in L_i} x_{ig}} = \frac{v_i(G^i \setminus L)}{\Delta - \ell}.
\end{align*}
This shows that all copies of $g$ in the instance $\cal I^{i, \x}(\Delta)$ satisfy the first property of Definition \ref{def:large-goods}. To prove the second property, consider any $g \notin L_i$ with $x_{ig} > 0$. We have
\begin{align*}
v_i(g) \le \frac{v_i(\x_i) - \sum_{g \in L_i} x_{ig} v_i(g)}{1 - \sum_{g \in L_i} x_{ig}} = \frac{v_i(G^i \setminus L)}{\Delta - \ell},
\end{align*}
which shows that all copies of $g$ in the instance $\cal I^{i, \x}(\Delta)$ satisfy the second property of Definition \ref{def:large-goods}. Since the set of large goods in the instance $\cal I^{i, \x}(\Delta)$ is unique (Observation \ref{obs:large-goods-unique}), this proves that $L$ must be the set of large goods of the instance $\cal I^{i, \x}(\Delta)$. This implies the number of large goods is equal to $\Delta \ell_i$.

Finally, we compute the mean value $\mu = \frac{v_i(G^i \setminus L)}{\Delta}$. From the above discussion, this is equal to
\begin{equation*}
\mu = \frac{v_i(G^i \setminus L)}{\Delta} = \frac{\Delta \left (\sum_{g \in G \setminus L_i} x_{ig} v_i(g) \right )}{\Delta} = \sum_{g \in G \setminus L_i} x_{ig} v_i(g) = \mu^i. \qedhere
\end{equation*}
\end{proof}

Next, we show that if an artificial instance $\cal I^{i, \x}(\Delta)$ is well-behaved for some feasible $\Delta$, this well-behavedness does not change if we change $\Delta$. 

\begin{lemma}\label{lem:well-behaved-connection}
Let $\x$ be a fractional allocation. If the artificial instance for agent $i$, $\cal I^{i, \x}(\Delta)$ is well-behaved for some $\Delta$ feasible with respect to $\x$, then it is well-behaved for all $\Delta$ feasible with respect to $\x$.

Moreover, there is a polynomial time algorithm that takes as input the fractional allocation $\x$ and agent $i$, and outputs whether $\cal I^{i, \x}(\Delta)$ is well-behaved for all $\Delta$ feasible with respect to $\x$.
\end{lemma}
\begin{proof}
Fix an agent $i \in N$. Both claims of the lemma follow from a simple observation: each of the five properties \ref{prop:a}--\ref{prop:e} defining a well-behaved instance can be rephrased as a statement about the \emph{fraction} of the $\Delta$ agents that satisfy a certain condition, where neither the condition nor the thresholds involved depend on $\Delta$.  

Fix any feasible $\Delta$, and let the artificial instance $\cal I^{i, \x}(\Delta)$ have large good set $L$ (with $|L| = \ell$), mean value $\mu$, and round robin allocation $W$ with small bundles $(S_1, \dots, S_{\Delta})$.

Recall that \ref{prop:a} asks that $\frac{\ell}{\Delta} \in \left [1 - \frac1e-\gamma, 1 - \frac1e + \gamma \right ]$. By Lemma \ref{lem:large-good-connection}, $\frac{\ell}{\Delta} = \ell_i$, a quantity of the fractional bundle $\x_i$ alone that does not depend on $\Delta$. Hence \ref{prop:a} is equivalent to the $\Delta$-independent condition $\ell_i \in \left [1 - \frac1e-\gamma, 1 - \frac1e + \gamma \right ]$; in particular it holds either for all feasible $\Delta$ or for none. Since $\ell_i = \sum_{g \in L_i} x_{ig}$ can be computed from $L_i$ in polynomial time, this condition can be verified in polynomial time.\footnote{Verifying \ref{prop:a} requires comparing a rational number $p/q$ (represented in polynomially many bits) with $e$. Even though $e$ is transcendental, this comparison can be done in polynomial time since $|e- \frac{p}{q}| > \frac{1}{O(q^2 \log q)}$ \cite{Davis1978Rational} and $e$ can be approximated within an additive error of $2^{-t}$ in time polynomial in $t$ \cite{Brent1976e}.}
 
The four remaining properties \ref{prop:b}--\ref{prop:e} are statements about the small bundles and large goods that agents receive in the round robin allocation $W$. Since $\Delta$ can be exponentially large, we cannot compute each agent's allocated bundle in $W$ efficiently. We get around this by exploiting the fact that $W$ contains only a few \emph{distinct} bundles, which lets us describe $W$ compactly and, crucially, without any reference to $\Delta$.
 
Let $G = \{g_1, \dots, g_m\}$ denote the set of goods of the original instance, ordered so that $v_i(g_1) \ge v_i(g_2) \ge \dots \ge v_i(g_m)$. We assume the round robin algorithm breaks ties in favor of goods with an earlier index; that is, the copies of $g_1$ are allocated before the copies of $g_2$, and so on, even if two goods in $G$ have identical value. This is without loss of generality: for each agent $j$, the values $v_i(W_j)$ and $v_i(S_j)$ do not depend on how these ties are broken, so the tie-breaking rule does not affect whether $W$ satisfies the properties of a well-behaved instance.

Our succinct representation of $W$ has the form $\{(X_j, a_j, b_j)\}_{j \in [k]}$, where each $X_{j} \subseteq G$ and $a_j, b_j \in [0, 1]$. Its interpretation is that, in the allocation $W$, the agents $\{\Delta a_j + 1, \Delta a_j + 2, \dots, \Delta b_j\}$ all receive the bundle $X_j$ (more pedantically, they each receive a bundle consisting of one copy of each good in $X_j$). To ensure that this is a complete description of the allocation, we require that $a_1 = 0$, $b_k = 1$, and that $a_{j} < b_j$ and $a_{j+1} = b_j$ for all $j \in [k]$. This ensures the index intervals $(a_j\Delta, b_j\Delta]$ partition the agent set $[\Delta]$. It is easy to see that any allocation (not just $W$) can be represented using this form. We will exploit this representation form by showing that $W$ can be represented in this form with $k \le m+1$. Moreover, the fraction of agents receiving the bundle $X_j$ will be independent of $\Delta$; that is the value $(b_j - a_j)$ will be a function of $\x_i$ that does not depend on $\Delta$. 
 
We construct this representation iteratively by simulating the round robin algorithm. We simulate the round robin algorithm in blocks: in the first block we allocate all copies of $g_1$, then in the next block we allocate all copies of $g_2$, and so on. There are $m$ such blocks. To simulate the $t$-th block (the allocation of the copies of $g_t$) we require an allocation $\{(X_j, a_j, b_j)\}_{j \in [k]}$ that corresponds to the round robin allocation at the end of the first $t-1$ blocks and a {\em starting pointer} that points to the agent who should receive the first good in the $t$-th block.

Initially, the allocation is empty and the first agent is due to pick a good. Therefore, we initialize our allocation using the tuple $\{(\emptyset, 0, 1)\}$ and initialize a starting pointer $p_1 = 0$, which denotes which agent should pick a good next; the value of the starting pointer $p_1$ is set such that the agent $p_1\Delta + 1$ is due to receive the next good. 
Simulating the first block, we allocate the $\Delta x_{ig_1}$ copies of good $g_1$; these goods are allocated to the first $\Delta x_{ig_1}$ agents. We can update our allocation by splitting the initial tuple $\{(\emptyset, 0, 1)\}$ into two tuples $\{(\{g_1\}, 0, x_{ig_1}), (\emptyset, x_{ig_1}, 1)\}$, and then updating our starting pointer to $p_2 = x_{ig_1}$. This is the allocation of the round robin algorithm after $\Delta x_{ig_1}$ iterations, and the starting pointer $p_2$ denotes that the agent $\Delta x_{ig_1} + 1$ should receive the next good. 

More generally, assume we have simulated the first $t-1$ blocks and would like to simulate the allocation of the copies of $g_t$. Given a starting pointer $p_t$, the good $g_t$ is allocated to all the agents $\{\Delta p_t + 1, \dots, \Delta (p_t + x_{ig_t})\}$ if $(p_t + x_{ig_t}) \le 1$ and the set of agents $\{\Delta p_t + 1, \dots, \Delta, 1, \dots, \Delta (p_t + x_{ig_t} - 1)\}$ otherwise. We can use this information to update our allocation $\{(X_j, a_j, b_j)\}_{j \in [k]}$, which we constructed by simulating the first $t-1$ blocks. As we show, this can be done without knowing the value of $\Delta$. 

For all the tuples $(X_{j}, a_j, b_j)$ where the interval $[a_j, b_j]$ is contained in $[p_t, \min\{1, p_t + x_{ig_t}\}]$ or $[0, \max\{p_t + x_{ig_t} - 1, 0\}]$ we update $X_j$ to $X_j \cup \{g_t\}$ since all of the agents in this interval receive a copy of $g_t$. For all the tuples $(X_{j}, a_j, b_j)$ where the interval $[a_j, b_j]$ is completely disjoint from $[p_t, \min\{1, p_t + x_{ig_t}\}]$ or $[0, \max\{p_t + x_{ig_t} - 1, 0\}]$, we leave $X_j$ unchanged since none of the agents in this interval receive a copy of $g_t$. 

Next, consider the tuple $(X_{j}, a_j, b_j)$ where the interval $(a_j, b_j)$ contains the endpoint $(p_t + x_{ig_t})$ or $(p_t + x_{ig_t} - 1)$. In this case, some of the agents in $\{\Delta a_j + 1, \dots, \Delta b_j\}$ receive the good $g_t$ and some agents do not. 
Specifically, if the interval contains the endpoint $p_t + x_{ig_t}$ and does not contain $p_t$, the set of agents $\{\Delta a_j + 1, \dots, \Delta (p_t + x_{ig_t})\}$ receive the good $g_t$ and the remaining agents do not. 
Therefore, to update the allocation, we split $(X_j, a_j, b_j)$ into $(X_{j} \cup \{g_t\}, a_j, p_t + x_{ig_t})$ and $(X_{j}, p_t + x_{ig_t}, b_j)$. 
We can do something similar if the interval contains $(p_t + x_{ig_t} - 1)$. We will show that no interval contains $p_t$, so there is no other case to worry about.

Finally, we update the starting pointer $p_{t+1}$ to be $p_t + x_{ig_t}$ if $p_t + x_{ig_t} < 1$ and $p_t + x_{ig_t} - 1$ otherwise. Note that this update step also means the allocation at the beginning of the $(t+1)$-th block will not have any tuple $(X_j, a_j, b_j)$ such that $(a_j, b_j)$ contains $p_{t+1}$. If it did, we would have split it in the simulation of the $t$-th block itself. More generally, for any block $t$, $p_t$ will not be present in the interval $(a_j, b_j)$ for any tuple $(X_j, a_j, b_j)$.
 
Using this procedure, we compute a succinct representation $\{(X_j, a_j, b_j)\}_{j \in [k]}$ of the allocation $W$. To see why it is succinct, note that at each iteration, we split at most one tuple $(X_j, a_j, b_j)$ --- this is the tuple that contains either $p_t + x_{ig_t}$ or $p_t + x_{ig_t} - 1$ (only one of these will be in the range $(0, 1)$). Since we start with an allocation represented by one tuple and each tuple is split into two tuples (if it is split), we terminate with an allocation represented by at most $m+1$ tuples. 
Moreover, all the $a_j$ and $b_j$ values are functions of the fractional bundle $\x_i$. Therefore, this is a representation that only takes up polynomial space. Additionally this construction is independent of the value of $\Delta$ since all the steps can be carried out without knowing $\Delta$. Therefore, we can construct this representation of $W$ in polynomial time, and this allocation $\{(X_j, a_j, b_j)\}_{j \in [k]}$ represents the round robin allocation of $\cal I^{i, \x}(\Delta)$ for any feasible value of $\Delta$. 

We now use the representation $\{(X_j, a_j, b_j)\}_{j \in [k]}$ of the allocation $W$ to express the four remaining properties. Consider the group of agents that receive the bundle $X_j$. Since every agent in this group receives the same set of goods, the following two quantities are common to the whole group, and both depend only on $\x_i$ (not on $\Delta$):
\begin{itemize}
\item the value of the small goods in the bundle $\sigma_j := v_i(X_j \setminus L_i)$, which equals $v_i(S_{j'})$ for every agent $j'$ in the group; and
\item if the group receives a large good, the value $\lambda_j := v_i(g)$ of that large good, where $g$ is the unique element of $X_j \cap L_i$.
\end{itemize}
Here we used that each agent who receives a large good receives \emph{exactly one} of them: by property (iii) of Definition \ref{def:large-goods-fractional}, $\ell_i = \sum_{g \in L_i} x_{ig} < 1$, so there are $\ell = \Delta\ell_i < \Delta$ copies of large goods in total; since the large goods are the most valuable goods, they are all allocated in the first round of the round robin algorithm, one to each of the agents $\{1, \dots, \ell\}$. Consequently, group $j$ receives a large good if and only if $X_j \cap L_i \neq \emptyset$, in which case $|X_j \cap L_i| = 1$.
 
Using $\sigma_j$, $\lambda_j$, $\mu^i$, and the group fractions $b_j - a_j$, each remaining property becomes a statement about the allocation $\{(X_j, a_j, b_j)\}_{j \in [k]}$ that does not involve $\Delta$:
\begin{description}
\item[\ref{prop:b}] holds if and only if the fraction of agents whose small bundle has value outside $\left [(1-\gamma)\mu^i, (1+\gamma)\mu^i \right ]$ is at most $\gamma$; that is, $\sum_{j:\, \sigma_j \notin [(1-\gamma)\mu^i,\, (1+\gamma)\mu^i]} (b_j - a_j) \le \gamma$.
\item[\ref{prop:c}] holds if and only if the fraction of agents who receive a large good and a small bundle of comparable value is at most $\gamma^2$; that is, $\sum_{j:\, X_j \cap L_i \neq \emptyset,\ \sigma_j \ge \gamma^2 \lambda_j} (b_j - a_j) \le \gamma^2$.
\item[\ref{prop:d}] holds if and only if the fraction of agents who receive a large good with value less than $\frac{\mu^i}{\gamma}$ is at most $\gamma^2$; that is, $\sum_{j:\, X_j \cap L_i \neq \emptyset,\ \lambda_j \le \mu^i/\gamma} (b_j - a_j) \le \gamma^2$.
\item[\ref{prop:e}] holds if and only if $\sigma_j \ge 10^{-4}\mu^i$ for every group; that is, $\sum_{j:\, \sigma_j < 10^{-4}\mu^i} (b_j - a_j) = 0$.
\end{description}
 
Each of these conditions is expressed purely in terms of $b_j - a_j$, $\sigma_j$, $\lambda_j$, and $\mu^i$, none of which depend on $\Delta$. Hence each of \ref{prop:b}--\ref{prop:e} holds either for all feasible $\Delta$ or for no feasible $\Delta$. Together with the analogous conclusion for \ref{prop:a}, this shows that the well-behavedness of $\cal I^{i, \x}(\Delta)$ does not depend on the choice of feasible $\Delta$, proving the first claim of the lemma. Finally, the representation $\{(X_j, a_j, b_j)\}_{j \in [k]}$ of $W$, the set $L_i$ and the value $\mu^i$ are all computable in polynomial time. Therefore all of the properties can be verified in polynomial time, proving the second claim of the lemma.
\end{proof}

We are now ready to define what it means for an agent $i$ to be well-behaved in a fractional allocation $\x$. This definition will also require the definition of the small good fraction property. 

\begin{definition}[Small Good Fraction Property]\label{def:small-good-fraction}
Given a fractional allocation $\x$, an agent $i$ satisfies the small good fraction property in $\x$ if $\sum_{g \in G \setminus L_i} \mathbbm{1}\{x_{ig} > 0.5 \} x_{ig} v_i(g) < \gamma \mu^i$.
\end{definition}

\begin{definition}[Well-behaved Agent]
Let $\x$ be a fractional allocation. We say that $\x$ is well-behaved for an agent $i$ if $i$ satisfies the small good fraction property in $\x$ and the artificial instance for agent $i$, $\cal I^{i, \x}(\Delta)$, is well-behaved for all feasible $\Delta$ with respect to $\x$.
\end{definition}

As a corollary of Lemma \ref{lem:well-behaved-connection}, we can identify well-behaved agents efficiently.

\begin{corr}\label{cor:well-behaved-identification}
Given a fractional allocation $\x$, there is a polynomial time algorithm that can efficiently identify the set of well-behaved agents in $\x$.
\end{corr}
\begin{proof}
Using Lemma \ref{lem:large-goods-unique-fractional}, we can find the set $L_i$ for each agent $i$ in polynomial time. Using this set, it is easy to check if an agent $i$ satisfies the small good fraction property. Using Lemma \ref{lem:well-behaved-connection}, we can check in polynomial time if $\cal I^{i, \x}(\Delta)$ is well-behaved for all feasible $\Delta$ with respect to $\x$. Therefore, we can decide in polynomial time if an agent $i$ is well-behaved, and the corollary follows.
\end{proof}

\subsection{Analyzing the Shmoys-Tardos Rounding Algorithm}
In this section, we prove a few lemmas lower-bounding $\ST(\x, i)$ (and a couple of other useful results). These results essentially translate the results from Sections \ref{sec:round-robin-guarantees} and \ref{sec:suboptimal} into the terminology defined in the previous subsection. All of these results have a very similar proof technique. So we present two proofs and then relegate the remaining proofs to the appendix.

We start with the small good fraction property. This property was introduced rather abruptly in the previous subsection, with no specific justification. The following lemma provides this justification showing that we improve over Remark \ref{rem:fl} if an agent $i$ does not satisfy the small good fraction property in $\x$. 

\begin{lemma}\label{lem:small-good-fraction}
Let $y$ be a configuration LP solution and $\x$ be the corresponding fractional allocation. If $\cal I^{i, \x}(\Delta)$ is well-behaved for all $\Delta$ feasible with respect to $\x$ and $i$ does not satisfy the small good fraction property in $\x$, then
\begin{align*}
\ST(\x, i) \ge \sum_{S \subseteq G} y_{i, S} \ln(v_i(S)) - \frac{1}{e} + \gamma^2.
\end{align*}
\end{lemma}
\begin{proof}
Let $\Delta$ be any integer that is feasible with respect to both $y$ and $\x$. Consider the artificial instance $\cal I^{i, \x}(\Delta)$ for agent $i$.  Note that $\cal I^{i, \x}(\Delta)$ is well-behaved by the assumption in the lemma.
Define the allocation $Y$ as one where $\Delta y_{i, S}$ agents receive the bundle $S$ for each $S \subseteq G$. This is the same allocation constructed in the proof of Lemma \ref{lem:y-substitute}. The goal of this proof is to invoke Lemma \ref{lem:suboptimal} with the allocation $Y$. Note that the allocation $Y$ has positive Nash welfare since the objective value of the configuration LP solution $y$ is well-defined; therefore, if $y_{i, S} > 0$, the value $v_i(S)$ is positive.

Let $G^i$ denote the set of goods of this instance $\cal I^{i, \x}(\Delta)$. This set of goods $G^i$ can be naturally partitioned based on which good in $G$ it is a copy of. Let $g'_1, g'_2, \dots, g'_k$ be the set of small goods for agent $i$, and let $G_1, G_2, \dots G_k$ be a partition of the small goods of $\cal I^{i, \x}(\Delta)$ where each $G_j$ consists of all the copies of good $g'_j$. Using Lemma \ref{lem:large-good-connection}, the set of small goods of $\cal I^{i, \x}(\Delta)$ is exactly the set of copies of $g'_1, \dots, g'_k$. Therefore, $G_1, \dots G_k$ partition the set of small goods of $\cal I^{i, \x}(\Delta)$ into classes where all goods in each class have the same value. 

Finally, note that for each small good $g$, there are $\Delta x_{ig}$ copies of $g$ in the instance $\cal I^{i, \x}(\Delta)$. Therefore, 
\begin{align*}
\sum_{j = 1}^k \mathbbm{1}\{|G_j| > 0.5\Delta\} v_i(G_j) = \sum_{j = 1}^k \mathbbm{1}\{x_{ig'_j} > 0.5\} \Delta x_{ig'_j} v_i(g'_j) \ge \gamma \Delta \mu^i = \gamma \Delta \mu.
\end{align*}

The final equality holds due to Lemma \ref{lem:large-good-connection}.
Additionally note that in the allocation $Y$, no agent gets two copies of a good in $G$ since each agent's allocation is a subset of $G$. 
We can therefore apply Lemma \ref{lem:suboptimal} to conclude that 
\begin{align}
\logNSW(Y) < \frac1{\Delta} \left [\sum_{g \in L}\ln v(g) + (\Delta - \ell)\ln \left (\frac{\Delta \mu}{\Delta - \ell} \right ) \right ]  - \gamma^2. \label{eq:small-good-fraction-1}
\end{align}

Let $W$ be the round robin allocation of this instance $\cal I^{i, \x}(\Delta)$. Using Theorem \ref{thm:round-robin-nash-lower-bound}, we have:
\begin{align}
\logNSW(W) \ge \frac1{\Delta} \left [\sum_{g \in L}\ln v(g) + (\Delta - \ell)\ln \left (\frac{\Delta \mu}{\Delta - \ell} \right ) \right ]  - \frac1e. \label{eq:small-good-fraction-2}
\end{align}

Combining \eqref{eq:small-good-fraction-1} and \eqref{eq:small-good-fraction-2}, we get
\begin{align*}
\logNSW(W) \ge \logNSW(Y) - \frac1e + \gamma^2.
\end{align*}
Using Lemma \ref{lem:w-worst}, we have $\ST(\x, i) \ge \logNSW(W)$. Further, using the arguments of Lemma \ref{lem:y-substitute}, we have $\logNSW(Y) = \sum_{S \subseteq G}y_{i, S} \ln v_i(S)$. Using these two observations, we can simplify the above inequality to conclude that
\begin{equation*}
\ST(\x, i) \ge \sum_{S \subseteq G} y_{i, S} \ln(v_i(S)) - \frac{1}{e} + \gamma^2. \qedhere
\end{equation*}
\end{proof}

\begin{lemma}\label{lem:non-well-behaved-instance-improvement-fractional}
Let $y$ be a configuration LP solution and let $\x$ be the corresponding fractional allocation. If $\x$ is not well-behaved for agent $i$, the following holds for some constant $c>0$:
\begin{align*}
\ST(\x, i) \ge  \sum_{S \subseteq G}y_{i, S}\ln v_i(S) - \frac{1}{e} + c.
\end{align*}
\end{lemma}
\begin{proof}
Let $\cal I^{i, \x}(\Delta)$ be the artificial instance of agent $i$ for some $\Delta$ that is feasible with respect to $\x$ and $y$. If $\cal I^{i, \x}(\Delta)$ is well-behaved but $i$ does not satisfy the small good fraction property, the lemma holds due to Lemma \ref{lem:small-good-fraction}. We therefore assume the instance $\cal I^{i, \x}(\Delta)$ is not well-behaved.

Let $W$ be the round robin allocation of this instance. Since the instance $\cal I^{i, \x}(\Delta)$ is not well-behaved, Theorem \ref{thm:non-well-behaved-instance-improvement} proves the following:
\begin{align*}
\logNSW(W) \ge  \logNSW(\cal I^{i, \x}(\Delta)) - \frac{1}{e} + c,
\end{align*}
for some constant $c > 0$.  The lemma follows by noting that $\ST(\x, i) \ge \logNSW(W)$ (Lemma \ref{lem:w-worst}) and $\logNSW(\cal I^{i, \x}(\Delta)) \ge \sum_{S \subseteq G} y_{i, S}\ln v_i(S)$ (Lemma \ref{lem:y-substitute}).
\end{proof}

The remaining proofs are relegated to Appendix \ref{apdx:st-rounding}.
\begin{restatable}{lemma}{lemfractionalnashlowerbound}\label{lem:fractional-nash-lower-bound}
Let $y$ be a configuration LP solution and let $\x$ be the corresponding fractional allocation. Let $i$ be an agent with set of large goods $L_i$, mean value $\mu^i$ and large good amount $\ell_i$. The following two inequalities hold:
\begin{align*}
\ST(\x, i) &\ge   \sum_{g \in L_i} x_{ig}\ln{v_i(g)} + (1 - \ell_i) \ln\left ( \frac{\mu^i}{1 - \ell_i} \right )  - \frac{1}{e}. \\
\ST(\x, i) &\ge \sum_{S \subseteq G} y_{i, S}\ln v_i(S) - \frac1e.
\end{align*}
\end{restatable}

The following results do not lower bound $\ST(\x, i)$ but they use the same proof technique, and will be useful in future sections.

\begin{restatable}{lemma}{lemfractionalpropertya}\label{lem:fractional-property-a}
Let $\x$ be a fractional allocation. Let $\ell_i$ be the large good amount of agent $i$. If agent $i$ is well-behaved in $\x$, then $\ell_i \in \left [1 - \frac1e - \gamma, 1 - \frac1e+\gamma \right ]$.
\end{restatable}

\begin{restatable}{lemma}{lemfractionalnashupperbound}\label{lem:fractional-nash-upper-bound}
Let $y$ be a configuration LP solution and let $\x$ be the corresponding fractional allocation. Let $i$ be an agent with set of large goods $L_i$, mean value $\mu^i$ and large good amount $\ell_i$. The following inequality holds:
\begin{align*}
\sum_{S\subseteq G} y_{i, S}\ln v_i(S) \le \sum_{g \in L_i} x_{ig}\ln{v_i(g)} + (1 - \ell_i) \ln \mu^i   + \frac{1}{e}. 
\end{align*}
\end{restatable}

\subsection{Improved Approximations when Agents are not Well-behaved}

In this subsection, we show that if a small number of agents are not well-behaved, then the ST rounding algorithm achieves an approximation ratio better than $e^{1/e}$. 

\begin{theorem}\label{thm:non-well-behaved-algo}
Let $y$ be a configuration LP solution and let $\x$ be the corresponding fractional allocation. If at least $\gamma^3 n$ agents are not well-behaved in $\x$, then the following holds
\begin{align*}
\ST(\x) \ge \frac{1}n \sum_{i \in N, S \subseteq G} y_{i, S}\ln v_i(S) - \frac1e + c^*_1
\end{align*}
for some constant $c^*_1 > 0$.
\end{theorem}
\begin{proof}
Let $\Delta$ be an integer feasible with respect to both $\x$ and $y$. 
Let $N_1$ be the set of agents $i$ such that $\x$ is not well-behaved for agent $i$. By the theorem statement, $|N_1| \ge \gamma^3 n$.

For each agent $i \in N_1$, the following holds due to Lemma \ref{lem:non-well-behaved-instance-improvement-fractional}:
\begin{align}
\ST(\x, i) \ge  \sum_{S \subseteq G}y_{i, S}\ln v_i(S) - \frac{1}{e} + c. \label{eq:n1-bound}
\end{align}

For all the other agents $i \in N \setminus N_1$, the following holds due to Lemma \ref{lem:fractional-nash-lower-bound}:
\begin{align}
\ST(\x, i) \ge \sum_{S \subseteq G} y_{i, S} \ln(v_i(S)) - \frac{1}{e} \label{eq:n-bound}.
\end{align}

Adding up \eqref{eq:n1-bound} and \eqref{eq:n-bound}, we get
\begin{align*}
\ST(\x) &\ge \frac1n \sum_{i \in N, S \subseteq G} y_{i, S}\ln v_i(S) -\frac1e + \frac{|N_1|c}{n}  \\
&\ge \frac1n \sum_{i \in N, S \subseteq G} y_{i, S} \ln v_i(S) - \frac1e + \gamma^3 c 
\end{align*}
Therefore the theorem holds with $c^*_1$ defined as the value $\gamma^3 c$.
\end{proof}

\section{Large Good Consistency}\label{sec:large-good-consistency}
In this section, we show that all goods $g \in G$ must be either large for all the agents (or at least most of them), or must be small throughout. For the case where this does not happen, we present a rounding algorithm that achieves an approximation ratio of $e^{1/e}(1 - c)$ for some constant $c > 0$. 

Throughout this section, we use $N_{wb}$ to denote the set of well-behaved agents, and $N_{nwb}$ to denote the set of non-well-behaved agents.

\begin{definition}[Well-behaved Globally Large Goods]\label{def:globally-large}
Given a fractional allocation $\x$, we refer to a good $g$ as {\em globally large} if it is large for at least one agent $i$. We define a globally large good as well-behaved if its total fractional allocation to well-behaved agents who consider it large is at least $1-\gamma$. Formally, a globally large good $g$ is well-behaved if $\sum_{i \in N_{wb}} \mathbbm{1}\{g \in L_i\}x_{ig} \ge 1-\gamma$.
\end{definition}

We use $G_{wb}$ to denote the set of well-behaved globally large goods, and let $G_{nwb}$ denote all the other goods. We define the large good consistency property.

\begin{description}
\item[Large Good Consistency Property] A fractional allocation $\x$ satisfies the large good consistency property if the total fractional amount that non-well-behaved goods are allocated to well-behaved agents who consider them large is at most $10\gamma n$. Formally, the large good consistency property is satisfied if $\sum_{g \in G_{nwb}} \sum_{i \in N_{wb}} \mathbbm{1}\{g \in L_i\}x_{ig} \le 10\gamma n$.
\end{description}

Our main result of this section is the following.

\begin{theorem}\label{thm:no-large-good-consistency}
Let $y$ be a configuration LP solution for an instance $\cal I$, and let $\x$ be the corresponding fractional allocation. If at most $\gamma^3 n$ agents are not well-behaved in $\x$ and $\x$ does not satisfy the large good consistency property, then there is a polynomial time algorithm that outputs a fractional allocation $\z$ such that 
\begin{align*}
\ST(\z) \ge \frac1n \sum_{i \in N, S \subseteq G} y_{i, S} \ln(v_i(S)) - \frac1e + c^*_2,
\end{align*}
for some constant $c^*_2 \in (0, 1)$.
\end{theorem}

\subsection{Proof of Theorem \ref{thm:no-large-good-consistency}}
Our rounding algorithm can be summarized at a high level as follows: for a few carefully chosen set of large goods that are not well-behaved, we move a small fraction of each good from agents who value it as a small good to agents who value it as a large good. Intuitively, removing a small good from an agent's fractional bundle only reduces the Nash welfare by a small amount but adding a large good to an agent's bundle increases the Nash welfare by a large amount. 
However we need to be careful to ensure that the Nash welfare improves.

For each good $g$, we define $\zeta(g)$ to denote the total fractional amount $g$ is allocated as a large good to well-behaved agents; that is, $\zeta(g) = \sum_{i \in N_{wb}} \mathbbm{1}\{g \in L_i\}x_{ig}$. Given any subset $G' \subseteq G$, we write $\zeta(G') = \sum_{g \in G'} \zeta(g)$. If the large good consistency property is violated, then $\zeta(G_{nwb}) > 10\gamma n$.  

Each non-well-behaved large good is allocated either as a small good or to a non-well-behaved agent with at least a $\gamma$ fraction, using Observation \ref{obs:x-complete}. We partition these goods based on whether this $\gamma$ fraction is allocated to well-behaved agents or non-well-behaved agents. Specifically, for each $g \in G_{nwb}$, the following holds:
\begin{align*}
\sum_{i \in N_{wb}} \mathbbm{1}\{g \notin L_i\} x_{ig} + \sum_{i \in  N_{nwb}} x_{ig} \ge \gamma.
\end{align*}

This implies at least one of $\sum_{i \in N_{wb}} \mathbbm{1}\{g \notin L_i\} x_{ig}$ and $\sum_{i \in N_{nwb}} x_{ig}$ must be at least $\gamma/2$. 

Let $G'\subseteq G_{nwb}$ denote the set of these goods which have $\sum_{i \in N_{wb}} \mathbbm{1}\{g \notin L_i\} x_{ig} \ge \gamma/2$, and let $G'' \subseteq G_{nwb}$ denote the set of these goods which have $\sum_{i \in N_{nwb}} x_{ig} \ge \gamma/2$.

Since $G' \cup G'' = G_{nwb}$, we must have $\zeta(G') + \zeta(G'') \ge \zeta(G_{nwb}) > 10\gamma n$. Our proof is split into two cases. In the first case, we assume $\zeta(G') \ge \gamma n$. In the second case, we assume $\zeta(G'') \ge 9\gamma n$. Since $\zeta(G') + \zeta(G'')  > 10\gamma n$, one of these two cases must hold.

\medskip 

\textbf{Case 1:} $\zeta(G') \ge \gamma n$. 

\medskip 
We create a new fractional allocation $\z$ which we initialize at $\x$ and modify by doing the following for each good $g \in G'$.
\begin{itemize}
		\item Let $L(g)$ denote the set of well-behaved agents $i$ such that $g \in L_i$, and let $S(g)$ be the set of well-behaved agents $i$ such that $g \in G \setminus L_i$. 
		\item Let $\zeta(g) = \sum_{i \in N_{wb}} \mathbbm{1}\{g \in L_i\} x_{ig}$ be the fractional amount of $g$ which is allocated as a large good, and let $s(g) = \sum_{i \in N_{wb}} \mathbbm{1}\{g \notin L_i\} x_{ig}$ denote the fractional amount of $g$ which is allocated as a small good. 
		\item For each $i \in L(g)$, $z_{ig} \gets x_{ig}(1 + \gamma^2)$.
		\item For each $i \in S(g)$, $z_{ig} \gets x_{ig} \left (1 - \frac{\gamma^2\zeta(g)}{s(g)} \right )$.
\end{itemize}

Our proof consists of two parts. We first show that $\z$ is a valid allocation. Then we show that applying ST rounding with $\z$ leads to a constant factor improvement in Nash welfare.

\begin{claim}\label{claim:z-valid}
For each good $g \in G$, $\sum_{i \in N}z_{ig} = \sum_{i \in N}x_{ig}$. Moreover, for each good $g$ and agent $i$, $z_{ig} \ge 0$. 
\end{claim}
\begin{proof}
This is trivial for each good not in $G'$ since their fractional allocations are the same in both $\x$ and $\z$. Consider any good $g \in G'$. We use the notation $L(g)$, $S(g)$, $\zeta(g)$ and $s(g)$ as defined in the above procedure. 

\begin{align*}
\sum_{i \in N} z_{ig} &= \sum_{i \in L(g)} z_{ig} + \sum_{i \in S(g)} z_{ig} + \sum_{i \in N \setminus (L(g) \cup S(g))} z_{ig} \\
&= \sum_{i \in L(g)} x_{ig}(1 + \gamma^2) + \sum_{i \in S(g)} x_{ig}\left ( 1- \frac{\gamma^2 \zeta(g)}{s(g)}\right ) + \sum_{i \in N \setminus (L(g) \cup S(g))} x_{ig} \\
&= \sum_{i \in N}x_{ig} + \sum_{i \in N_{wb}} \mathbbm{1}\{g \in L_i\}\gamma^2 x_{ig} - \frac{\gamma^2\zeta(g)}{s(g)} \sum_{i \in N_{wb}} \mathbbm{1}\{g \notin L_i\} x_{ig} \\
&= \sum_{i \in N} x_{ig} + \gamma^2 \zeta(g) - \frac{\gamma^2\zeta(g)}{\sum_{i \in N_{wb}}\mathbbm{1}\{g \notin L_i\}x_{ig}} \sum_{i \in N_{wb}} \mathbbm{1}\{g \notin L_i\} x_{ig} \\
&= \sum_{i \in N} x_{ig}.
\end{align*}

To show that each $z_{ig}$ is non-negative, we only need to consider the cases where $z_{ig} < x_{ig}$. These are the cases where we set $z_{ig} \gets x_{ig}\left ( 1- \frac{\gamma^2 \zeta(g)}{s(g)}\right )$. In this case, from the definition of goods in $G'$, $s(g) \ge \frac{\gamma}{2}$ and we trivially have $\zeta(g) \le 1$. Therefore, 
\begin{align*}
\left ( 1- \frac{\gamma^2 \zeta(g)}{s(g)}\right ) \ge (1 - 2\gamma) > 0.
\end{align*}
This proves that $z_{ig}$ is always non-negative.
\end{proof}

Fix an agent $i \in N$. To analyze $\ST(\z, i)$, we consider two artificial instances $\cal I^{i, \x}(\Delta)$ and $\cal I^{i, \z}(\Delta)$ for some $\Delta$ that is feasible with respect to $\x$, $\z$ and the configuration LP solution $y$. Let $W^i$ denote the round robin allocation of $\cal I^{i, \x}(\Delta)$ and let $U^i$ denote the round robin allocation of $\cal I^{i, \z}(\Delta)$.

We start with the well-behaved agents: assume $i \in N_{wb}$.
The instance $\cal I^{i, \z}(\Delta)$ can be constructed from $\cal I^{i, \x}(\Delta)$ by removing some small goods and adding some large goods. For each good $g \in G'$, if $g$ is a large good for the agent $i$, we add $\gamma^2 \Delta x_{ig}$ copies of $g$ to $\cal I^{i,\x}(\Delta)$, and if $g$ is small for agent $i$, we remove $\frac{\gamma^2\zeta(g)}{s(g)}\Delta x_{ig}$ copies of $g$ from $\cal I^{i, \x}(\Delta)$.


Let $R^i$ be the set of small goods removed from $\cal I^{i, \x}(\Delta)$ and let $A^i$ be the set of goods added to create $\cal I^{i, \z}(\Delta)$. Our goal is to use Corollary \ref{cor:small-good-removal-large-good-addition} to lower bound the log Nash welfare of $U^i$. For this, we need to show that conditions of Corollary \ref{cor:small-good-removal-large-good-addition} are satisfied. First, the instance $\cal I^{i, \x}(\Delta)$ is well-behaved since agent $i$ is well-behaved. 

For each good $g \in G'$. We remove $\frac{\gamma^2 \zeta(g)}{s(g)} \Delta x_{ig}$ copies of the good. Note that since $s(g) \ge \gamma/2$, we have
\begin{align*}
\frac{\gamma^2 \zeta(g)}{s(g)} \le 2\gamma.
\end{align*}

This implies the total value removed is at most $\sum_{g \in G \setminus L_i} 2\gamma \Delta x_{ig} v_i(g) \le 2\gamma \Delta \mu^i = 2\gamma \Delta \mu$, where $\mu$ is mean value of the instance $\cal I^{i, \x}(\Delta)$. This shows that the total amount removed is at most $0.1\Delta \mu$. Moreover the number of large goods we add is at most $\gamma^2\ell_i \Delta \le \gamma^2 \Delta$. Therefore, the conditions of Corollary \ref{cor:small-good-removal-large-good-addition} are satisfied. However, to bound the final change in Nash welfare, we need to lower bound the number of goods we add with value at least $\frac{\mu}{\gamma}$. We do this in the following claim:

\begin{claim}\label{claim:z-addition}
The following two statements hold for each $i \in N_{wb}$:
\begin{enumerate}[(i)]
\item $|A^i| = \Delta \gamma^2 \sum_{g \in G'}\mathbbm{1}\{g \in L_i\} x_{ig}$, and
\item there are at least $\max \{|A^i| - \gamma^4 \Delta, 0\}$ goods $g$ in $A^i$ which have value at least $v_i(g) \ge \frac{\mu}{\gamma}$
\end{enumerate}
\end{claim}
\begin{proof}
We add $\gamma^2\Delta x_{ig}$ copies of each good $g \in G'$ which is a large good for the agent $i$. Therefore, the first statement holds. To prove the second statement, let $\widetilde{G} \subseteq G'$ be the set of large goods in $G'$ such that for each $g \in \widetilde{G}$, we have $v_i(g) \le \frac{\mu}{\gamma}$. By \ref{prop:d} of the well-behaved instance $\cal I^{i, \x}(\Delta)$, we have $\Delta \sum_{g \in \widetilde{G}} x_{ig} \le \gamma^2 \Delta$. This implies the number of copies of goods in $\widetilde{G}$ that we add is at most $\gamma^4 \Delta$.
\end{proof}

Applying Corollary \ref{cor:small-good-removal-large-good-addition}, we get for each agent $i \in N_{wb}$: 
\begin{align*}
\logNSW(U^i) \ge \logNSW(\cal I^{i, \x}(\Delta)) - \frac{1}{e} + \frac{10\max\{|A^i| - \gamma^4 \Delta, 0\}}{\Delta} - \frac{2|R^i|}{\Delta}.
\end{align*}

Applying Lemma \ref{lem:w-worst} and \ref{lem:y-substitute} to the above inequality, we get for each $i \in N_{wb}$:
\begin{align}
\ST(\z, i) \ge \sum_{S \subseteq G} y_{i, S}\ln v_i(S) -\frac1e + \frac{10\max\{|A^i| - \gamma^4 \Delta, 0\}}{\Delta} - \frac{2|R^i|}{\Delta}. \label{eq:wb-bound-large-good-1}
\end{align}

For all the other agents $i \in N_{nwb}$, both $\x_i$ and $\z_i$ are identical. Therefore, $\cal I^{i, \z}(\Delta)$ and $\cal I^{i, \x}(\Delta)$ are identical. This allows us to use Lemma \ref{lem:barman-eonebye} to get:
\begin{align*}
\logNSW(U^i) \ge \logNSW(\cal I^{i, \z}(\Delta)) - \frac1e  = \logNSW(\cal I^{i, \x}(\Delta)) - \frac1e.
\end{align*}

We can then apply Lemma \ref{lem:w-worst} and \ref{lem:y-substitute} to the above inequality to get for each $i \in N_{nwb}$:
\begin{align}
\ST(\z, i) \ge \sum_{S \subseteq G} y_{i, S}\ln v_i(S) - \frac1e.\label{eq:nwb-bound-large-good-1}
\end{align}

Adding up \eqref{eq:wb-bound-large-good-1} and \eqref{eq:nwb-bound-large-good-1} gives us
\begin{align*}
\ST(\z) &\ge \frac1n \sum_{i \in N} \sum_{S \subseteq G} y_{i, S} \ln(v_i(S)) - \frac1e - 10\gamma^4 + \frac{1}{\Delta n} \left (10\sum_{i \in N_{wb}} |A^i| - 2\sum_{i \in N_{wb}} |R^i| \right ) 
\end{align*}

In Claim \ref{claim:z-valid}, we show that the total amount that each good is allocated in $\z$ is the same as $\x$. Therefore $\sum_{i \in N_{wb}} |A^i| = \sum_{i \in N_{wb}} |R^i|$. Plugging this into the above inequality gives us
\begin{align}
\ST(\z) &\ge \frac1n \sum_{i \in N} \sum_{S \subseteq G} y_{i, S} \ln(v_i(S)) - \frac1e - 10\gamma^4 + \frac{8}{\Delta n} \sum_{i \in N_{wb}} |A^i| \label{eq:intermediate-large-good-1}
\end{align}

We lower bound $\sum_{i \in N_{wb}} |A^i|$ using the first statement of Claim \ref{claim:z-addition}.
\begin{align*}
\sum_{i \in N_{wb}} |A^i| = \Delta \gamma^2 \sum_{i \in N_{wb}} \sum_{g \in G'} \mathbbm{1}\{g \in L_i\} x_{ig} = \Delta \gamma^2 \zeta(G') \ge \Delta \gamma^3 n.
\end{align*} 
The final inequality uses the assumption in this case that $\zeta(G') \ge \gamma n$. Plugging this into \eqref{eq:intermediate-large-good-1}, we get

\begin{align*}
\ST(\z) &\ge \frac1n \sum_{i \in N} \sum_{S \subseteq G} y_{i, S} \ln(v_i(S)) - \frac1e - 10\gamma^4 + 8\gamma^3
\end{align*}

This proves the theorem with the constant $c^*_2 = 8\gamma^3 - 10\gamma^4$.

\medskip 
\noindent\textbf{Case 2:} $\zeta(G'') \ge 9\gamma n$. 
\medskip

Recall that  $G''$ is the set of these goods which have $\sum_{i \in N_{nwb}} x_{ig} \ge \gamma/2$. We consider the fractional bundle $\x_i$ for some agent $i$. We partition this bundle into buckets as described in the ST rounding algorithm. For each agent $i \in N$ and good $g \in G$, we define the value $f_{ig}$ to denote the fraction of $g$ in the first bucket of the bundle $\x_i$. 

We use $\widetilde{G} \subseteq G''$ to denote the set of goods in $G''$ with $\sum_{i \in  N_{nwb}} f_{ig} \le \frac{\gamma}{4}$; this is the set of goods which have a total fractional allocation of at most $\frac{\gamma}4$ in the first bucket of the fractional bundles of non-well-behaved agents. 
 
\begin{claim}\label{claim:tildeg-lower-bound}
$\zeta(\widetilde{G}) \ge 8\gamma n$
\end{claim}
\begin{proof}
Each good $g \in G'' \setminus \widetilde{G}$ satisfies $\sum_{i \in N_{nwb}} f_{ig} > \frac{\gamma}4$. This implies

\begin{align*}
|N_{nwb}| \ge \sum_{i \in N_{nwb}} \sum_{g \in G} f_{ig} \ge \sum_{i \in N_{nwb}} \sum_{g \in G'' \setminus \widetilde{G}} f_{ig} > \frac{\gamma}4 |G'' \setminus \widetilde{G}|. 
\end{align*}

Combining this with the fact that $|N_{nwb}| \le \gamma^3 n$, we get $|G'' \setminus \widetilde{G}| \le 4\gamma^2 n < \gamma n$. Using this, we can upper bound $\zeta(G'')$ as
\begin{align*}
\zeta(G'') = \zeta(\widetilde{G}) + \zeta(G'' \setminus \widetilde{G}) &\le \zeta(\widetilde{G}) + |G'' \setminus \widetilde{G}| \\
&\le \zeta(\widetilde{G}) + \gamma n.
\end{align*}

Noting that $\zeta(G'') \ge 9\gamma n$, we can re-arrange the terms of the above inequality to infer that $\zeta(\widetilde{G}) \ge 8\gamma n$. 
\end{proof}

We are now ready to define our new fractional allocation. In this algorithm, we move small fractions of goods in $\widetilde{G}$ from non-well-behaved agents to well-behaved agents. We define a fractional allocation $\z$ initialized at $\x$ and modified according to the following procedure. We do the following for each $g \in \widetilde{G}$:
\begin{itemize}
		\item Let $L(g)$ denote the set of well-behaved agents $i$ such that $g \in L_i$. For each $g \in G$ and $i \in N$, let $f_{ig}$ denote the first bucket allocation of good $g$ in the bundle $\x_i$. 
		\item For each $i \in L(g)$, $z_{ig} \gets x_{ig}(1 + \gamma^2)$.
		\item For each $i \in N_{nwb}$, $z_{ig} \gets \frac{x_{ig} + f_{ig}}{2}$.
\end{itemize}
\medskip

Our proof has the same two steps as the first case. We first show that this rounding procedure is valid. Then, we lower bound the value of $\ST(\z)$. 

\begin{claim}\label{claim:z-valid-case-2}
For each good $g \in G$, $\sum_{i \in N}z_{ig} \le \sum_{i \in N}x_{ig}$. Moreover for each good $g$ and agent $i$, $z_{ig} \ge 0$. 
\end{claim}
\begin{proof}
We restrict our attention to goods in $\widetilde{G}$. Goods outside of $\widetilde{G}$ are allocated the same way in both allocations. Fix a good $g \in \widetilde{G}$. We need to show that the amount of $g$ we add to well-behaved agents is less than the amount of $g$ we remove from non-well-behaved agents. The total amount of $g$ we add to well-behaved agents is at most 
\begin{align*}
\sum_{i \in N_{wb}} \mathbbm{1}\{g \in L_i\} \gamma^2 x_{ig} \le \gamma^2 \sum_{i \in N} x_{ig} \le \gamma^2.
\end{align*}

The amount we remove from non-well-behaved agents is at least
\begin{align*}
\sum_{i \in N_{nwb}} \frac{x_{ig} - f_{ig}}{2} = \sum_{i \in N_{nwb}} \frac{x_{ig}}{2} - \sum_{i \in N_{nwb}} \frac{f_{ig}}{2} \ge \frac{\gamma}{4} - \frac{\gamma}{8} \ge \frac{\gamma}{8}.
\end{align*} 
Here, we use the fact that $g$ is in both $G''$ and $\widetilde{G}$. Therefore, $\sum_{i \in N_{nwb}} x_{ig} \ge \frac{\gamma}{2}$ and $\sum_{i \in N_{nwb}} f_{ig} \le \frac{\gamma}{4}$. $\gamma/8$ is clearly greater than $\gamma^2$. This proves the first statement. 

The second statement is easy to see. Every time we set $z_{ig}$ such that $z_{ig} < x_{ig}$, we set $z_{ig} \gets \frac{x_{ig} + f_{ig}}{2}$ which is clearly at least $0$.
\end{proof}

To analyze the change in Nash welfare, we use a similar technique to the previous case. Fix an agent $i \in N$. We consider two artificial instances $\cal I^{i, \x}(\Delta)$ and $\cal I^{i, \z}(\Delta)$ for some $\Delta$ that is feasible with respect to $\x$, $\z$ and the configuration LP solution $y$. Let $W^i$ denote the round robin allocation of $\cal I^{i, \x}(\Delta)$ and let $U^i$ denote the round robin allocation of $\cal I^{i, \z}(\Delta)$.

Note that $\cal I^{i, \z}(\Delta)$ can be constructed from $\cal I^{i, \x}(\Delta)$ by removing some goods if $i$ is non-well-behaved or adding some large goods if $i$ is well-behaved. 

Consider the case where $i$ is well-behaved.
For each good in $\widetilde{G}$, if $g$ is a large good for the agent $i$, we add $\gamma^2 \Delta x_{ig}$ copies of $g$ to $\cal I^{i, \x}(\Delta)$ to create $\cal I^{i, \z}(\Delta)$. Let $A^i$ be the set of goods added. To bound the change in Nash welfare, we use Corollary \ref{cor:small-good-removal-large-good-addition} with the set of removed goods being the empty set and the set of added goods being $A^i$. It is easy to see that conditions of Corollary \ref{cor:small-good-removal-large-good-addition} are satisfied since $i$ is well-behaved (Lemma \ref{lem:well-behaved-connection}) and we add at most $\gamma^2\Delta$ goods $A^i$ to the instance $\cal I^{i, \x}(\Delta)$.  However, we need to bound the high valued goods in $\widetilde{G}$ for this, which we do using the following claim. The proof of this claim is almost identical to the proof of Claim \ref{claim:z-addition}.

\begin{claim}\label{claim:z-addition-case-2}
The following two statements hold:
\begin{enumerate}[(i)]
\item $|A^i| = \Delta \gamma^2 \sum_{g \in \widetilde{G}}\mathbbm{1}\{g \in L_i\} x_{ig}$, and
\item there are at least $\max \{|A^i| - \gamma^4 \Delta, 0\}$ goods $g$ in $A^i$ which have value at least $v_i(g) \ge \frac{\mu}{\gamma}$.
\end{enumerate}
\end{claim}

Using this claim, we can invoke Corollary \ref{cor:small-good-removal-large-good-addition} to get the following bound for each agent $i \in N_{wb}$
\begin{align*}
\logNSW(U^i) \ge \logNSW(\cal I^{i, \x}(\Delta)) - \frac{1}{e} + \frac{10\max\{|A^i| - \gamma^4 \Delta, 0\}}{\Delta}.
\end{align*}

Applying Lemma \ref{lem:w-worst} and \ref{lem:y-substitute} to the above inequality, we get for each $i \in N_{wb}$:
\begin{align}
\ST(\z, i) \ge \sum_{S \subseteq G} y_{i, S}\ln v_i(S) -\frac1e -10\gamma^4+ \frac{10|A^i|}{\Delta}. \label{eq:wb-bound-large-good-2}
\end{align}

Now consider the case where $i$ is not well-behaved. Let $G^i$ be the set of goods of the instance $\cal I^{i, \x}(\Delta)$. To bound the Nash welfare in this case, we use Theorem \ref{thm:non-well-behaved-removal}. For this, we need to construct a set of $\Delta$ goods $G^{\Delta}$ which are weakly more valuable than all the other goods of the instance. We construct $G^{\Delta}$ as the set which contains $\Delta f_{ig}$ copies of each good $g$, where $f_{ig}$ is the amount of $g$ present in the first bucket of agent $i$. By the definition of our bucketing procedure, this creates a set of $\Delta$ goods $G^{\Delta}$ such that each good in $G^{\Delta}$ is weakly more valuable than each good in $G^i \setminus G^{\Delta}$.

We partition the goods $G^i$ based on the copy of the good $g \in G$ they are. Note that there are $\Delta x_{ig}$ copies of each good $g \in G$ in the instance $\cal I^{i, \x}(\Delta)$. We obtain $\cal I^{i, \z}(\Delta)$ from $\cal I^{i, \x}(\Delta)$ by removing $\Delta \frac{x_{ig} - f_{ig}}{2}$ copies of each $g \in \widetilde{G}$. Therefore, this removal amount satisfies the constraints of Theorem \ref{thm:non-well-behaved-removal} and we can use its Nash welfare bound. 
This gives us for each $i \in N_{nwb}$,
\begin{align*}
\logNSW(U^i) \ge \logNSW(\cal I^{i, \x}(\Delta)) - \frac1e - 2.
\end{align*}

We can then apply Lemma \ref{lem:w-worst} and \ref{lem:y-substitute} to the above inequality to get for each $i \in N_{nwb}$:
\begin{align}
\ST(\z, i) \ge \sum_{S \subseteq G} y_{i, S}\ln v_i(S) - \frac1e - 2. \label{eq:nwb-bound-large-good-2}
\end{align}

Adding up \eqref{eq:wb-bound-large-good-2} and \eqref{eq:nwb-bound-large-good-2} gives us
\begin{align}
\ST(\z) &\ge \frac1n \sum_{i \in N} \sum_{S \subseteq G} y_{i, S} \ln(v_i(S)) - \frac1e - \frac{2|N_{nwb}|}{n} - 10\gamma^4 + \frac{10}{\Delta n} \sum_{i \in N_{wb}} |A^i| \label{eq:intermediate-large-good-2}
\end{align}

We lower bound $\sum_{i \in N_{wb}} |A^i|$ using the first statement of Claim \ref{claim:z-addition-case-2}.
\begin{align*}
\sum_{i \in N_{wb}} |A^i| = \Delta \gamma^2 \sum_{i \in N_{wb}} \sum_{g \in \widetilde{G}} \mathbbm{1}\{g \in L_i\} x_{ig} = \Delta \gamma^2 \zeta(\widetilde{G}) \ge \Delta 8\gamma^3 n.
\end{align*} 

The final inequality uses Claim \ref{claim:tildeg-lower-bound} which shows that $\zeta(\widetilde{G}) \ge 8\gamma n$. Plugging this into \eqref{eq:intermediate-large-good-2} and upper bounding $|N_{nwb}|$ using $\gamma^3n$, we get

\begin{align*}
\ST(\z) &\ge \frac1n \sum_{i \in N} \sum_{S \subseteq G} y_{i, S} \ln(v_i(S)) - \frac1e - 10\gamma^4 + 78\gamma^3
\end{align*}

This proves the theorem with the constant $c^*_2 = 78\gamma^3 - 10\gamma^4$ (which is strictly larger than the case 1 constant).

In both cases, note that the allocation $\z$ can trivially be constructed in polynomial time. Moreover, we can easily identify which case to follow in polynomial time as well. 

\subsection{Corollaries of The Large Good Consistency Property}
The large good consistency property by itself is hard to directly apply in future proofs. 
In this subsection, we prove some useful corollaries of the large good consistency property that are easier to apply. 

\begin{corr}\label{cor:well-behaved-large-good-lower-bound}
Let $y$ be a configuration LP solution and $\x$ be its corresponding fractional allocation. If there are at most $\gamma^3 n$ non-well-behaved agents in $\x$ and $\x$ satisfies the large good consistency property, then $|G_{wb}| \ge\left (1 - \frac1e - 12\gamma \right )n$.
\end{corr}
\begin{proof}
From Lemma \ref{lem:fractional-property-a}, we have for each $i \in N_{wb}$, $\sum_{g \in G}\mathbbm{1}\{g \in L_i\}x_{ig} \ge 1 - \frac1e - \gamma$. There are at least $n - \gamma^3 n$ well-behaved agents. Therefore,
\begin{align*}
\sum_{i \in N_{wb}}\sum_{g \in G}\mathbbm{1}\{g \in L_i\}x_{ig} \ge \left (1 - \frac1e - \gamma \right )(1 - \gamma^3)n.
\end{align*}

The set $G$ can be split into $G_{wb}$ and $G_{nwb}$. Using this,
\begin{align*}
\sum_{i \in N_{wb}}\sum_{g \in G_{wb}}\mathbbm{1}\{g \in L_i\}x_{ig} \ge \left (1 - \frac1e - \gamma \right )(1 - \gamma^3)n - \sum_{i \in N_{wb}}\sum_{g \in G_{nwb}}\mathbbm{1}\{g \in L_i\}x_{ig}. 
\end{align*}

The term $\sum_{i \in N_{wb}}\sum_{g \in G_{nwb}}\mathbbm{1}\{g \in L_i\}x_{ig}$ is at most $10\gamma n$ via the large good consistency property. Additionally, the term $\sum_{i \in N_{wb}}\sum_{g \in G_{wb}}\mathbbm{1}\{g \in L_i\}x_{ig}$ is at most $|G_{wb}|$ since $\x$ is a valid fractional allocation. Therefore:
\begin{align*}
|G_{wb}| &\ge \left (1 - \frac1e - \gamma \right )(1 - \gamma^3)n - 10\gamma n\\
&\ge \left ( 1 - \frac1e - 12\gamma \right )n.
\end{align*}
This completes the proof.
\end{proof}

\begin{corr}\label{cor:weird-upper-bound}
Let $y$ be a rational configuration LP solution and $\x$ be its corresponding fractional allocation. If there are at most $\gamma^3 n$ non-well-behaved agents in $\x$ and $\x$ satisfies the large good consistency property, then 
\begin{align*}
\sum_{i \in N_{wb}}\sum_{g \in G \setminus L_i}x_{ig} \sum_{j \in N_{wb}}\mathbbm{1}\{g \in L_j\}x_{jg} \le 11\gamma n.
\end{align*}
\end{corr}
\begin{proof}
We can re-write the left hand side of the desired inequality as follows:
\begin{align}
\sum_{i \in N_{wb}}\sum_{g \in G \setminus L_i}x_{ig} \sum_{j \in N_{wb}}\mathbbm{1}\{g \in L_j\}x_{jg} &= \sum_{g \in G} \sum_{i \in N_{wb}} \mathbbm{1}\{g \notin L_i\} x_{ig} \sum_{j \in N_{wb}} \mathbbm{1}\{g \in L_j\} x_{jg} \notag \\
&= \sum_{g \in G_{wb}} \sum_{i \in N_{wb}} \mathbbm{1}\{g \notin L_i\} x_{ig} \sum_{j \in N_{wb}} \mathbbm{1}\{g \in L_j\} x_{jg} \notag \\
&\qquad + \sum_{g \in G_{nwb}} \sum_{i \in N_{wb}} \mathbbm{1}\{g \notin L_i\} x_{ig} \sum_{j \in N_{wb}} \mathbbm{1}\{g \in L_j\} x_{jg}. \label{eq:weird-upper-bound-0}
\end{align}

Essentially, we split the set of goods $G$ into $G_{wb}$ and $G_{nwb}$. We upper bound each term separately. By the definition of well-behaved large goods,
\begin{align}
\sum_{g \in G_{wb}} \sum_{i \in N_{wb}} \mathbbm{1}\{g \notin L_i\} x_{ig} \sum_{j \in N_{wb}} \mathbbm{1}\{g \in L_j\} x_{jg} &\le \sum_{g \in G_{wb}}\gamma \sum_{j \in N_{wb}} \mathbbm{1}\{g \in L_j\} x_{jg} \notag \\
&= \gamma \sum_{j \in N_{wb}} \sum_{g \in G_{wb}}  \mathbbm{1}\{g \in L_j\} x_{jg} \notag \\
&\le \gamma \sum_{j \in N_{wb}} \ell_j \notag\\
&\le \gamma n. \label{eq:weird-upper-bound-1}
\end{align}

In the final inequality we use the fact that $\ell_j \le 1$ for all agents. Next, by the definition of the large good consistency property,
\begin{align}
\sum_{g \in G_{nwb}} \sum_{i \in N_{wb}} \mathbbm{1}\{g \notin L_i\} x_{ig} \sum_{j \in N_{wb}} \mathbbm{1}\{g \in L_j\} x_{jg} &\le \sum_{g \in G_{nwb}} \sum_{j \in N_{wb}} \mathbbm{1}\{g \in L_j\} x_{jg} \notag \\
&\le 10\gamma n. \label{eq:weird-upper-bound-2}
\end{align}

In the first inequality, we upper bound $\sum_{i \in N_{wb}} \mathbbm{1}\{g \notin L_i\} x_{ig}$ using $1$. 
Plugging in \eqref{eq:weird-upper-bound-1} and \eqref{eq:weird-upper-bound-2} into \eqref{eq:weird-upper-bound-0}, we prove the result.
\end{proof}

\section{The Main Algorithm}\label{sec:main-algo}
In this section, we present our final algorithm which handles all the cases not covered by the previous algorithms. Specifically, our algorithm achieves a better than $e^{1/e}$-approximation of the Nash welfare when
\begin{inparaenum}[(a)]
\item at most $\gamma^3n$ agents are not well-behaved, and
\item the large good consistency property is satisfied.
\end{inparaenum}

\begin{theorem}\label{thm:final-algo}
Let $y$ be a configuration LP solution, and let $\x$ be the corresponding fractional allocation. If at most $\gamma^3 n$ agents are not well-behaved in $\x$ and $\x$ satisfies the large good consistency property, then there is a polynomial time algorithm that outputs a (randomized) fractional allocation $\z$ such that 
\begin{align*}
\E[\ST(\z)] \ge \frac1n \sum_{i \in N, S \subseteq G} y_{i, S} \ln(v_i(S)) - \frac1e + c^*_3,
\end{align*}
for some constant $c^*_3 > 0$.
\end{theorem}

The remainder of this section proves Theorem \ref{thm:final-algo}.
At a high level, the allocation $\z$ is initialized at $\x$ and then modified over three steps. 
\begin{description}
\item[Step 1] We modify the allocation of the large goods of well-behaved agents to create a fractional allocation where most agents have a large good amount either close to $1$ or close to $0$. 
\item[Step 2] We modify the allocation of the non-well-behaved agents by removing a small amount of goods from their bundles.
\item[Step 3] We modify the allocation of small goods of well-behaved agents via a redistribution procedure. We move small goods from agents with a large good amount close to $1$ to agents with a large good amount close to $0$.
\end{description}

\subsection{Step 1: Dependent Partial Rounding for Well-Behaved Agents}\label{sec:step-1}
In this section, we describe the first part of our allocation modification procedure where we change the allocation of the large goods in $\x$. We call our routine {\em dependent partial rounding}. This routine is similar to the dependent rounding procedure of \cite{gandhi2006dependent}. However, the procedure will technically not be a `rounding' procedure since it likely will not allocate the large goods integrally.

Let $N_{wb}$ denote the well-behaved agents in $\x$ and $N_{nwb}$ denote the set of non-well-behaved agents. We will construct a fractional allocation $\z$ using $\x$ such that the two allocations only differ in the allocation of large goods to well-behaved agents. As a result, agents may have different large good amounts in $\x$ and $\z$. We therefore use $\ell^{\x}_i$ and $\ell^{\z}_i$ to denote the large good amounts of $\x$ and $\z$ respectively. We emphasize here that even though $\z$ may have a different set of large goods than $L_i$ (using Definition \ref{def:large-goods-fractional}), we measure $\ell^{\x}_i$ and $\ell^{\z}_i$ using the set $L_i$ --- the set of large goods in the allocation $\x$ for agent $i$. Formally, $\ell^{\z}_i = \sum_{g \in G}\mathbbm{1}\{g \in L_i\}z_{ig}$ and $\ell^{\x}_i = \sum_{g \in G}\mathbbm{1}\{g \in L_i\}x_{ig}$. The dependent partial rounding procedure satisfies the following properties:
\begin{description}
\descitem{(D1)}{d1} \textbf{Expectation.} For each $i \in N_{wb}$ and $g \in L_i$, $\E[z_{ig}] = x_{ig}$.
\descitem{(D2)}{d2} \textbf{Negative Correlation}. For any set of agents $S \subseteq N_{wb}$, we have $\E[\prod_{i \in S} \ell^{\z}_i] \le (1 - 1/e + \gamma)^{|S|}$. 
\descitem{(D3)}{d3} \textbf{Preservation}. For all goods $g$, $\sum_{i \in N_{wb}} \mathbbm{1}\{g \in L_i\} z_{ig} = \sum_{i \in N_{wb}} \mathbbm{1}\{g \in L_i\} x_{ig}$.
\descitem{(D4)}{d4} \textbf{High Variance}. $\sum_{i \in N_{wb}}\left (\ell^{\z}_i\right )^2 \ge \left [\left (1 - \frac1e \right ) - 15\gamma \right ]n$.
\end{description}

(D4) is the property that replaces integrality in the rounding process. If we were to run dependent rounding to allocate the large goods, roughly $(1-1/e)$ fraction of the agents would receive a large good, which means $\sum_{i \in N_{wb}} \left (\ell^{\z}_i\right )^2$ would be roughly around $(1-1/e)$. (D4) shows that our approach comes close (in some loose sense) to allocating large goods integrally. In contrast, note that $\sum_{i \in N_{wb}} \left (\ell^{\x}_i\right )^2 \approx (1-1/e)^2 n$.

We describe the routine below. 
The allocation $\z$ is initialized as the allocation $\x$. We iteratively modify the allocation $\z$, maintaining the following invariants:
\begin{enumerate}[(1)]
\item $z_{ig}$ is never modified if $g \notin L_i$ or $i \notin N_{wb}$ --- only large good allocations of well-behaved agents are modified.
\item $\ell^{\z}_i \in [0, 1]$ throughout for each agent $i$. Moreover, once it reaches $0$ or $1$, it is no longer modified.
\end{enumerate}

This algorithm constructs a bipartite graph $\cal H = (N_{wb} \cup G, E)$ with well-behaved agents on one side and goods on another. Between any well-behaved agent $i$ and good $g$, an edge is created if $\mathbbm1\{g \in L_i\}z_{ig} > 0$ and this edge has weight set to $z_{ig}$. Once this graph is constructed, edges are classified into two types: {\em variable} or {\em fixed}. We start by labelling all edges as variable and proceed by fixing edges based on the following two rules:
\begin{enumerate}[(1)]
\item If there is any good $g$ in the graph with exactly one variable edge incident on it, change the label of this edge to fixed.
\item If there is any agent $i$ with $\ell^{\z}_i = 1$, label all of its incident edges as fixed. 
\end{enumerate}

We apply the above two rules to fix edges till it is no longer possible. Moreover, throughout the algorithm, we will repeatedly apply these two rules every time edge weights in the graph are modified. This results in a graph with a set of fixed edges and a set of variable edges. As given by its name, fixed edges are no longer modified; only variable edges are modified. Moreover, once an edge is fixed it is never changed back to variable. At a high level, we run dependent rounding with the constraint that no agent gets a large good amount greater than one. 

Formally, each iteration consists of two parts: weight update and edge fixing. 

\textbf{Weight Update Phase:}
we find either a cycle $C$ or a maximal path $P$ (if a cycle does not exist) of {\em variable edges} in this graph $\cal H$, which we then decompose into two matchings $M_1$ and $M_2$. We define the following two values:
\begin{align*}
\alpha &= \min\left \{\delta > 0 : (\exists(i,g) \in M_1 : z_{ig} + \delta = 1) \bigvee (\exists(i,g) \in M_2 : z_{ig} - \delta = 0)\right \}, \\
\beta &= \min \left \{\delta > 0 : (\exists(i,g) \in M_1 : z_{ig} - \delta = 0) \bigvee (\exists(i,g) \in M_2 : z_{ig} + \delta = 1)\right\}.
\end{align*}

It is easy to see that the positive reals $\alpha$ and $\beta$ exist whenever there is at least one variable edge in the graph. To add the constraint that no agent ends up with a large good amount greater than one, we refine the $\alpha$ and $\beta$ values. Specifically, assume a maximal path $P$ is chosen by the algorithm and decomposed into matchings $M_1$ and $M_2$. By the rules of our edge fixing procedure, no maximal path of variable edges can have a good $g$ at one of its endpoints. Therefore, there are two agents $i_1$ and $i_2$ at either endpoint of the path. Let $(i_1, g_1)$ denote the (unique) edge in $P$ incident on $i_1$ and let $(i_2, g_2)$ denote the (unique) edge in $P$ incident on $i_2$. Since the graph is bipartite, the path $P$ must have an even number of edges. Therefore, $(i_1, g_1)$ and $(i_2, g_2)$ must be present in different matchings $M_1$ and $M_2$. Assume without loss of generality that $(i_1, g_1) \in M_1$ and $(i_2, g_2) \in M_2$. We define the following refinement of $\alpha$ and $\beta$:

\begin{align*}
\alpha^* &= \min\left \{\alpha, 1 - \ell^{\z}_{i_1} \right \}, \\
\beta^* &= \min \left \{\beta, 1 - \ell^{\z}_{i_2} \right\}.
\end{align*}

These values are positive since $\alpha$ and $\beta$ are positive, and $\ell^{\z}_{i_1}$ and $\ell^{\z}_{i_2}$ are strictly less than $1$ by the rules of our edge fixing procedure. The above update applies only to the case where we choose a maximal path $P$. For the other case where we choose a cycle $C$, we set $\alpha^* = \alpha$ and $\beta^* = \beta$.
We then execute the following random step.

\begin{quote}
With probability $\beta^*/(\alpha^* + \beta^*)$, set $z_{ig} := z_{ig} + \alpha^*$ for all $(i,g) \in M_1$, and $z_{ig} := z_{ig} - \alpha^*$ for all $(i,g) \in M_2$; with the complementary probability of $\alpha^*/(\alpha^* + \beta^*)$, set $z_{ig} := z_{ig} - \beta^*$ for all $(i,g) \in M_1$, and $z_{ig} := z_{ig} + \beta^*$ for all $(i, g) \in M_2$.
\end{quote}

\textbf{Edge Fixing Phase:} At the end of the weight update phase of an iteration, it could be that an edge weight is rounded to either $0$ or $1$. If an edge is rounded to $0$ weight, we remove the edge from the graph. If an edge is rounded to $1$ weight, we apply the edge fixing rules to fix the edge. In both cases, we continue to apply the edge fixing rules till it is no longer possible.

If no edge is rounded to $0$ or $1$, it must be the case that we pick a maximal path $P$ with agents $i_1$ and $i_2$ at the endpoints and do one of the following:
\begin{inparaenum}[(a)]
\item Set $\alpha^* = 1 - \ell^{\z}_{i_1}$ and update $z_{ig} := z_{ig} + \alpha^*$ for all $(i,g) \in M_1$, or
\item Set $\beta^* = 1 - \ell^{\z}_{i_2}$ and update $z_{ig} := z_{ig} + \beta^*$ for all $(i,g) \in M_2$.
\end{inparaenum}
In both cases, at least one agent $i$ has $\ell^{\z}_i$ set to $1$. In this case, we apply the edge fixing rules to fix all the edges incident on $i$, and all other edges that can be fixed. 

The algorithm terminates when there are no remaining variable edges in the graph. At each iteration, the number of variable edges reduces by at least $1$.  Therefore, the algorithm terminates in a polynomial number of iterations, and therefore, polynomial time. 

We show that it satisfies the four required properties. To this end, we define the following notation. We define $z_{ig, k}$ to denote the value of $z_{ig}$ at the start of the $k$-th iteration. We use $z_{ig, \infty}$ to denote the final value of $z_{ig}$. Similarly, we define $\ell^{\z}_{i, k}$ as the value of $\ell^{\z}_i$ at the start of the $k$-th iteration.

The proofs for \ref{d1} and \ref{d2} are very similar to \cite{gandhi2006dependent}, and are therefore relegated to Appendix \ref{apdx:main-algo}.
\begin{restatable}{lemma}{lemdependentroundingexpectation}\label{lem:dependent-rounding-expectation}
The dependent partial rounding procedure satisfies \ref{d1}.
\end{restatable}

\ref{d1} also has the following straightforward but incredibly useful corollary.

\begin{corr}\label{cor:d1}
For any $i \in N_{wb}$, $\E[\ell^{\z}_i] = \ell^{\x}_i \in \left [1 - \frac1e - \gamma, 1-\frac1e + \gamma \right ]$.
\end{corr}
\begin{proof}
This follows from combining Lemma \ref{lem:dependent-rounding-expectation} with Lemma \ref{lem:fractional-property-a}.
\end{proof}

\begin{restatable}{lemma}{lemdependentroundingcorrelation}\label{lem:dependent-rounding-correlation}
The dependent partial rounding procedure satisfies \ref{d2}.
\end{restatable}

\begin{lemma}\label{lem:dependent-rounding-preservation}
The dependent partial rounding procedure satisfies \ref{d3}.
\end{lemma}
\begin{proof}
Consider some good $g$ and some iteration $k$. The only way $\sum_{i \in N_{wb}} \mathbbm1\{g \in L_i\}z_{ig, k}$ changes during iteration $k$ is if $g$ is at the end of the maximal path $P$ chosen at iteration $k$. However, by the rules of the edge fixing phase, all goods $g$ either have zero variable edges incident on it or at least two variable edges incident on it. This implies $g$ can never be at the end of a maximal path of variable edges. Therefore, $\sum_{i \in N_{wb}} \mathbbm1\{g \in L_i\}z_{ig, k}$ remains unchanged during iteration $k$. Since $k$ was chosen arbitrarily, it must hold that for all goods $g$, $\sum_{i \in N_{wb}}\mathbbm{1}\{g \in L_i\}z_{ig, \infty} = \sum_{i \in N_{wb}}\mathbbm{1}\{g \in L_i\}z_{ig,1} = \sum_{i \in N_{wb}}\mathbbm{1}\{g \in L_i\}x_{ig}$. 
\end{proof}

\begin{lemma}\label{lem:dependent-rounding-variance}
The dependent partial rounding procedure satisfies \ref{d4}.
\end{lemma}
\begin{proof}
Let $G_{wb}$ denote the set of well-behaved large goods of the instance. Recall that a large good is well-behaved if $\sum_{i \in N_{wb}} \mathbbm{1}\{g \in L_i\}x_{ig} \ge 1 - \gamma$. There are at least $(1 - 1/e - 12\gamma)n$ well-behaved large goods $|G_{wb}|$ in the instance since it satisfies the large good consistency property (Corollary \ref{cor:well-behaved-large-good-lower-bound}). We define a good $g$ as desirable at the start of iteration $k$ if
\begin{align*}
\sum_{i \in N_{wb}} \mathbbm{1}\{g \in L_i\} \mathbbm{1}\{\text{the $(i, g)$ edge is variable at the start of iteration $k$}\}z_{ig} \ge 1 - \gamma.
\end{align*}

Each well-behaved agent $i$ has $\ell^{\x}_i \le (1 - 1/e + \gamma)$ (Lemma \ref{lem:fractional-property-a}). This implies that each well-behaved large good has a degree of at least $2$, and is therefore desirable at the start of the algorithm. Since no good is desirable at the end of the algorithm (no edge is variable), there is a point during the algorithm for each well-behaved large good $g$ where it stops being desirable.

For each well-behaved large good $g$, we pair it with a unique well-behaved agent $i$ such that $\ell^{\z}_{i, \infty} \ge (1 - \gamma)$. This implies
\begin{align*}
\sum_{i \in N_{wb}} \left (\ell^{\z}_{i, \infty} \right )^2 \ge |G_{wb}|(1-\gamma)^2 \ge \left (1 - \frac1e - 12\gamma \right )n(1 -\gamma)^2 \ge \left (1 - \frac{1}{e} - 15\gamma \right )n,
\end{align*}
which proves the lemma.

In the remainder of this proof, we describe our agent-good pairing process. Consider the point in the algorithm where the well-behaved large good $g$ stops being desirable. The exact point where the good $g$ stops being desirable must correspond to the point where some variable edge incident on $g$ is fixed. This must either be because some agent $i$ has $\ell^{\z}_i$ set to $1$ during the iteration, or the good $g$ has exactly one variable edge incident to it at some point in the algorithm. 

Consider the first case where $g$ stops being desirable because a well-behaved agent $i$ has $\ell^{\z}_i$ set to $1$ during the $k$-th iteration, and the $(i, g)$ edge has positive weight at the end of the $k$-th iteration. 
We pair $g$ with this agent $i$. It clearly satisfies our criteria since $\ell^{\z}_{i, \infty} = \ell^{\z}_{i, k+1} = 1$. 
To show uniqueness, we observe that $i$ must be an endpoint of the path $P$ chosen by the algorithm during the $k$-th iteration. Since all paths are maximal, this must mean that the agent $i$ has exactly one variable edge incident to it and this must correspond to the $(i, g)$ edge that gets fixed at the end of this iteration. This implies that the only edge fixed by $\ell^{\z}_{i, k+1}$ being set to $1$ is the $(i, g)$ edge. 


Next, consider the case where $g$ stops being desirable because it has one variable edge $(i, g)$ incident to it during the edge fixing phase of some iteration $k$, and this edge gets fixed. Let $i$ be the agent that this singular variable edge is incident on. This must mean that $z_{ig, k+1} \ge 1 - \gamma$ since the good $g$ was desirable at the end of the weight update phase of the $k$-th iteration. We pair $g$ with the agent $i$. 
We need to show two things to complete this proof: first, $\ell^{\z}_{i, \infty} \ge 1-\gamma$, and second, no other good is paired with agent $i$. Both of these follow from the observation that $\ell^{\z}_{i, k+1} \ge 1 - \gamma$ and $\ell^{\z}_{i, 1} \le (1 - 1/e + \gamma)$; this implies that $\ell^{\z}_{i, k+1} > \ell^{\z}_{i, 1}$. That is, the large good amount of agent $i$ increased from its initial value of $\ell^{\z}_{i, 1}$. This is only possible if at some iteration, the algorithm chose a path $P$ in the bipartite graph $\cal H$ with the agent $i$ as one of its endpoints during some iteration $k' \le k$. Since paths are chosen maximally, this implies that the agent $i$ has at most one variable edge incident to it at the end of the weight update phase of iteration $k$. This edge is the edge $(i, g)$ that gets fixed at the end of this iteration. Therefore, once the edge is fixed, $i$ no longer has any incident variable edges and will not be updated in future iterations. We can therefore, conclude that $\ell^{\z}_{i, \infty} = \ell^{\z}_{i, k+1} \ge 1- \gamma$. 

In both cases, when we pair a good $g$ with an agent $i$, the agent $i$ has exactly one variable edge (the $(i,g)$ edge) incident on it at the moment before $g$ stops being desirable, and this edge is fixed when $g$ stops being desirable. Therefore, it cannot be that two goods get paired with the same agent. This completes the proof.
\end{proof}

\subsection{Step 2: Shrinking for Non-Well-Behaved Agents}\label{sec:step-2}
In this step, we modify the bundles of non-well-behaved agents using an approach similar to Section \ref{sec:large-good-consistency}. Essentially, we bucket the goods as we did in the ST rounding procedure such that each bucket has a fractional size of at most one. Then we shrink all buckets other than the first bucket by half. Formally, for each agent $i \in N_{nwb}$ and good $g \in G$, we use $f^{\x}_{ig}$ to denote the fractional amount of $g$ in the first bucket of agent $i$ in the allocation $\x$. We do the following for each $i \in N_{nwb}$ and $g \in G$:

\begin{align*}
z_{ig} \gets \frac{x_{ig} + f^{\x}_{ig}}{2}.
\end{align*}

\subsection{Step 3: Small Good Redistribution}\label{sec:step-3}
In this final step, we modify the small good allocations of the well-behaved agents. We do the following for each well-behaved agent $i \in N_{wb}$ and good $g \in G \setminus L_i$:
\begin{align*}
z_{ig} \gets x_{ig}\left [1 - \max \left \{\frac{\ell^{\z}_i}{100} - \gamma, 0 \right \} + \sum_{j \in N_{wb}}\max \left \{\frac{\ell^{\z}_j}{100} - \gamma, 0 \right \}\mathbbm{1}\{g \notin L_j\}x_{jg} + \sum_{j \in N_{nwb}} \frac{x_{jg} - f^{\x}_{jg}}{2} \right ]
\end{align*}

Note that this procedure only modifies the small good allocations of the agents. Therefore, the value $\ell^{\z}_i$ remains unchanged throughout this modification. Intuitively, small goods $g$ are removed in proportion to the large good amount $\ell^{\z}_i$ and redistributed among the agents in proportion to the value $x_{ig}$. This results in agents with high large good amounts losing small goods and agents with low large good amounts gaining small goods.
\subsection{Analysis}
It is easy to see that all three steps can be carried out in polynomial time. Therefore, the analysis consists of only two parts: we first show that $\z$ is a valid allocation, and then we lower bound $\E[\ST(\z, i)]$ for each agent $i$. 

\begin{lemma}
The allocation $\z$ is a valid fractional allocation. That is, $\sum_{i \in N}z_{ig} \le \sum_{i \in N}x_{ig}$ for each good $g \in G$, and $z_{ig} \ge 0$ for each agent $i \in N$ and good $g \in G$.
\end{lemma}
\begin{proof}
Fix a good $g \in G$. We can compute $\sum_{i \in N}z_{ig}$ by partitioning it into three parts.
\begin{align}
\sum_{i \in N}z_{ig} = \sum_{i \in N_{wb}}\mathbbm{1}\{g \in L_i\}z_{ig} + \sum_{i \in N_{wb}}\mathbbm{1}\{g \notin L_i\}z_{ig} + \sum_{i \in N_{nwb}}z_{ig}. \label{eq:valid-1}
\end{align}

In the right hand side, the first term is decided by Step 1 (Section \ref{sec:step-1}), the second term is decided by Step 3 (Section \ref{sec:step-3}) and the third term is decided by Step 2 (Section \ref{sec:step-2}).
From the Preservation property \ref{d3} of Step 1, we have:
\begin{align}
\sum_{i \in N_{wb}}\mathbbm{1}\{g \in L_i\}z_{ig} = \sum_{i \in N_{wb}}\mathbbm{1}\{g \in L_i\}x_{ig}.\label{eq:valid-2}
\end{align}

From the modifications carried out in Step 2, we have:
\begin{align}
\sum_{i \in N_{nwb}}z_{ig} = \sum_{i \in N_{nwb}} \frac{x_{ig} + f^{\x}_{ig}}{2}. \label{eq:valid-3}
\end{align}

Finally, from the modification carried out in Step 3, we have:
\begin{align}
\sum_{i \in N_{wb}}\mathbbm{1}\{g \notin L_i\}z_{ig} &= \sum_{i \in N_{wb}} \mathbbm{1}\{g \notin L_i\}x_{ig} \left [1 - \max \left \{\frac{\ell^{\z}_i}{100} - \gamma, 0 \right \} \right . \notag \\&\qquad \qquad \left .+ \sum_{j \in N_{wb}}\max \left \{\frac{\ell^{\z}_j}{100} - \gamma, 0 \right \}\mathbbm{1}\{g \notin L_j\}x_{jg} + \sum_{j \in N_{nwb}} \frac{x_{jg} - f^{\x}_{jg}}{2} \right ] \notag \\
&\le \sum_{i \in N_{wb}}\mathbbm{1}\{g \notin L_i\}x_{ig}  - \sum_{i \in N_{wb}}\mathbbm{1}\{g \notin L_i\} \max \left \{\frac{\ell^{\z}_i}{100} - \gamma, 0 \right \}x_{ig} \notag \\
&\qquad \qquad +  \sum_{j \in N_{wb}}\max \left \{\frac{\ell^{\z}_j}{100} - \gamma, 0 \right \}\mathbbm{1}\{g \notin L_j\}x_{jg} + \sum_{j \in N_{nwb}}\frac{x_{jg} - f^{\x}_{jg}}{2}  \notag \\
&= \sum_{i \in N_{wb}}\mathbbm{1}\{g \notin L_i\}x_{ig} + \sum_{j \in N_{nwb}}\frac{x_{jg} - f^{\x}_{jg}}{2}. \label{eq:valid-4}
\end{align}

Plugging in \eqref{eq:valid-2}, \eqref{eq:valid-3} and \eqref{eq:valid-4} into \eqref{eq:valid-1}, we get
\begin{align*}
\sum_{i \in N}z_{ig} &= \sum_{i \in N_{wb}}\mathbbm{1}\{g \in L_i\}z_{ig} + \sum_{i \in N_{wb}}\mathbbm{1}\{g \notin L_i\}z_{ig} + \sum_{i \in N_{nwb}}z_{ig} \\
&\le \sum_{i \in N_{wb}}\mathbbm{1}\{g \in L_i\}x_{ig} + \sum_{i \in N_{wb}}\mathbbm{1}\{g \notin L_i\}x_{ig} + \sum_{j \in N_{nwb}}\frac{x_{jg} - f^{\x}_{jg}}{2} + \sum_{i \in N_{nwb}} \frac{x_{ig} + f^{\x}_{ig}}{2} \\
&= \sum_{i \in N}x_{ig}.
\end{align*}

Finally, we show that $z_{ig} \ge 0$ for all $i \in N$ and $g \in G$. This is easy to see if $i \in N_{wb}$ and $g \in L_i$ since the dependent partial rounding scheme ensures that $z_{ig}$ remains non-negative. If $i \in N_{nwb}$, then by Step 2, $z_{ig} = \frac{x_{ig} + f^{\x}_{ig}}2 \ge 0$. If $i \in N_{wb}$ and $g \notin L_i$, then
\begin{align*}
z_{ig} &= x_{ig}\left [1 - \max \left \{\frac{\ell^{\z}_i}{100} - \gamma, 0 \right \} + \sum_{j \in N_{wb}}\max \left \{\frac{\ell^{\z}_j}{100} - \gamma, 0 \right \}\mathbbm{1}\{g \notin L_j\}x_{jg} + \sum_{j \in N_{nwb}} \frac{x_{jg} - f^{\x}_{jg}}{2}\right ] \\
&\ge x_{ig} \left [1 - \max \left \{\frac{\ell^{\z}_i}{100} - \gamma, 0 \right \} \right ] \\
&\ge 0.
\end{align*} 
The last inequality follows from $\ell^{\z}_i \in [0, 1]$. Therefore, the allocation $\z$ is valid. 
\end{proof}

Next, we lower bound $\ST(\z, i)$ for each agent $i$. To analyze $\ST(\z, i)$, we consider the instances $\cal I^{i, \x}(\Delta)$ and $\cal I^{i, \z}(\Delta)$ for some $\Delta$ that is feasible with respect to both $\x$, $\z$ and the configuration LP solution $y$. We let $W^i$ denote the round robin allocation of $\cal I^{i, \x}(\Delta)$ and $U^i$ denote the round robin allocation of $\cal I^{i, \z}(\Delta)$. The main idea here is we lower bound $\logNSW(U^i)$ for each agent $i$ using the techniques developed in Section \ref{sec:add-remove}, and then we lower bound $\ST(\z, i)$ using $\logNSW(U^i)$ (Lemma \ref{lem:w-worst}). 

If $i \in N_{nwb}$, then the only change made to the bundle $\x_i$ is from Step 2. Therefore, the instance $\cal I^{i, \z}(\Delta)$ can be constructed by removing a few goods from $\cal I^{i, \x}(\Delta)$. We can bound $\ST(\z, i)$ using Theorem \ref{thm:non-well-behaved-removal}. We made this exact argument in Section \ref{sec:large-good-consistency} as well but we repeat it here for completeness. 

Let $G^i$ be the set of goods of the instance $\cal I^{i, \x}(\Delta)$. To use Theorem \ref{thm:non-well-behaved-removal}, we need to construct a set of $\Delta$ goods $G^{\Delta}$ which are weakly more valuable than all the other goods of the instance. We construct $G^{\Delta}$ as the set which contains $\Delta f^{\x}_{ig}$ copies of each good $g$. By the definition of our bucketing procedure, this creates a set of $\Delta$ goods $G^{\Delta}$ such that each good in $G^{\Delta}$ is weakly more valuable than each good in $G^i \setminus G^{\Delta}$.

We partition the goods $G^i$ based on the copy of the good $g \in G$ they are. Note that there are $\Delta x_{ig}$ copies of each good $g \in G$ in the instance $\cal I^{i, \x}(\Delta)$. We obtain $\cal I^{i, \z}(\Delta)$ from $\cal I^{i, \x}(\Delta)$ by removing $\Delta \frac{x_{ig} - f_{ig}}{2}$ copies of each $g \in G$. Therefore, this removal amount satisfies the constraints of Theorem \ref{thm:non-well-behaved-removal} and we can use its Nash welfare bound. 
This gives us for each $i \in N_{nwb}$,
\begin{align*}
\logNSW(U^i) \ge \logNSW(\cal I^{i, \x}(\Delta)) - \frac1e - 2.
\end{align*}

We can upper bound $\logNSW(U^i)$ using Lemma \ref{lem:w-worst}, and we can lower bound $\logNSW(\cal I^{i, \x}(\Delta))$ using Lemma \ref{lem:y-substitute}. This gives us
\begin{align}
\ST(\z, i) \ge \sum_{S \subseteq G}y_{i, S}\ln v_i(S) - \frac1e - 2. \label{eq:final-nwb-bound}
\end{align}

Next, we analyze the case where $i \in N_{wb}$. We again compare $\cal I^{i, \z}(\Delta)$ and $\cal I^{i, \x}(\Delta)$. However in this case, there are more significant differences between the two instances. The set of large goods of $\cal I^{i, \x}(\Delta)$ is replaced with a new set of large goods decided by Step 1, where the instance gets $\Delta z_{ig}$ copies of each good $g \in L_i$\footnote{We remind the reader again that the set of large goods of the allocation $\z$ is still defined with respect to the allocation $\x$. Therefore, we count the copies of the goods in $L_i$ and do not include any other goods. }. Additionally, some small goods are added and removed based on Step 3. Specifically, for each good $g \in G \setminus L_i$, we remove $\Delta x_{ig} \max \left \{\frac{\ell^{\z}_i}{100} - \gamma, 0 \right \}$ copies and add $\Delta x_{ig}\left [\sum_{j \in N_{wb}}\max \left \{\frac{\ell^{\z}_j}{100} - \gamma, 0 \right \}\mathbbm{1}\{g \notin L_j\}x_{jg} + \sum_{j \in N_{nwb}} \frac{x_{jg} - f^{\x}_{jg}}{2} \right ]$ copies to the instance $\cal I^{i, \x}(\Delta)$.

To bound $\ST(\z, i)$, we invoke Theorems \ref{thm:small-good-addition-large-good-removal} and \ref{thm:general-add-removal}. We let $R^i$ denote the set of removed small goods, $A^i$ denote the added small goods, and $L^*$ denote the set of new large goods. 
Note that $|L^*| = \sum_{g \in L_i}\Delta z_{ig} = \ell^{\z}_i \Delta$. We consider two different cases depending on whether $\ell^{\z}_i \ge 100\gamma$. 

Let $N_1$ denote the set of well-behaved agents where $\ell^{\z}_i > 100\gamma$ and let $N_2$ denote the set of well-behaved agents where $\ell^{\z}_i \le 100\gamma$. We deal with agents in $N_1$ first. 

To bound $\logNSW(U^i)$, we apply Theorem \ref{thm:general-add-removal}. There are three conditions required by the theorem. 
First, $|L^*|$ must lie in the interval $[100\gamma\Delta, \Delta]$. This holds because $|L^*| = \ell^{\z}_i{\Delta}$ and $\ell^{\z}_i$ lies in the interval $(100\gamma, 1]$. Second, we must have $v_i(R^i) \le \frac{|L^*|\mu}{100}$. This holds due to the following observation.

\begin{obs}\label{obs:ri-upper-bound}
$v_i(R^i) \le \frac{|L^*|\mu}{100} = \frac{\ell^{\z}_i \mu^i \Delta}{100}$.
\end{obs} 
\begin{proof}
$R^i$ consists of $\Delta x_{ig} \max \left \{\frac{\ell^{\z}_i}{100} - \gamma, 0 \right \}$ copies of each good $g \in G \setminus L_i$. Therefore, 
\begin{align*}
v_i(R^i) = \sum_{g \in G \setminus L_i} \Delta x_{ig} v_i(g)\max \left \{\frac{\ell^{\z}_i}{100} - \gamma, 0 \right \} \le \frac{\ell^\z_i \Delta}{100} \sum_{g \in G \setminus L_i} x_{ig}v_i(g) = \frac{\ell^{\z}_i\mu^i \Delta}{100}.
\end{align*}

The observation follows from noting that $\ell^{\z}_i\Delta = |L^*|$ and $\mu^i = \mu$ (Lemma \ref{lem:large-good-connection}).
\end{proof}
Third and final, we must have $v(A^i) \le \mu^i \Delta$. This holds since
\begin{align*}
v_i(A^i) &= \sum_{g \in G \setminus L_i} \Delta x_{ig}\left [\sum_{j \in N_{wb}}\max \left \{\frac{\ell^{\z}_j}{100} - \gamma, 0 \right \}\mathbbm{1}\{g \notin L_j\}x_{jg} + \sum_{j \in N_{nwb}} \frac{x_{jg} - f^{\x}_{jg}}{2} \right ] v_i(g) \\
&\le \sum_{g \in G \setminus L_i} \Delta x_{ig} \left [\frac{\sum_{j \in N} x_{jg}}{2} \right ] v_i(g) \\
&\le \Delta \sum_{g \in G \setminus L_i} x_{ig}v_i(g) = \Delta \mu^i.
\end{align*}
In the first inequality, we use the fact that $\ell^{\z}_j \le 1$, which implies, $\sum_{j \in N_{wb}}\frac{\ell^{\z}_j}{100} \mathbbm{1}\{g \notin L_j\}x_{jg} \le \frac{1}{100}$.
Therefore, we can apply Theorem \ref{thm:general-add-removal} to conclude that
\begin{align*}
\logNSW(U^{i}) \ge \frac{1}{\Delta} \left [\sum_{g \in L^*} \ln v_i(g) + (\Delta - |L^*|)\ln \mu \right ] - \frac{3(\Delta - |L^*|)v_i(R^i)}{\Delta^2 \mu} + \frac{(\Delta - |L^*|)v_i(A^i)}{10\Delta^2 \mu} - 15\gamma.
\end{align*}

To translate this result to the terminology of $\z$, we use the fact that $\mu = \mu^i$, $|L^*| = \ell^{\z}_i \Delta$ and $L^*$ consists of $\Delta z_{ig}$ copies of each good $g \in L_i$. This gives us for each $i \in N_1$,
\begin{align*}
\logNSW(U^{i}) \ge \sum_{g \in L_i} z_{ig}\ln v_i(g) + (1 - \ell^{\z}_i)\ln \mu^i - \frac{3(1- \ell^{\z}_i)v_i(R^i)}{\Delta \mu^i} + \frac{(1 - \ell^{\z}_i)v_i(A^i)}{10\Delta \mu^i} - 15\gamma.
\end{align*} 

We upper bound $\logNSW(U^{i})$ using Lemma \ref{lem:w-worst}, and we upper bound $v_i(R^i)$ using Observation \ref{obs:ri-upper-bound}. This gives us:
\begin{align}
\ST(\z, i) \ge \sum_{g\in L_i} z_{ig}\ln v_i(g) + (1 - \ell^{\z}_i)\ln \mu^i - \frac{(1 - \ell^{\z}_i)\ell^{\z}_i}{30}  + \frac{(1- \ell^{\z}_i)v_i(A^i)}{10\Delta \mu^i} - 15\gamma.\label{eq:n1-before-expectation}
\end{align}

Next, we turn to agents in $N_2$. Consider some agent $i \in N_2$. Since $\ell^{\z}_i \le 100\gamma$, we have $\max\{\frac{\ell^{\z}_i}{100}- \gamma, 0\} = 0$. This means the set of removed goods $R^i$ is empty. We use Theorem \ref{thm:small-good-addition-large-good-removal} to bound $\logNSW(U^i)$. It is easy to see from the arguments of the previous case that the conditions of Theorem \ref{thm:small-good-addition-large-good-removal} hold ($v_i(A^i) \le \Delta \mu$ and $|L^*| \le \Delta$).

Therefore, we can apply Theorem \ref{thm:small-good-addition-large-good-removal} to conclude that
\begin{align*}
\logNSW(U^{i}) \ge \frac{1}{\Delta} \left [\sum_{g \in L^*} \ln v_i(g) + (\Delta - |L^*|)\ln \mu \right ] + \frac{(\Delta - |L^*|)v_i(A^i)}{5\Delta^2 \mu} - 12\gamma.
\end{align*}

To translate this result to the terminology of $\z$, we use the fact that $\mu = \mu^i$, $|L^*| = \ell^{\z}_i \Delta$ and $L^*$ consists of $\Delta z_{ig}$ copies of each good $g \in L_i$. This gives us for each $i \in N_2$,
\begin{align*}
\logNSW(U^{i}) \ge \sum_{g \in L_i} z_{ig}\ln v_i(g) + (1 - \ell^{\z}_i)\ln \mu^i + \frac{(1 - \ell^{\z}_i)v_i(A^i)}{5\Delta \mu^i} - 12\gamma.
\end{align*} 

We upper bound $\logNSW(U^{i})$ using Lemma \ref{lem:w-worst}. This gives us for each $i \in N_{2}$:
\begin{align}
\ST(\z, i) \ge \sum_{g\in L_i} z_{ig}\ln v_i(g) + (1 - \ell^{\z}_i)\ln \mu^i  + \frac{(1- \ell^{\z}_i)v_i(A^i)}{5\Delta \mu^i} - 12\gamma. \label{eq:n2-before-expectation}
\end{align}

We add up all the $\ST(\z, i)$ values before simplifying it further. Specifically, we combine \eqref{eq:n1-before-expectation} and \eqref{eq:n2-before-expectation} to get
\begin{align*}
\sum_{i \in N_{wb}} \ST(\z, i) &\ge \sum_{i \in N_{wb}}\sum_{g\in L_i} z_{ig}\ln v_i(g) + \sum_{i \in N_{wb}}(1 - \ell^{\z}_i)\ln \mu^i \\&\qquad \qquad - \sum_{i \in N_{wb}}\frac{(1 - \ell^{\z}_i)\ell^{\z}_i}{30}  + \sum_{i \in N_{wb}}\frac{(1- \ell^{\z}_i)v_i(A^i)}{10\Delta \mu^i} - 15\gamma n.
\end{align*}

Next, we take an expectation on both sides. This gives us:
\begin{align*}
\sum_{i \in N_{wb}}\E[\ST(\z, i)] &\ge \sum_{i \in N_{wb}} \sum_{g\in L_i} \E[z_{ig}\ln v_i(g)] + \sum_{i \in N_{wb}} \E[(1 - \ell^{\z}_i)\ln \mu^i] \\ & \qquad \qquad - \sum_{i \in N_{wb}}\frac{\E[(1 - \ell^{\z}_i)\ell^{\z}_i]}{30}  + \sum_{i \in N_{wb}}\frac{\E[(1- \ell^{\z}_i)v_i(A^i)]}{10\Delta \mu^i} - 15\gamma n.
\end{align*}
This expectation is over the randomness in Step 1.
To simplify this expression, we apply $\E[z_{ig}] = x_{ig}$ and $\E[\ell^{\z}_i] = \E[\ell^{\x}_i]$ from the Expectation property \ref{d1} and Corollary \ref{cor:d1}. This gives us
\begin{align*}
\sum_{i \in N_{wb}} \E[\ST(\z, i)] &\ge \sum_{i \in N_{wb}}\sum_{g\in L_i} x_{ig}\ln v_i(g) + \sum_{i \in N_{wb}} (1 - \ell^{\x}_i)\ln \mu^i \\ &\qquad \qquad - \sum_{i \in N_{wb}} \frac{\E[(1 - \ell^{\z}_i)\ell^{\z}_i]}{30} + \sum_{i \in N_{wb}} \frac{\E[(1 - \ell^{\z}_i)v_i(A^i)]}{10\Delta \mu^i} - 15\gamma n. 
\end{align*}

From Lemma \ref{lem:fractional-nash-upper-bound}, we have $\sum_{g\in L_i} x_{ig}\ln v_i(g) + (1 - \ell^{\x}_i)\ln \mu^i \ge \sum_{S \subseteq G} y_{i, S}\ln v_i(S) - \frac1e$ for each agent $i \in N$. Plugging this in,
\begin{align}
\sum_{i \in N_{wb}} \E[\ST(\z, i)] \ge \sum_{i \in N_{wb}, S \subseteq G}y_{i, S}\ln v_i(S) - \frac{|N_{wb}|}{e} - \sum_{i \in N_{wb}} \frac{\E[(1 - \ell^{\z}_i)\ell^{\z}_i]}{30} + \sum_{i \in N_{wb}} \frac{\E[(1 - \ell^{\z}_i)v_i(A^i)]}{10\Delta \mu^i} - 15\gamma n. \label{eq:before-combine}
\end{align}

We can now lower bound $\E[\ST(\z)]$ by combining \eqref{eq:final-nwb-bound} and \eqref{eq:before-combine}. 
\begin{align}
\E[\ST(\z)] \ge \frac1n \sum_{i \in N, S \subseteq G}y_{i, S}\ln v_i(S) - \frac{1}{e} - \sum_{i \in N_{wb}} \frac{\E[(1 - \ell^{\z}_i)\ell^{\z}_i]}{30n} + \sum_{i \in N_{wb}} \frac{\E[(1 - \ell^{\z}_i)v_i(A^i)]}{10\Delta \mu^i n} - 15\gamma - \frac{2|N_{nwb}|}{n}. \label{eq:stz-initial-lowerbound}
\end{align}

To simplify this further, we need to analyze each of the last four terms separately. To start with, we have $|N_{nwb}| \le \gamma^3 n$. Next, we lower bound, $v_i(A^i)$ using the following lemma. 
\begin{lemma}\label{lem:ai-upper-bound}
For each agent $i \in N_{wb}$, the following holds:
\begin{align*}
\frac{\E[(1 - \ell^{\z}_i)v_i(A^i)]}{\Delta \mu^i} \ge 10^{-4} - \sum_{g \in G \setminus L_i} x_{ig} \sum_{j \in N_{wb}} \frac{\mathbbm{1}\{g \in L_j\}x_{jg}}{30} - 3 \sum_{g \in G \setminus L_i} x_{ig} \sum_{j \in N_{nwb}} f^{\x}_{jg}.
\end{align*}
\end{lemma}
\begin{proof}
$A^i$ consists of $\Delta x_{ig}\left [\sum_{j \in N_{wb}}\max \left \{\frac{\ell^{\z}_j}{100} - \gamma, 0 \right \}\mathbbm{1}\{g \notin L_j\}x_{jg} + \sum_{j \in N_{nwb}} \frac{x_{jg} - f^{\x}_{jg}}{2} \right ]$ copies of each good $g \in G \setminus L_i$. Therefore,
\begin{align*}
\frac{v_i(A^i)}{\Delta\mu^i} &= \sum_{g \in G \setminus L_i} \frac{x_{ig}v_i(g)}{\mu^i}\left [\sum_{j \in N_{wb}}\max \left \{\frac{\ell^{\z}_j}{100} - \gamma, 0 \right \}\mathbbm{1}\{g \notin L_j\}x_{jg} + \sum_{j \in N_{nwb}} \frac{x_{jg} - f^{\x}_{jg}}{2} \right ] \\
&\ge \sum_{g \in G \setminus L_i} \frac{x_{ig}v_i(g)}{\mu^i}\left [\sum_{j \in N_{wb} \setminus \{i\}}\left (\frac{\ell^{\z}_j}{100} - \gamma \right )x_{jg} - \sum_{j \in N_{wb}}\left (\frac{\ell^{\z}_j}{100} - \gamma \right ) \mathbbm{1}\{g \in L_j\} x_{jg} + \sum_{j \in N_{nwb}} \frac{x_{jg} - f^{\x}_{jg}}{2} \right ]  \\
\end{align*}

To simplify this expression, we separate the terms with negative coefficients and upper bound them. 
\begin{align*}
\frac{v_i(A^i)}{\Delta \mu^i} &\ge \sum_{g \in G \setminus L_i} \frac{x_{ig}v_i(g)}{\mu^i}\left [\sum_{j \in N_{wb} \setminus \{i\}}\left (\frac{\ell^{\z}_j}{100}\right )x_{jg} + \sum_{j \in N_{nwb}} \frac{x_{jg}}{2} -\gamma \right ] \\
& \qquad \qquad - \sum_{g \in G \setminus L_i} \frac{x_{ig}v_i(g)}{\mu^i} \left [\sum_{j \in N_{wb}} \left (\frac{\ell^{\z}_j}{100} \right )\mathbbm{1}\{g \in L_j\}x_{jg} + \sum_{j \in N_{nwb}} \frac{f^{\x}_{jg}}{2} \right ]
\end{align*}

We use the fact that $v_i(g) \le \frac{\mu^i}{1-\ell^{\x}_i} \le 3\mu^i$ for all small goods $g$ since $i$ is well-behaved and Lemma \ref{lem:fractional-property-a} holds. This yields
\begin{align*}
\frac{v_i(A^i)}{\Delta \mu^i} &\ge \sum_{g \in G \setminus L_i} \frac{x_{ig}v_i(g)}{\mu^i}\left [\sum_{j \in N_{wb} \setminus \{i\}}\left (\frac{\ell^{\z}_j}{100}\right )x_{jg} + \sum_{j \in N_{nwb}} \frac{x_{jg}}{2} - \gamma \right ] \\
& \qquad \qquad - \sum_{g \in G \setminus L_i} 3x_{ig} \left [\sum_{j \in N_{wb}} \frac{\mathbbm{1}\{g \in L_j\}x_{jg}}{100} + \sum_{j \in N_{nwb}} f^{\x}_{jg}\right ]. \\
&\ge \sum_{g \in G \setminus L_i} \frac{x_{ig}v_i(g)}{\mu^i}\left [\sum_{j \in N_{wb} \setminus \{i\}}\left (\frac{\ell^{\z}_j}{100}\right )x_{jg} + \sum_{j \in N_{nwb}} \frac{x_{jg}}{2} -\gamma\right ] \\
& \qquad \qquad  - \sum_{g \in G \setminus L_i} x_{ig} \sum_{j \in N_{wb}} \frac{\mathbbm{1}\{g \in L_j\}x_{jg}}{30} - 3 \sum_{g \in G \setminus L_i} x_{ig} \sum_{j \in N_{nwb}} f^{\x}_{jg}.
\end{align*}

We use $C$ to denote the term $\sum_{g \in G \setminus L_i} x_{ig} \sum_{j \in N_{wb}} \frac{\mathbbm{1}\{g \in L_j\}x_{jg}}{30} + 3 \sum_{g \in G \setminus L_i} x_{ig} \sum_{j \in N_{nwb}} f^{\x}_{jg}$. Note that $C$ is independent of the randomness of Step 1. If we multiply both sides of the above inequality by $(1-\ell^{\z}_i)$ and take an expectation, we get
\begin{align*}
\frac{\E[(1-\ell^{\z}_i) v_i(A^i)]}{\Delta \mu^i} &\ge \sum_{g \in G \setminus L_i} \frac{x_{ig}v_i(g)}{\mu^i}\left [\sum_{j \in N_{wb} \setminus \{i\}}\left (\frac{\E[(1-\ell^{\z}_i)\ell^{\z}_j]}{100}\right )x_{jg} + \sum_{j \in N_{nwb}} \frac{\E[1-\ell^{\z}_i]x_{jg}}{2} - \E[1-\ell^{\z}_i] \gamma\right ] \\
&\qquad \qquad - \E[1-\ell^{\z}_i]C
\end{align*}

Here, we use Negative Correlation \ref{d2} to upper bound $\E[\ell^{\z}_i \ell^{\z}_j]$ and we use Corollary \ref{cor:d1} to bound $\E[\ell^{\z}_i]$. This simplifies our expression to

\begin{align*}
\frac{\E[(1-\ell^{\z}_i) v_i(A^i)]}{\Delta \mu^i} &\ge \sum_{g \in G \setminus L_i} \frac{x_{ig}v_i(g)}{\mu^i}\left [\sum_{j \in N \setminus \{i\}}\frac{x_{jg}}{1000} -  \gamma\right ] - C
\end{align*}

We separate out the $\gamma$ term and write $\sum_{j \in N \setminus \{i\}}x_{jg} = 1 - x_{ig}$, which holds due to Observation \ref{obs:x-complete}. Using $\mu^i = \sum_{g \in G \setminus L_i}x_{ig}v_i(g)$, we get

\begin{align*}
\frac{\E[(1-\ell^{\z}_i) v_i(A^i)]}{\Delta \mu^i} &\ge \sum_{g \in G \setminus L_i} \frac{x_{ig}(1-x_{ig})v_i(g)}{1000\mu^i} - \gamma - C
\end{align*}

We complete our analysis by only considering goods in $G \setminus L_i$ where $x_{ig} \le 0.5$. For all of these goods $g$, $(1 - x_{ig}) \ge 1/2$. Moreover, $\sum_{g \in G\setminus L_i}\mathbbm{1}\{x_{ig} \le 0.5\} x_{ig}v_i(g) \ge (1-\gamma)\mu^i$ via the small good fraction property (Definition \ref{def:small-good-fraction}). Therefore, 

\begin{align*}
\frac{\E[(1-\ell^{\z}_i) v_i(A^i)]}{\Delta \mu^i} &\ge \sum_{g \in G \setminus L_i} \mathbbm{1}\{x_{ig} \le 0.5\}\frac{x_{ig}(1-x_{ig})v_i(g)}{1000\mu^i} - \gamma - C \\
&\ge \sum_{g \in G \setminus L_i} \mathbbm{1}\{x_{ig} \le 0.5\}\frac{x_{ig}v_i(g)}{2000\mu^i} -\gamma- C \\
&\ge \frac{(1-\gamma)}{2000} - \gamma - C \\
&\ge 10^{-4} - C.
\end{align*}
This completes the proof.
\end{proof}

We apply Lemma \ref{lem:ai-upper-bound} to \eqref{eq:stz-initial-lowerbound} and note that $|N_{wb}| \in \left [(1-\gamma^3)n, n \right ]$. This gives us:

\begin{align*}
\E[\ST(\z)] &\ge \frac{1}{n}\sum_{i \in N, S \subseteq G}y_{i, S}\ln v_i(S)  - \frac1e + 10^{-5} -  \sum_{i \in N_{wb}}\frac{\E[\ell^{\z}_i(1 - \ell^{\z}_i)]}{30 n} - 17\gamma \\ &\qquad \qquad -  \sum_{i \in N_{wb}} \sum_{g \in G \setminus L_i} x_{ig} \sum_{j \in N_{wb}} \frac{\mathbbm{1}\{g \in L_j\}x_{jg}}{30 n} - \frac{3}{n} \sum_{i \in N_{wb}} \sum_{g \in G \setminus L_i} x_{ig} \sum_{j \in N_{nwb}} f^{\x}_{jg}
\end{align*}

While this may initially seem scary, all of the previous sections of this paper have been written solely to ensure that every term in the above inequality with a negative coefficient (with the exception of $\frac1e$) is at most $\gamma$.

To start with, $\sum_{i \in N_{wb}} \sum_{g \in G \setminus L_i} x_{ig} \sum_{j \in N_{wb}}\mathbbm{1}\{g \in L_j\}x_{jg} \le 11\gamma n$ from Corollary \ref{cor:weird-upper-bound}. This simplifies the above expression slightly to
\begin{align}
\E[\ST(\z)] &\ge \frac{1}{n}\sum_{i \in N, S \subseteq G}y_{i, S}\ln v_i(S)  - \frac1e + 10^{-5} \notag\\ & \qquad \qquad - \sum_{i \in N_{wb}}\frac{\E[\ell^{\z}_i(1 - \ell^{\z}_i)]}{30n} - 18\gamma  - \frac{3}{n} \sum_{i \in N_{wb}} \sum_{g \in G \setminus L_i} x_{ig} \sum_{j \in N_{nwb}} f^{\x}_{jg}. \label{eq:final-bound-0}
\end{align}

To upper bound the remaining two scary terms, we prove the following lemmas.
\begin{lemma}\label{lem:n1-loss-upper-bound}
$\sum_{i \in N_{wb}}\frac{\E[\ell^{\z}_i(1 - \ell^{\z}_i)]}{30 n} \le \gamma.$
\end{lemma}
\begin{proof}
This holds essentially due to the High Variance property \ref{d4} and Corollary \ref{cor:d1}.
\begin{align*}
\sum_{i \in N_{wb}}\frac{\E[\ell^{\z}_i(1 - \ell^{\z}_i)]}{30 n} &\le \frac1{30 n}\E\left [ \sum_{i \in N_{wb}}\ell^{\z}_i - \sum_{i \in N_{wb}}\left (\ell^{\z}_i \right)^2 \right ] \\
&\le \frac1{30 n}\left [ \left (1 - \frac1e + \gamma \right )n - \left (1-\frac1e - 15\gamma \right )n \right ] \\
&\le \gamma.
\end{align*}
In the second inequality, we use the High Variance property \ref{d4} and Corollary \ref{cor:d1}. 
\end{proof}
\begin{lemma}\label{lem:fj-upper-bound}
$\frac{3}{n} \sum_{i \in N_{wb}} \sum_{g \in G \setminus L_i} x_{ig} \sum_{j \in N_{nwb}} f^{\x}_{jg} \le \gamma.$
\end{lemma}
\begin{proof}
The lemma holds because $|N_{nwb}| \le \gamma^3 n$ and the following sequence of inequalities.
\begin{align*}
\frac{3}{n} \sum_{i \in N_{wb}} \sum_{g \in G \setminus L_i} x_{ig} \sum_{j \in N_{nwb}} f^{\x}_{jg} &\le \frac{3}{n} \sum_{g \in G} \sum_{i \in N_{wb}} x_{ig} \sum_{j \in N_{nwb}} f^{\x}_{jg} \\
&\le \frac{3}{n}  \sum_{g \in G} \sum_{j \in N_{nwb}} f^{\x}_{jg} \\
&\le \frac{3}{n}  |N_{nwb}| \le \gamma.
\end{align*}
In the second inequality we upper bound $\sum_{i \in N_{wb}} x_{ig}$ using $1$. In the third inequality, we use the fact that the first bucket has size at most $1$ for any agent $i$. In the final inequality, we use $|N_{nwb}| \le \gamma^3 n$.
\end{proof}

These lemmas simplify \eqref{eq:final-bound-0} to 
\begin{align*}
\E[\ST(\z)] &\ge \frac{1}{n}\sum_{i \in N, S \subseteq G}y_{i, S}\ln v_i(S)  - \frac1e + 10^{-5} - 20\gamma \\
&\ge \frac{1}{n}\sum_{i \in N, S \subseteq G}y_{i, S}\ln v_i(S)  - \frac1e + 10^{-6}.
\end{align*}

This completes the proof of Theorem \ref{thm:final-algo}.

\section{Putting the Pieces Together}\label{sec:main-theorem}

We are finally ready to prove the main result of this paper:
\thmnash*
\begin{proof}
Let $\vare > 0$ be some constant that we will decide later. Let $\cal I$ be an instance of the fair allocation problem with additive valuations. 
Given the instance $\cal I$, let $y$ be a rational configuration LP solution which is optimal up to a $\ln{(1 + \vare)}$ additive error (computed using Theorem \ref{thm:configuration-lp}). We can assume that this solution $y$ has a well-defined objective value due to Remark \ref{rem:well-defined-objective}. Let $\x$ be its corresponding fractional allocation. 

We compute a fractional allocation $\z$ such that for some constant $c^* > 0$,
\begin{align*}
\E[\ST(\z)] &\ge \frac1n \sum_{i \in N, S\subseteq G} y_{i, S} \ln v_i(S) - \frac{1}{e} + c^*, \\
&\ge \logNSW(\cal I) - \ln(1 + \vare) - \frac{1}{e} + c^*.
\end{align*} 
If more than $\gamma^3 n$ agents are not well-behaved in $\x$, then $\z = \x$ (Theorem \ref{thm:non-well-behaved-algo}). If at most $\gamma^3 n$ agents are not well-behaved in $\x$ but $\x$ does not satisfy the large good consistency, we compute $\z$ using Theorem \ref{thm:no-large-good-consistency}. If at most $\gamma^3 n$ agents are not well-behaved in $\x$ and $\x$ satisfies the large good consistency property, we compute $\z$ using Theorem \ref{thm:final-algo}. 

Then we round $\z$ using the ST rounding procedure to create an integral allocation $X$. Via Fact \ref{fact:log-nash-welfare},

\begin{align*}
\E[\NSW(X)] \ge e^{\E[\logNSW(X)]} = e^{\E[\ST(\z)]} \ge \exp\left \{\logNSW(\cal I) - \ln(1 + \vare) - \frac{1}{e} + c^*\right \}.
\end{align*}

Setting $\vare = \frac{c^*}{2}$ and upper bounding $\ln(1 + \vare)$ by $\vare$, we get:
\begin{align*}
\E[\NSW(X)] \ge e^{\E[\logNSW(X)]} \ge \exp\left \{\logNSW(\cal I) - \frac{1}{e} + \frac{c^*}{2} \right \} = \frac{\NSW(\cal I)}{e^{\frac{1}{e} - \frac{c^*}{2}}}.
\end{align*}

Every single step of the above process can be implemented in polynomial time. We start by finding a configuration LP solution using Theorem \ref{thm:configuration-lp} and its corresponding fractional allocation $\x$. Then, we decide which algorithm to use to compute $\z$ --- this can be done in polynomial time since the set of well-behaved agents can be found in polynomial time (Corollary \ref{cor:well-behaved-identification}) and the large good consistency property can be verified in polynomial time. Then we compute $\z$; all the algorithms we present for computing $\z$ are polynomial time algorithms. Finally, Shmoys-Tardos rounding can be implemented in polynomial time using any of the standard techniques \cite{birkhoff1946bvndecomposition,gandhi2006dependent}.
\end{proof}

\begin{remark}
A loose estimate of the constant $c$ puts it at around $10^{-80}$. Throughout this paper, we have made no attempt to optimize this constant, choosing instead to simplify the proof as much as possible. Even if optimized, it seems unlikely that our technique will result in a constant better than $0.001$. 
\end{remark}

\paragraph{Acknowledgements} The author is funded by the National Science Foundation (NSF) Career Award 2441296 and Grant RI-2327057.

\bibliographystyle{alpha}
\bibliography{abb,references}

\newpage
\appendix

\section{Missing Proofs from Section \ref{sec:round-robin-guarantees}}\label{apdx:round-robin-guarantees}

\obslargegoodsunique*
\begin{proof}
We first show that there always exists a set $L$ satisfying the properties of Definition \ref{def:large-goods}. This set $L$ can be constructed iteratively. We start by initializing $L$ as the empty set, and while there is some good $g \in G \setminus L$ such that $v(g) > \frac{v(G \setminus L)}{\Delta - \ell}$ (where $\ell = |L|$), we add this $g$ to $L$. Note that this addition only reduces the value of $\frac{v(G \setminus L)}{\Delta - \ell}$ and therefore, does not violate the largeness of the goods already in $L$. 

\begin{align*}
\frac{v(G \setminus (L \cup \{g\}))}{\Delta - \ell - 1} = \frac{v(G \setminus L) - v(g)}{\Delta - \ell - 1} < \frac{v(G \setminus L) - \frac{v(G \setminus L)}{\Delta - \ell}}{\Delta - \ell - 1} = \frac{v(G \setminus L)}{\Delta - \ell}.
\end{align*}

This iterative procedure stops when $|L| = \Delta - 1$ because $v(g) \le v(G \setminus L)$ for any $g \in G \setminus L$. Therefore, $L$ satisfies Definition \ref{def:large-goods}. 

We show uniqueness next. Assume there are two different sets of large goods $L$ and $L'$ satisfying the properties of Definition \ref{def:large-goods}. Note that if any good $g$ is in $L$ (resp. $L'$), all goods with value at least $v(g)$ must also be in $L$ (resp. $L'$). Therefore, let $L = \{g_1, \dots, g_{\ell}\}$ and let $L' = \{g_1, \dots, g_{\ell'}\}$, and assume $\ell < \ell'$. By applying the properties of large goods with the set $L$, we get for any $g \in L' \setminus L$, 
\begin{align*}
v(g) \le \frac{v(G \setminus L)}{\Delta - \ell} = \frac{v(G \setminus L') + v(L' \setminus L)}{\Delta - \ell}.
\end{align*}

Adding the above inequality up for all $g \in L' \setminus L$, we get
\begin{align*}
v(L' \setminus L) \le \left (\frac{\ell' - \ell}{\Delta - \ell} \right ) (v(G \setminus L') + v(L' \setminus L)).
\end{align*}

Re-arranging terms gives us
\begin{align*}
v(L' \setminus L) \le \frac{(\ell' - \ell)v(G \setminus L')}{\Delta - \ell'}.
\end{align*}

This contradicts the set $L'$ being a set of large goods since each good $g \in L'\setminus L$ satisfies
\begin{equation*}
v(g) > \frac{v(G \setminus L')}{\Delta - \ell'}. \qedhere
\end{equation*}
\end{proof}

\section{Missing Proofs from Section \ref{sec:add-remove}}\label{apdx:add-remove}

\thmsmallgoodadditionlargegoodremoval*
\begin{proof}
Similar to the previous proof, this proof does not explicitly use the large goods of the instance $\cal I^{mod}$. When we use the value $\mu$, we mean the value $\frac{v(G \setminus L)}{\Delta}$, where $L$ is the set of large goods in the instance $\cal I$. In the allocation $W^{mod}$, we use $S^{mod}_i$ to denote the set $W^{mod}_i \setminus L'$ for each agent $i$.  

Like all the previous proofs, we partition the agents to lower bound the log Nash welfare. Let $N_{L'}$ be the set of agents who receive a good in $L'$ in the allocation $W^{mod}$.
Each of the large goods in the new set $L'$ is more valuable than all the small goods of the instance. Therefore, the set $N_{L'} = \{1, 2, \dots, |L'|\}$. Let $N_S$ denote the remaining agents. Of the agents in $N_S$, let $N_1$ denote the set of agents $i$ such that $v(W^{mod}_i) \ge (1-\gamma)\mu$. Let $N_2$ denote the remaining agents $i$ where $v(W^{mod}_i) < (1-\gamma)\mu$. 

If $N_S$ is empty (or equivalently $|L'| = \Delta$), then the theorem trivially follows because
\begin{align*}
\logNSW(W^{mod}) \ge \frac1\Delta \left [\sum_{g \in L'}\ln v(g) \right ],
\end{align*}
which proves the theorem for the special case where $\Delta = |L'|$. From here on, we assume $|L'| < \Delta$; that is, $N_S$ is non-empty.

If all agents $i$ in $N_S$ satisfy $v(W^{mod}_i) \ge 2\mu$, then the log Nash welfare of $W^{mod}$ can be easily lower bounded as follows:
\begin{align*}
\logNSW(W^{mod}) &\ge \frac1\Delta \left [\sum_{g \in L'} \ln v(g) + (\Delta - |L'|) \ln(2\mu)\right ] \\
&\ge \frac1\Delta \left [\sum_{g \in L'} \ln v(g) + (\Delta - |L'|) \ln\left (\mu \left (1 + \frac{v(A)}{\Delta \mu} \right )\right )\right ] \\
&\ge \frac1\Delta \left [\sum_{g \in L'} \ln v(g) + (\Delta - |L'|) \ln\mu \right ] + \frac{(\Delta - |L'|)v(A)}{2\Delta^2\mu}.
\end{align*} 
In the final inequality, we use Fact \ref{fact:identity-inequality}. This proves the theorem. So from here on, we assume $v(W^{mod}_{\Delta}) < 2\mu$. That is, the worst bundle of the allocation has utility at most $2\mu$. We make the following observations.

\begin{obs}\label{obs:value-upper-bound}
For each $i \in N_S$, $v(W^{mod}_{i}) \le 5\mu$. 
\end{obs}
\begin{proof}
This follows from Fact \ref{fact:round-robin-ef1}. 
\begin{equation*}
v(W^{mod}_i) \le v(W^{mod}_{\Delta}) + \max_{g \in W^{mod}_i} v(g) \le 2\mu + \frac{\Delta\mu}{\Delta - \ell} \le 5\mu. 
\end{equation*}
The final inequality holds due to \ref{prop:a}.
\end{proof}

\begin{obs}\label{obs:round-robin-better-off}
For each $i \in [\Delta - |L'|]$, $v(S^{mod}_{|L'| + i}) \ge v(S_{\ell + i})$ if $\ell + i \le \Delta$ and $v(S^{mod}_{|L'| + i}) \ge v(S_{\ell + i - \Delta})$ otherwise. 
\end{obs}
\begin{proof}
The bundle $S^{mod}_{|L'| + i}$ can be constructed using the round robin algorithm over the goods $(G \setminus L) \cup A$ with picking sequence $(|L'| + 1, \dots, \Delta, 1, \dots, |L'|)$. The small bundles $(S_1, \dots, S_{\Delta})$ can be constructed using the round robin algorithm over the goods $G \setminus L$ with picking sequence $(\ell +  1, \dots, \Delta, 1, \dots, \ell)$. Since $(G \setminus L) \cup A$ is a superset of $G \setminus L$, $v(S^{mod}_{|L'| + i}) \ge v(S_j)$ where $j$ is the $i$-th agent in the picking sequence $(\ell +  1, \dots, \Delta, 1, \dots, \ell)$. The observation follows from noting that $j$ is $\ell + i$ if $\ell + i \le \Delta$ and $\ell + i - \Delta$ otherwise.
\end{proof}

Since $\cal I$ is well-behaved, the above observation implies that the set $N_2$ has size at most $\gamma \Delta$ (\ref{prop:b}), and each agent $i$ in $N_2$ satisfies $v(W^{mod}_i) \ge 10^{-4}\mu$ (\ref{prop:e}).
For each agent $i \in N_1$, $v(W^{mod}_i) \in [(1-\gamma)\mu, 5\mu]$ (using Observation \ref{obs:value-upper-bound}). We can therefore define a value $\alpha_i \in [0, 1]$ for each $i \in N_1$ that satisfies $v(W^{mod}_i) = \alpha_i (5\mu) + (1-\alpha_i)\mu(1-\gamma)$. Define $\alpha = \sum_{i \in N_1} \alpha_i$. 

We can lower bound the log Nash welfare of $W^{mod}$ using the concavity of the $\log$ function.
\begin{align*}
\logNSW(W^{mod}) &\ge \frac{1}{\Delta} \left [ \sum_{g \in L'} \ln v(g) + \sum_{i \in N_1} \ln v(W^{mod}_i) + \sum_{i \in N_2} \ln v(W^{mod}_i)\right ] \\
&\ge \frac{1}{\Delta} \left [ \sum_{g \in L'} \ln v(g) + \alpha\ln(5\mu) + (|N_1| - \alpha)\ln((1-\gamma)\mu) + |N_2| \ln(10^{-4} \mu)\right ] \\
&\ge \frac{1}{\Delta} \left [ \sum_{g \in L'} \ln v(g) + (\Delta - |L'|)\ln(\mu) \right ] + \frac{\alpha}{\Delta}\ln(5) + \frac{|N_1|}{\Delta}\ln(1-\gamma) + \frac{|N_2|}{\Delta} \ln(10^{-4})
\end{align*}

We can simplify this further using $|N_1| \le \Delta$ and $|N_2| \le \gamma \Delta$. 
\begin{align}
\logNSW(W^{mod}) &\ge \frac{1}{\Delta} \left [ \sum_{g \in L'} \ln v(g) + (\Delta - |L'|)\ln(\mu) \right ] + \frac{\alpha}{\Delta} - 12\gamma. \label{eq:small-good-addition-large-good-removal-0}
\end{align}

To complete this proof, we need to lower bound $\alpha$. By the definition of the round robin algorithm, $\{S^{mod}_{|L'|+1}, \dots, S^{mod}_{\Delta}\}$ are the $\Delta - |L'|$ highest valued small bundles in the allocation. Therefore, 

\begin{align}
\sum_{i = |L'| + 1}^{\Delta} v(W^{mod}_{i}) \ge \frac{\Delta - |L'|}{\Delta} \sum_{i = 1}^{\Delta} v(S^{mod}_i) = (\Delta - |L'|) \mu + \frac{(\Delta - |L'|)v(A)}{\Delta}. \label{eq:small-good-addition-large-good-removal-1} 
\end{align}

We upper bound the left hand side by partitioning the agents into $N_1$ and $N_2$. 
\begin{align*}
\sum_{i = |L'| + 1}^{\Delta} v(W^{mod}_{i}) &= \sum_{i \in N_1} v(W^{mod}_i) + \sum_{i \in N_2}v(W^{mod}_i) \\
&\le \alpha (5\mu) + (|N_1| - \alpha)\mu(1-\gamma) + |N_2|\mu(1-\gamma) \\
&\le 5\alpha\mu + (\Delta - |L'|)\mu.
\end{align*}

Plugging this back into \eqref{eq:small-good-addition-large-good-removal-1}, we get:
\begin{align*}
5\alpha\mu + (\Delta - |L'|)\mu \ge (\Delta - |L'|) \mu + \frac{(\Delta - |L'|)v(A)}{\Delta}.
\end{align*}

Re-arranging this inequality gives us $\alpha \ge \frac{(\Delta - |L'|)v(A)}{5\Delta \mu}$. Plugging this back into \eqref{eq:small-good-addition-large-good-removal-0}, we conclude that
\begin{align*}
\logNSW(W^{mod}) &\ge \frac{1}{\Delta} \left [ \sum_{g \in L'} \ln v(g) + (\Delta - |L'|)\ln(\mu) \right ] + \frac{(\Delta - |L'|)v(A)}{5\Delta^2 \mu} - 12\gamma, 
\end{align*}
which proves the theorem.
\end{proof}

\thmgeneraladdremoval*
\begin{proof}
Similar to the previous proof, this proof does not explicitly use the large goods of the instance $\cal I^{mod}$. When we use the value $\mu$, we mean the value $\frac{v(G \setminus L)}{\Delta}$, where $L$ is the set of large goods in the instance $\cal I$.  Like all previous proofs, this one also partitions the agents and uses the concavity of the log function to lower bound the log Nash welfare of $W^{mod}$. However, because small goods are both added and removed, we analyze this in two phases: in the first phase, we analyze the effect of removing small goods, and in the second phase, we analyze the effect of adding small goods. 

Formally, let $\cal I^{int}$ be an intermediate instance that results from replacing $L$ with $L'$ and removing the small goods $R$ from the instance $\cal I$. Let $W^{int}$ denote the round robin allocation of this instance. We also use $S^{int}_i$ to denote for each agent $i$, the set $W^{int}_i \setminus L'$. Define $N_{L'}$ to be the set of agents who receive a good from the set $L'$ in $W^{int}$. Since $L'$ are the highest value goods of the instance, $N_{L'} = \{1, \dots, |L'|\}$. Let $N_S$ denote the remaining agents. We first show that if $N_S$ is non-empty,
\begin{align}
\frac1{\Delta}\sum_{i \in N_S} \ln v(W^{int}_i) \ge \frac{(\Delta - |L'|)}{\Delta} \ln \mu - \frac{2(\Delta - |L'|)v(R)}{\Delta^2 \mu} - 12\gamma. \label{eq:general-add-removal-phase-1}
\end{align}

The above inequality is true when each agent $i \in N_S$ satisfies $v(W^{int}_i) \ge (1 + \gamma)\mu$. Therefore, we assume $v(W^{int}_{\Delta})$ is strictly less than $(1 + \gamma)\mu$. Using this, we partition the agents in $N_S$ further. We define $N_1$ to be the set of agents $i$ that satisfy $v(W^{int}_i) > (1 + \gamma)\mu$. We define $N_2$ to be the remaining agents. Note that each agent $i \in N_2$ satisfies $v(W^{int}_i) \in [v(W^{int}_{\Delta}), (1 + \gamma)\mu]$. We therefore define $\alpha_i \in [0, 1]$ for each $i \in N_2$ where $v(W^{int}_i) = (1 - \alpha_i)((1 + \gamma) \mu) + \alpha_i v(W^{int}_{\Delta})$. We define $\alpha = \sum_{i \in N_2} \alpha_i$.

We can lower bound $\frac1{\Delta}\sum_{i \in N_S} \ln v(W^{int}_i)$ as follows: 

\begin{align}
\frac1{\Delta}\sum_{i \in N_S} \ln v(W^{int}_i)&\ge \frac{1}{\Delta} \left [|N_1|\ln(\mu) + \sum_{i \in N_2} \ln v(W^{int}_i) \right ]\notag \\
&\ge \frac{1}{\Delta} \left [|N_1|\ln(\mu) + (|N_2| - \alpha)\ln \mu + \alpha \ln v(W^{int}_{\Delta})\right ] \notag \\
&=  \frac{1}{\Delta} \left [(\Delta - |L'|)\ln(\mu)\right ] + \frac{\alpha}{\Delta}\ln \left (\frac{v(W^{int}_{\Delta})}{\mu} \right ) \label{eq:general-add-removal-1}
\end{align}

The key part left to analyze is the term $\frac{\alpha}{\Delta}\ln \left (\frac{v(W^{int}_{\Delta})}{\mu} \right )$. This is done using the following observations. 

\begin{obs}\label{obs:n1-upper-bound}
$|N_1| \le \gamma\Delta$, and for any agent $i \in N_S$, $v(W^{int}_i) \le 5\mu$. 
\end{obs}
\begin{proof}
The first part follows from \ref{prop:b} and the fact that small bundles only decrease in value after removing the set of goods $R$. The second part follows from the fact that $v(W^{int}_i) \le v(W^{int}_{\Delta}) + \frac{\Delta\mu}{\Delta - \ell}$ by Fact \ref{fact:round-robin-ef1}. \ref{prop:a} implies that $\frac{\Delta\mu}{\Delta - \ell} \le 3\mu$ and we assume that $v(W^{int}_{\Delta}) \le (1+\gamma)\mu$. Therefore, $v(W^{int}_i) \le 5\mu$.
\end{proof}

\begin{obs}\label{obs:general-add-removal-w-delta-lower-bound}
$\frac{\alpha}{\Delta}\ln \left ( \frac{v(W^{int}_{\Delta})}{\mu} \right ) \ge - \frac{2(\Delta - |L'|)v(R)}{\Delta^2 \mu} - 12\gamma$.
\end{obs}
\begin{proof}
We lower bound $v(W^{int}_{\Delta})$ by using the fact that $\Delta \mu - v(R) = \sum_{i \in [\Delta]}v(S^{int}_i)$. This argument mimics Observation \ref{obs:w-mod-delta-lower-bound}.
\begin{align*}
\Delta \mu - v(R) &= \sum_{i \in [\Delta]}v(S^{int}_i) \\
&\le \sum_{i \in N_{L'}} v(S^{int}_i) + \sum_{i \in N_{1}} v(W^{int}_i) + \sum_{i \in N_2} v(W^{int}_i) \\
&\le |L'|v(W^{int}_{\Delta}) + |N_1|\left (5\mu \right ) + (|N_2| - \alpha)(1 + \gamma)\mu + \alpha v(W^{int}_{\Delta})
\end{align*}

In the final inequality, we upper bound $v(S^{int}_i)$ using $v(W^{int}_{\Delta})$ for each $i \in N_{L'}$. We also use Observation \ref{obs:n1-upper-bound} to upper bound $v(W^{int}_i)$ for each $i \in N_1$. This can be simplified further by utilitizing $|N_1| \le \gamma\Delta$. 
\begin{align*}
\Delta \mu - v(R) &\le |{L'}|v(W^{int}_{\Delta}) + \gamma\Delta\left (5\mu \right ) + (|N_2| - \alpha)(1 + \gamma)\mu + \alpha v(W^{int}_{\Delta}) \\
&\le (|{L'}| + \alpha)v(W^{int}_{\Delta}) + (|N_2|-\alpha + 6\gamma\Delta)\mu.
\end{align*}

Re-arranging terms and using $\Delta \ge |{L'}| + |N_2|$, we get
\begin{align*}
\frac{v(W^{int}_{\Delta})}{\mu} \ge \left (1 - \frac{v(R) + 6\gamma\Delta\mu}{(|{L'}| + \alpha)\mu} \right )
\end{align*}

From the statement of the theorem, $v(R) \le \frac{|L'|\mu}{100}$ and $\gamma \Delta \le \frac{|L'|}{100}$. Therefore, using Fact \ref{fact:identity-inequality}, we can simplify the above expression to:
\begin{align*}
\frac{v(W^{int}_{\Delta})}{\mu} \ge e^{\frac{-2(v(R) + 6\gamma\Delta\mu)}{(|L'| + \alpha)\mu}}.
\end{align*}

Using this, the observation follows:
\begin{align*}
\frac{\alpha}{\Delta}\ln \left ( \frac{v(W^{int}_{\Delta})}{\mu} \right ) &\ge \frac{\alpha}{\Delta} \left ( \frac{-2(v(R) + 6\gamma\Delta\mu)}{(|L'| + \alpha)\mu} \right ) \\
&\ge -\frac{2(\Delta - |L'|)v(R)}{\Delta^2 \mu} - 12\gamma.
\end{align*}
The final inequality uses the fact that $\alpha \le \Delta - |L'|$. 
\end{proof}

Plugging in Observation \ref{obs:general-add-removal-w-delta-lower-bound} into \eqref{eq:general-add-removal-1} proves \eqref{eq:general-add-removal-phase-1}. 

We now move from $\cal I^{int}$ to $\cal I^{mod}$ by adding the set of small goods $A$ to the instance. For each agent $i$, we define $S^{mod}_i$ as the set of goods $W^{mod}_i \setminus L'$. By adding the set of small goods $A$, the value of each agent's bundle can only increase. Exploiting this, we define the value $\beta_i \ge 0$ for each agent $i \in N_S$ as the value $v(W^{mod}_i) - v(W^{int}_i)$. 

First, consider the case when $|L'| = \Delta$; that is, $N_S$ is empty. Then the log Nash welfare of $W^{mod}$ can be lower bounded as
\begin{align*}
\logNSW(W^{mod}) \ge \frac1\Delta \left [\sum_{g \in L'} \ln v(g) \right ].
\end{align*}
This proves the theorem for the case where $|L'| = \Delta$. We therefore assume that $|L'| < \Delta$.

If $v(W^{mod}_{\Delta}) \ge 2\mu$, then each agent who does not receive a good in $L'$ receives a utility of at least $2\mu$. This allows us to lower bound the Nash welfare as
\begin{align*}
\logNSW(W^{mod}) &\ge \frac1{\Delta} \left [\sum_{g \in L'}\ln v(g) + (\Delta - |L'|)\ln(2\mu) \right ] \\
&\ge \frac1{\Delta} \left [\sum_{g \in L'}\ln v(g) + (\Delta - |L'|)\ln(\mu) \right ] + \frac{\Delta - |L'|}{2\Delta}.
\end{align*}

Note that $v(A) \le \Delta \mu$, which implies $1 \ge \frac{v(A) - v(R)}{\Delta \mu}$. Therefore, the above inequality can be further lower bounded as
\begin{align*}
\logNSW(W^{mod}) &\ge \frac1{\Delta} \left [\sum_{g \in L'}\ln v(g) + (\Delta - |L'|)\ln(\mu) \right ] + \frac{(\Delta - |L'|)(v(A)-v(R))}{4\Delta^2 \mu},
\end{align*}
which proves the Theorem. We therefore assume $v(W^{mod}_{\Delta}) < 2\mu$. Using this assumption, we make a simple observation.

\begin{obs}\label{obs:beta-upper-bound}
If $v(W^{mod}_{\Delta}) < 2\mu$, then for any $i \in N_S$, $v(W^{mod}_i) \le 5\mu$. Crucially, this implies that $\beta_i \le 5\mu$ for any $i \in N_S$.
\end{obs}
\begin{proof}
Via Fact \ref{fact:round-robin-ef1}, $v(W^{mod}_i) \le v(W^{mod}_{\Delta}) + \max_{g \in W^{mod}_i} v(g)$. The first term is at most $2\mu$ and the second term is at most $\frac{\Delta \mu}{\Delta - \ell} \le 3\mu$. 
\end{proof}

We can now lower bound the log Nash welfare of $W^{mod}$. 
\begin{align*}
\logNSW(W^{mod}) &\ge \frac1{\Delta} \left [ \sum_{g \in L'} \ln v(g) + \sum_{i \in N_S} \ln (v(W^{mod}_i))\right ] \\
&= \frac1{\Delta} \left [ \sum_{g \in L'} \ln v(g) + \sum_{i \in N_S} \ln (v(W^{int}_i) + \beta_i)\right ] \\
&\ge \frac1{\Delta} \left [ \sum_{g \in L'} \ln v(g) + \sum_{i \in N_S} \ln \left (v(W^{int}_i) \left (1+ \frac{\beta_i}{5\mu}\right ) \right )\right ].
\end{align*}

In the last inequality, we use $v(W^{int}_i)\le 5\mu$ (Observation \ref{obs:n1-upper-bound}). Since $\beta_i \le 5\mu$, we can lower bound $\ln \left (1 + \frac{\beta_i}{5\mu} \right )$ using $\frac{\beta_i}{10\mu}$ (Fact \ref{fact:identity-inequality}). This simplifies the above expression to:
\begin{align}
\logNSW(W^{mod}) \ge \frac1{\Delta} \left [ \sum_{g \in L'} \ln v(g) + \sum_{i \in N_S} \ln \left (v(W^{int}_i) \right )\right ] + \frac{\sum_{i \in N_S} \beta_i}{10\Delta\mu}. \label{eq:general-add-removal-phase-2}
\end{align}

\eqref{eq:general-add-removal-phase-1} lowerbounds $\sum_{i \in N_S} \ln \left (v(W^{int}_i) \right )$. The following observation lower bounds $\sum_{i \in N_S} \beta_i$.
\begin{obs}\label{obs:beta-lower-bound}
$\sum_{i \in N_S} \beta_i \ge \frac{(\Delta - |L'|)(v(A) - v(R))}{\Delta} - 6\gamma\Delta\mu$.
\end{obs}
\begin{proof}
By the definition of the round robin algorithm, $\{S^{mod}_{|L'|+1}, \dots, S^{mod}_{\Delta}\}$ are the $(\Delta - |L'|)$ highest valued bundles of $\{S^{mod}_1, \dots, S^{mod}_\Delta\}$. Therefore, 
\begin{align*}
\sum_{i \in N_S}v(W^{mod}_i) \ge \frac{\Delta - |L'|}{\Delta} \sum_{i \in [\Delta]} v(S^{mod}_i) \ge \frac{\Delta - |L'|}{\Delta} \left [\Delta \mu + v(A) - v(R) \right ].
\end{align*}

We can upper bound the left hand side as follows:
\begin{align*}
\sum_{i \in N_S} v(W^{mod}_i) &= \sum_{i \in N_S}v(W^{int}_i) + \sum_{i \in N_S} \beta_i \\
&= \sum_{i \in N_1} v(W^{int}_i) + \sum_{i \in N_2} v(W^{int}_i) + \sum_{i \in N_S}\beta_i \\
&\le 5|N_1|\mu + |N_2|\mu(1 + \gamma) + \sum_{i \in N_S}\beta_i \\
&\le 5\gamma\Delta\mu + (\Delta - |L'|)\mu(1 + \gamma) + \sum_{i \in N_S}\beta_i.
\end{align*}

Here, we use both the grouping of the agents in $N_S$ and the properties of $N_1$ given by Observation \ref{obs:n1-upper-bound}. This implies
\begin{align*}
5\gamma\Delta\mu + (\Delta - |L'|)\mu(1 + \gamma) + \sum_{i \in N_S}\beta_i \ge \frac{(\Delta - |L'|)(\Delta\mu + v(A) - v(R))}{\Delta}.
\end{align*}

Re-arranging terms gives us 
\begin{align*}
\sum_{i \in N_S}\beta_i \ge \frac{(\Delta - |L'|)(v(A) - v(R))}{\Delta} - 6\gamma\Delta\mu.
\end{align*}
This proves the Observation.
\end{proof}

Plugging in Observation \ref{obs:beta-lower-bound} and \eqref{eq:general-add-removal-phase-1} into \eqref{eq:general-add-removal-phase-2}, we prove the Theorem.
\end{proof}

\thmnonwellbehavedremoval*
\begin{proof}
Let $W$ be the round robin allocation in the instance $\cal I$. We show that 
\begin{align*}
\logNSW(W^{mod}) \ge \logNSW(W) - 2.
\end{align*}

Combining this with the fact that $\frac{\NSW(\cal I)}{\NSW(W)} \le e^{1/e}$ (Lemma \ref{lem:barman-eonebye}), this proves the theorem. 

Let the goods of $\cal I^{mod}$ be the set $\{g'_1, g'_2, \dots, g'_{m'}\}$. Let the classes of goods $G_1, \dots, G_{k'}$ be ordered such that each good in $G_1$ has value weakly greater than each good in $G_2$, and so on. 
To compare the allocation $W^{mod}$ and $W$, we compare the allocated bundles for each agent. Fix some agent $i \in [\Delta]$. 
We have $W_i = \{g_i, g_{\Delta + i}, g_{2\Delta + i}, \dots\}$ for each agent $i$ by our notation that $G = \{g_1, \dots, g_m\}$ is the set of goods of the original instance $\cal I$. Similarly, assuming $v(g'_1) \ge v(g'_2) \ge \dots \ge v(g'_{m'})$, we have $W^{mod}_i = \{g'_{i}, g'_{\Delta + i}, g'_{2\Delta + i} \dots\}$ for each agent $i$ by the definition of the round robin allocation (assuming ties are broken in favor of lower index). We prove the following claims to relate the two bundles.

\begin{obs}\label{obs:iprime-equals-i}
$v(g'_{i}) \ge v(g_i)$.
\end{obs}
\begin{proof}
Let $\widetilde{G}$ be the set of goods which have the same value as $g_{\Delta + 1}$. Let $G'_{1}, \dots, G'_{t}$ be the classes whose goods are present in $\widetilde{G}$. The classes $G'_{1}, \dots, G'_{t}$ are the only classes which can have goods in $G^{\Delta}$ and outside $G^{\Delta}$. All the classes whose goods have value greater than $v(g_{\Delta + 1})$ are entirely contained within $G^{\Delta}$ and therefore, we do not remove any goods from these classes. 

This implies if $v(g_i) > v(g_{\Delta + 1})$ we have $v(g'_i) \ge v(g_i)$. If $g_i \in \widetilde{G}$, then $\widetilde{G}$ intersects with $G^{\Delta}$, and the number of goods we remove from $\widetilde{G}$ is at most $\frac12\left (|\widetilde{G}| - |\widetilde{G} \cap G^{\Delta}| \right )$. This means that after removal, $g'_{\Delta}$ is a good in $\widetilde{G}$ since there are at least $|\widetilde{G} \cap G^{\Delta}|$ goods in the set $\widetilde{G}$ which are not removed. This implies $v(g'_i) \ge v(g'_{\Delta}) = v(g_i)$.  
\end{proof}
\begin{obs}\label{obs:tdelta-bound}
For each integer $t\ge 0$, $v(g'_{t\Delta + i}) \ge v(g_{(2t+1)\Delta + i})$.
\end{obs}
\begin{proof}
Let $G_{t^*}$ be the class that $g_{(2t+1)\Delta + i}$ belongs to. This implies there are at least $(2t+1)\Delta + i$ goods in the first $t^*$ classes. For each of these classes $G_{t'}$, we remove at most $\frac{|G_{t'}|}{2}$ goods. Therefore, the number of goods from the first $t^*$ classes that are {\em not} removed is at least $\frac12\left((2t + 1)\Delta + i \right ) \ge t\Delta + i$. This implies the $g'_{t\Delta + i}$ belongs to one of the first $t^*$ classes. Therefore, $v(g'_{t\Delta + i}) \ge v(g_{(2t+1)\Delta + i})$.
\end{proof}

We can use Observations \ref{obs:iprime-equals-i} and \ref{obs:tdelta-bound} to upper bound $v(W_i)$.
Using the observations, we can lower bound $v(W^{mod}_i)$ as follows:
\begin{align*}
v(W_i) &= v(g_i) + v(g_{\Delta + i}) + v(g_{2\Delta + i}) + \dots \\
&\le v(g'_i) + v(g_{\Delta + i}) + v(g_{2\Delta + i}) + \dots \\
&= v(g'_i) + \left ( v(g_{\Delta + i}) + v(g_{2\Delta + i}) \right ) + \left ( v(g_{3\Delta + i}) + v(g_{4\Delta + i}) \right ) + \left ( v(g_{5\Delta + i}) + v(g_{6\Delta + i}) \right ) + \dots \\
&\le v(g'_i) + \left (v(g'_i) + v(g'_i) \right ) + \left (v(g'_{\Delta + i}) + v(g'_{\Delta + i}) \right ) +\left (v(g'_{2\Delta + i}) + v(g'_{2\Delta + i}) \right ) + \dots \\
&\le 3v(g'_i) + 3v(g'_{\Delta + i}) + 3v(g'_{2\Delta + i}) + \dots \\
&\le 3v(W^{mod}_i)
\end{align*}
In the first inequality, we used Observation \ref{obs:iprime-equals-i} and in the second inequality, we used Observation \ref{obs:tdelta-bound}.

We have shown that for each agent $i$, $v(W_i)\le 3v(W^{mod}_i)$. This immediately gives us
\begin{align*}
\logNSW(W^{mod}) = \frac1{\Delta} \sum_{i \in [\Delta]} \ln(v(W^{mod}_i)) \ge \frac1{\Delta} \sum_{i \in [\Delta]} \left (\ln v(W_i) - \ln(3) \right ) \ge \frac1{\Delta} \sum_{i \in [\Delta]}  \ln v(W_i) - 2.
\end{align*}
This completes the proof.
\end{proof}

\section{Missing Proofs from Section \ref{sec:st-rounding}}\label{apdx:st-rounding}

\lemfractionalnashlowerbound*
\begin{proof}
Let $\cal I^{i, \x}(\Delta)$ be the artificial instance of agent $i$ for some $\Delta$ that is feasible with respect to both $\x$ and $y$. Let $W$ be the round robin allocation of this instance. Theorem \ref{thm:round-robin-nash-lower-bound} proves the following:
\begin{align*}
\logNSW(W) \ge  \frac{1}{\Delta}\left [ \sum_{g \in L} \ln{v_i(g)} + (\Delta - \ell) \ln\left ( \frac{\Delta\mu}{\Delta - \ell} \right ) \right ] - \frac{1}{e}.
\end{align*}

Combining this with Lemma \ref{lem:w-worst}, we get:
\begin{align}
\ST(\x, i) \ge \frac{1}{\Delta}\left [ \sum_{g \in L} \ln{v_i(g)} + (\Delta - \ell) \ln\left ( \frac{\Delta\mu}{\Delta - \ell} \right ) \right ] - \frac{1}{e}. \label{eq:fractional-nash-lower-bound-1}
\end{align}

We use Lemma \ref{lem:large-good-connection} which states that $\ell = \Delta \ell_i$, $\mu^i = \mu$ and $L$ consists of $\Delta x_{ig}$ copies of each $g \in L_i$. This simplifies \eqref{eq:fractional-nash-lower-bound-1} to
\begin{align*}
\ST(\x, i) \ge \sum_{g \in L_i} x_{ig}\ln{v_i(g)} + (1 - \ell_i) \ln\left ( \frac{\mu^i}{1 - \ell_i} \right )  - \frac{1}{e}.
\end{align*}

This proves the first inequality. To prove the second inequality, we use Lemma \ref{lem:nash-upper-bound} which states that
\begin{align*}
\logNSW(\cal I^{i, \x}(\Delta)) \le \frac{1}{\Delta}\left [ \sum_{g \in L} \ln{v_i(g)} + (\Delta - \ell) \ln\left ( \frac{\Delta\mu}{\Delta - \ell} \right ) \right ]
\end{align*}

Plugging this into \eqref{eq:fractional-nash-lower-bound-1} gives us
\begin{align*}
\ST(\x, i) \ge \logNSW(\cal I^{i, \x}(\Delta)) - \frac{1}{e}.
\end{align*}

Finally we use Lemma \ref{lem:y-substitute} to lower bound $\logNSW(\cal I^{i, \x}(\Delta))$ using $\sum_{S \subseteq G} y_{i, S}\ln v_i(S)$. This proves the lemma. 
\end{proof}

\lemfractionalpropertya*
\begin{proof}
Let $\cal I^{i, \x}(\Delta)$ be the artificial instance for agent $i$ for some $\Delta$ feasible with respect to $\x$. Let $\ell$ be the number of large goods of the instance $\cal I^{i, \x}(\Delta)$. Since $i$ is well-behaved in $\x$, $\cal I^{i, \x}(\Delta)$ must be well-behaved, which implies $\frac{\ell}{\Delta} \in \left [1 - \frac1e - \gamma, 1 - \frac1e+\gamma \right ]$ (\ref{prop:a}). The lemma follows by noting that $\ell_i = \frac{\ell}{\Delta}$ (Lemma \ref{lem:large-good-connection}). 
\end{proof}

\lemfractionalnashupperbound*
\begin{proof}
Let $\cal I^{i, \x}(\Delta)$ be the artificial instance of agent $i$ for some $\Delta$ that is feasible with respect to $\x$ and $y$. Lemma \ref{lem:nash-upper-bound} shows the following
\begin{align*}
\logNSW(\cal I^{i, \x}(\Delta)) \le \frac{1}{\Delta} \left [\sum_{g \in L}\ln v_i(g) + (\Delta - \ell)\ln \mu \right ] + \frac1e.
\end{align*}

We use Lemma \ref{lem:large-good-connection} which states that $\ell = \Delta \ell_i$, $\mu^i = \mu$ and $L$ consists of $\Delta x_{ig}$ copies of each $g \in L_i$. This simplifies the above inequality to
\begin{align*}
\logNSW(\cal I^{i, \x}(\Delta)) \le \sum_{g \in L_i} x_{ig}\ln{v_i(g)} + (1 - \ell_i) \ln \mu^i   + \frac{1}{e}. 
\end{align*}

Finally, we use Lemma \ref{lem:y-substitute} to lower bound $\logNSW(\cal I^{i, \x}(\Delta))$ using $\sum_{S \subseteq G} y_{i, S}\ln v_i(S)$, and prove the lemma.
\end{proof}

\section{Missing Proofs from Section \ref{sec:main-algo}}\label{apdx:main-algo}

\lemdependentroundingexpectation*
\begin{proof}
We show that $\E[z_{ig, k+1}] = \E[z_{ig, k}]$ for all $k \ge 1$. This implies that $\E[z_{ig, \infty}] = \E[z_{ig, 1}]$ and since $\E[z_{ig, 1}] = x_{ig}$, the lemma is proved. 

We consider two events at iteration $k$. 

\textbf{Event A:} The edge $(i, g)$ is not in the path $P$ or the cycle $C$ chosen by the algorithm. In this case, $\E[z_{ig, k+1}] = \E[z_{ig, k}]$ since it remains unchanged during iteration $k$. 

\textbf{Event B:} The edge $(i, g)$ is in the path $P$ or cycle $C$ chosen by the algorithm. In this case, assume without loss of generality that $(i, g) \in M_1$. The proof for the other case is similar. We can compute $\E[z_{ig, k+1}]$ using the definition of our update rule:

\begin{align*}
\E[z_{ig, k+1}] = \E\left [ \frac{\beta^*}{\beta^* + \alpha^*}(z_{ig, k} + \alpha^*) + \frac{\alpha^*}{\beta^* + \alpha^*}(z_{ig, k} - \beta^*) \right ] = \E[z_{ig, k}].
\end{align*}

Since the two events $A$ and $B$ cover all possible cases, this proves the lemma.
\end{proof}

\lemdependentroundingcorrelation*
\begin{proof}
This proof is similar to the previous lemma. We show that for all $k \ge 1$,
\begin{align}
\E \left [\prod_{i \in S} \ell^{\z}_{i, k+1} \right ] \le \E \left [\prod_{i \in S} \ell^{\z}_{i, k} \right ]. \label{eq:negative-correlation} 
\end{align}

This implies that $\E \left [\prod_{i \in S} \ell^{\z}_{i, \infty} \right ] \le \E \left [\prod_{i \in S} \ell^{\z}_{i, 1} \right ]$ and the lemma is proved due to
\begin{align*}
\E \left [\prod_{i \in S} \ell^{\z}_{i, 1} \right ] = \prod_{i \in S} \ell^{\x}_{i}\le \left (1 - \frac1e + \gamma \right )^{|S|}.
\end{align*}
The final inequality holds due to Lemma \ref{lem:fractional-property-a}. 

Let us consider iteration $k$ (for some $k$). In this iteration, if a cycle $C$ is chosen by the algorithm, $\ell^{\z}_{i, k+1} = \ell^{\z}_{i, k}$ for all agents $i$ since no agent sees their large good amount modified. Therefore, \eqref{eq:negative-correlation} holds. If a maximal path $P$ is chosen by the algorithm, only agents at the endpoints of the path see a change in their $\ell^{\z}_i$ values. 

Let $i_1, i_2$ be the two agents at the endpoint of the path $P$ chosen by the algorithm at iteration $k$. Note that for all agents $i \in S \setminus \{i_1, i_2\}$, we have $\ell^{\z}_{i, k+1} = \ell^{\z}_{i, k}$. Therefore, if neither $i_1$ nor $i_2$ are in $S$, then \eqref{eq:negative-correlation} holds.
If both $i_1$ and $i_2$ are in $S$, \eqref{eq:negative-correlation} holds if $\E[\ell^{\z}_{i_1, k+1} \ell^{\z}_{i_2, k+1}] \le \E[\ell^{\z}_{i_1, k} \ell^{\z}_{i_2, k}]$. This is true because

\begin{align*}
\E[\ell^{\z}_{i_1, k+1} \ell^{\z}_{i_2, k+1}] &= \E \left [\frac{\beta^*}{\beta^* + \alpha^*}(\ell^{\z}_{i_1, k} + \alpha^*)(\ell^{\z}_{i_2, k} - \alpha^*) + \frac{\alpha^*}{\beta^* + \alpha^*}(\ell^{\z}_{i_1, k} - \beta^*)(\ell^{\z}_{i_2, k} + \beta^*)\right ] \\
&\le \E \left [\ell^{\z}_{i_1, k}\ell^{\z}_{i_2, k} - \alpha^* \beta^* \right ] \\
&\le \E \left [\ell^{\z}_{i_1, k}\ell^{\z}_{i_2, k} \right ].
\end{align*}
The above proof assumes that the edges of $P$ incident on $i_1$ and $i_2$ fall into different matchings. Note that this holds because the graph is bipartite; therefore, if the two endpoints of a path are two agents, then the path must have even length. This implies the two edges at either end of the path must fall into different matchings. 

If one agent (say $i_1$) at the endpoint of the path $P$ is in the set $S$, then using the arguments of Lemma \ref{lem:dependent-rounding-expectation}, we can show that $\E[\ell^{\z}_{i_1, k+1}] = \E[\ell^{\z}_{i_1, k}]$. This proves \eqref{eq:negative-correlation} for this case. 

This covers all possible cases, and therefore the lemma holds.
\end{proof}

\end{document}